\documentclass{fundam}

\usepackage[utf8]{inputenc}
\usepackage{xspace}

\publyear{24}
\papernumber{2192}
\volume{192}
\issue{2}

\finalVersionForARXIV

\usepackage{hyperref}

\usepackage{ucs}
\usepackage{inputenc}
\usepackage[ruled,lined]{algorithm2e}

\usepackage{listings}
\usepackage[colorinlistoftodos,prependcaption]{todonotes}

\usepackage{textcomp}
\usepackage{cancel}

\usepackage{tikz}
\usetikzlibrary{patterns}

\usepackage{graphicx}

\usepackage{amsmath,amssymb,latexsym,fancyhdr}

\usepackage{etex}
\usetikzlibrary{calc} 						
\usetikzlibrary{shapes,arrows}				
\usetikzlibrary{shapes.multipart}			
\usetikzlibrary{positioning}
\usepackage[all,cmtip]{xy}					
\usepackage{mathtools}
\usepackage{url} 
\usepackage{xspace}
\usepackage{amsmath}
\usepackage{graphics}

\usepackage{wrapfig}

\def\defemb#1#2{\expandafter\def\csname #1\endcsname
{\relax\ifmmode #2\else\hbox{$#2$}\fi}}
\defemb{cA}{\mathcal{A}}
\defemb{cB}{\mathcal{B}}
\defemb{cC}{\mathcal{C}}
\defemb{cD}{\mathcal{D}}
\defemb{cE}{\mathcal{E}}
\defemb{cF}{\mathcal{F}}
\defemb{cG}{\mathcal{G}}
\defemb{cH}{\mathcal{H}}
\defemb{cI}{\mathcal{I}}
\defemb{cJ}{\mathcal{J}}
\defemb{cK}{\mathcal{K}}
\defemb{cL}{\mathcal{L}}
\defemb{cM}{\mathcal{M}}
\defemb{cN}{\mathcal{N}}
\defemb{cO}{\mathcal{O}}
\defemb{cP}{\mathcal{P}}
\defemb{cQ}{\mathcal{Q}}
\defemb{cR}{\mathcal{R}}
\defemb{cS}{\mathcal{S}}
\defemb{cT}{\mathcal{T}}
\defemb{cU}{\mathcal{U}}
\defemb{cV}{\mathcal{V}}
\defemb{cX}{\mathcal{X}}
\defemb{cZ}{\mathcal{Z}}

\newcommand{\ul}[1]{\underline{#1}}

\newcommand{\AProVE}{{\sf AProVE}}

\newcommand{\Haskell}{{\sf Haskell}}

\newcommand{\muterm}{\mbox{\sc mu-term\/}}
\newcommand{\CONFident}{\textsf{CONFident\/}}

\newcommand{\infChecker}{\mbox{\sf infChecker}}
\newcommand{\AGES}{\mbox{\sf AGES}}
\newcommand{\MaceFour}{\mbox{\sf Mace4}}
\newcommand{\ProverNine}{\mbox{\sf Prover9}}

\newcommand{\TTTTwo}{\textsf{T\kern-0.15em\raisebox{-0.55ex}T\kern-0.15emT\kern-0.15em\raisebox{-0.55ex}2}}

\newcommand{\FORT}{\textsf{Fort}}
\newcommand{\ConfCSR}{\textsf{ConfCSR\/}}
\newcommand{\Hakusan}{\textsf{Hakusan\/}}

\newcommand{\var}{{\cV}ar} 
\newcommand{\Var}{{\cV}ar} 

\newcommand{\bigfracn}[3]{
\begin{array}[b]{c}
\displaystyle #1 \\\hline\displaystyle #2
\end{array}
\hbox to 0pt{\raisebox{0.7em}{{\tiny (#3)}}}
}

\newcommand{\NPos}[1]{\ol{\mathcal{ P}os^{#1}}} 

\newcommand{\toppos}{{\Lambda}}
\newcommand{\ol}[1]{\overline{#1}}  
\newcommand{\Pos}{{\mathcal{P}os}}

\newcommand{\rootTerm}{\mathit{root}}
\newcommand{\grounding}[1]{{{#1}^\downarrow}}
\newcommand{\groundingTermOn}[2]{{{#1}^{\downarrow_{#2}}}}
\newcommand{\level}{{\textit{lv}}}

\newcommand{\Symbols}{{\cF}}

\newcommand{\Variables}{{\cX}}

\newcommand{\TermsOn}[2]{{\cT(#1,#2)}}

\newcommand{\Terms}{{\TermsOn{\Symbols}{\Variables}}}

\newcommand{\dom}{{\cD}om}

\newcommand{\csr}{\text{CSR}}

\newcommand{\srs}{\text{SRS}}
\newcommand{\trs}{\text{TRS}}

\newcommand{\cstrs}{\text{CS-TRS}}
\newcommand{\csctrs}{\text{CS-CTRS}}
\newcommand{\ctrs}{\text{CTRS}}
\newcommand{\dctrs}{\text{DCTRS}}

\newcommand{\gtrs}{\text{GTRS}}

\newcommand{\Rmaps}[1]{{M_{#1}}}
\newcommand{\CRmaps}[1]{{\mbox{\it CM}_{#1}}}
\newcommand{\CnvRmaps}[1]{{\mbox{\it CnvM}_{#1}}}

\newcommand{\muCan}{{\mu^{can}_\cR}}
\newcommand{\muCanOf}[1]{{\mu^{can}_{#1}}}

\newcommand{\muCnv}{{\mu^{cnv}_\cR}}
\newcommand{\muCnvOf}[1]{{\mu^{cnv}_{#1}}}

\newcommand{\muTop}{{\mu_\top}}
\newcommand{\muBot}{{\mu_\bot}}

\newcommand{\NVar}[1]{\mathcal{V}ar^{\cancel{#1\,}}} 

\newcommand{\SPredicates}{{\Pi}}

\newcommand{\RuleHornClause}{\text{HC}}

\newcommand{\RuleReflexivity}{\text{Rf}}
\newcommand{\RuleCompatibility}{\text{Co}} 
\newcommand{\RulePropagation}{\text{Pr}} 

\newcommand{\combinationOfModules}[2]{{\fS{Comb}(#1,#2)}}
\newcommand{\disjointComb}[2]{{\fS{DisjU}(#1,#2)}}
\newcommand{\constructorSharingComb}[2]{{\fS{CShC}(#1,#2)}}
\newcommand{\composableComb}[2]{{\fS{CompC}(#1,#2)}}

\newcommand{\RuleJO}{\text{J}} 
\newcommand{\RuleOR}{\text{O}} 
\newcommand{\RuleSE}{\text{SE}} 

\newcommand{\Uconf}{{\cU_{\text{conf}}}}

\newcommand{\fcond}{\gamma}
\newcommand{\fseq}{\mathsf{F}}

\newcommand{\CRproblem}[1]{{\mathit{CR}(#1)}}
\newcommand{\ECRproblem}[1]{{\mathit{ECR}(#1)}}
\newcommand{\GCRproblem}[1]{{\mathit{GCR}(#1)}}
\newcommand{\WCRproblem}[1]{{\mathit{WCR}(#1)}}
\newcommand{\SCRproblem}[1]{{\mathit{SCR}(#1)}}

\newcommand{\JOproblem}[1]{{\mathit{JO}(#1)}}
\newcommand{\SJOproblem}[1]{{\mathit{SJO}(#1)}}

\newcommand{\yesTree}{{\sf yes}}
\newcommand{\noTree}{{\sf no}}

\newcommand{\ProcCanCSR}{{\Proc_{CanCR}}}
\newcommand{\ProcCanJoinability}{{\Proc_{CanJ}}}
\newcommand{\ProcCnvJoinability}{{\Proc_{CnvJ}}}

\newcommand{\ProcHuetExtended}{{\Proc_{\it HE}}}

\newcommand{\ProcOrthogonality}{{\Proc_{\it Orth}}}
\newcommand{\ProcConfluence}{{\Proc_{\it CR}}}
\newcommand{\ProcLocalConfluence}{{\Proc_{\it WCR}}}
\newcommand{\ProcStrongConfluence}{{\Proc_{\it SCR}}}
\newcommand{\ProcKnuthBendix}{{\Proc_{\it KB}}}

\newcommand{\ProcJoinability}{{\Proc_{\it JO}}}

\newcommand{\ProcModularDecomp}{{\Proc_{\it MD}}}
\newcommand{\ProcSimplify}{{\Proc_{\it Simp}}}
\newcommand{\ProcInlining}{{\Proc_{\it Inl}}}
\newcommand{\ProcExtraVars}{{\Proc_{\it EVar}}}
\newcommand{\ProcTrU}{{\Proc_\cU}}

\newcommand{\ProcTrUconf}{{\Proc_\Uconf}}

\newcommand{\GLtheory}{\mathsf{Th}}

\newcommand{\rewtheoryOf}[1]{{\ol{#1}}}

\newcommand{\deductionInThOf}[2]{{#1\vdash #2}}

\newcommand{\lhsr}{\ell}
\newcommand{\rhsr}{r}

\newcommand{\elhsr}{\lambda} 
\newcommand{\erhsr}{\rho}

\newcommand{\IF}{\Leftarrow}
\newcommand{\IFhorn}{\Leftarrow}

\newcommand{\gencond}{c}

\newcommand{\cto}{\approx}

\newcommand{\CP}{{\sf CP}}

\newcommand{\LHCP}{{\sf LHCP}}

\newcommand{\LHRV}{{\sf LHRV}}
\newcommand{\LH}[1]{{\sf LH}_{#1}}

\newcommand{\CCP}{{\sf CCP}}
\newcommand{\pCCP}{{\sf pCCP}}
\newcommand{\iCCP}{{\sf iCCP}}

\newcommand{\CVP}{{\sf CVP}}

\newcommand{\ECCP}{{\sf ECCP}}

\newcommand{\genrelation}{\mathrel{\mathsf{R}}}

\newcommand{\successorsBy}[1]{\mathsf{Suc}_{#1}}
\newcommand{\directSuccessorsBy}[1]{\mathsf{Suc}^=_{#1}}

\newcommand{\rew}[1]{\to_{#1}}

\newcommand{\rews}[1]{\to^*_{#1}}
\newcommand{\tos}[1]{\to^*_{#1}}

\newcommand{\leftrew}[1]{{{\:}_{#1}\!\!\leftarrow\:}}

\newcommand{\leftarrown}[1]{{{\:}^n\!\!\leftarrow}}

\newcommand{\leftarrowp}[1]{{{\:}^+\!\!\leftarrow}}
\newcommand{\leftarrows}[1]{{{\:}^*\!\!\leftarrow}}

\newcommand{\hto}{\hookrightarrow}
\newcommand{\csrew}[1]{\hto_{#1}}
\newcommand{\csrews}[1]{\hto^*_{#1}}

\newcommand{\DPHabbr}{\mathsf{h}}
\newcommand{\rewdph}[1]{\stackrel{\DPHabbr}{\to}}

\newcommand{\Proc}{\mathsf{P}}

\pgfdeclarepatternformonly{VGrid}
{
  \pgfpointorigin
}
{
  \pgfqpoint{0.20cm}{0.20cm}
}
{  \pgfqpoint{0.20cm}{0.20cm}
}
{
  \pgfsetlinewidth{0.3pt}
  \pgfpathmoveto{\pgfqpoint{0cm}{0cm}}
  \pgfpathlineto{\pgfqpoint{0cm}{0.20cm}}
  \pgfusepath{stroke}
}

\pgfdeclarepatternformonly{HVGrid}
{
  \pgfpointorigin
}
{
  \pgfqpoint{0.20cm}{0.20cm}
}
{  \pgfqpoint{0.20cm}{0.20cm}
}
{
  \pgfsetlinewidth{0.3pt}
  \pgfpathmoveto{\pgfqpoint{0cm}{0cm}}
  \pgfpathlineto{\pgfqpoint{0cm}{0.20cm}}
  \pgfpathmoveto{\pgfqpoint{0cm}{0cm}}
  \pgfpathlineto{\pgfqpoint{0.20cm}{0cm}}
  \pgfusepath{stroke}
}

\newcommand{\ignore}[1]{}
\newcommand{\ignoreproofs}[1]{}

\newcommand{\fS}[1]{\mathsf{#1}}

\tikzstyle{decision} = [diamond, draw, fill=yellow!20, text width=5em, text badly centered, minimum height=4em, inner sep=0pt, aspect=2]
\tikzstyle{block} = [rectangle, draw,fill=blue!20, text width=5em, text centered, minimum height=4em, rounded corners]
\tikzstyle{cloud} = [ellipse, draw,fill=red!20, text width=5em, text centered, minimum height=4em]
\tikzstyle{line} = [draw, -latex']
\tikzstyle{blockR} = [rectangle, draw, fill=red!20, text centered, minimum height=4em, rounded corners, minimum height=0.75cm]
\tikzstyle{blockB} = [rectangle, draw, fill=blue!20, text centered, minimum height=4em, rounded corners, minimum height=0.75cm]
\tikzstyle{blockG} = [rectangle, draw, fill=green!20, text centered, minimum height=4em, rounded corners, minimum height=0.75cm]
\tikzstyle{blockY} = [rectangle, draw, fill=yellow!20, text centered, minimum height=4em, rounded corners, minimum height=0.75cm]
\tikzstyle{blockW} = [rectangle, draw, fill=white!20, text centered, minimum height=4em, rounded corners, minimum height=0.75cm]
\tikzstyle{blockK} = [rectangle, draw, fill=gray!20, text centered, minimum height=4em, rounded corners, minimum height=0.75cm]
\tikzstyle{mblockB} = [rectangle, draw, fill=blue!20, text centered, minimum height=4em, double,rounded corners, minimum height=0.75cm]
\tikzstyle{mblockW} = [rectangle, draw, fill=white!20, text centered, minimum height=4em, double,rounded corners, minimum height=0.75cm]
\tikzstyle{mblockR} = [rectangle, draw, fill=red!20, text centered, minimum height=4em, double,rounded corners, minimum height=0.75cm]
\tikzstyle{mblockG} = [rectangle, draw, fill=green!20, text centered, minimum height=4em, double,rounded corners, minimum height=0.75cm]

\tikzstyle{circleY} = [circle, draw, fill=yellow!20, text centered, minimum height=4em, rounded corners, minimum height=0.75cm]
\tikzstyle{circleR} = [circle, draw, fill=red!20, text centered, minimum height=4em, rounded corners, minimum height=0.75cm]
\tikzstyle{circleG} = [circle, draw, fill=green!20, text centered, minimum height=4em, rounded corners, minimum height=0.75cm]
\tikzstyle{circleW} = [circle, draw, fill=white!20, text centered, minimum height=4em, rounded corners, minimum height=0.75cm]
\tikzstyle{triangleW} = [isosceles triangle, draw, fill=white!20, text centered, shape border rotate = 90, isosceles triangle stretches]
\tikzstyle{triangleG} = [isosceles triangle, draw, fill=green!20, text centered, shape border rotate = 90, isosceles triangle stretches]
\tikzstyle{triangleR} = [isosceles triangle, draw, fill=red!20, text centered, shape border rotate = 90, isosceles triangle stretches]

\begin{document}

\title{Proving Confluence in the Confluence Framework with \\ CONFident}

\author{Ra\'ul Guti\'errez
\\
DLSIIS, Universidad Polit\'ecnica de Madrid, Spain\\
r.gutierrez{@}upm.es
\and Salvador Lucas\thanks{Address for correspondence: Departamento de Sistemas Inform\'aticos y Computaci\'on.
Universitat Polit\`ecnica de Val\`encia, Camino de Vera s/n, 46022 Valencia, Spain.}
\\
DSIC \& VRAIN, Universitat Polit\`ecnica de Val\`encia, Spain\\
slucas{@}dsic.upv.es
\and Miguel V\'itores\\
VRAIN, Universitat Polit\`ecnica de Val\`encia, Spain\\
miguelvitoresv{@}gmail.com
}
\maketitle

\runninghead{R. Guti\'errez et al.}{Proving Confluence with CONFident}

\begin{abstract}
This article describes the \emph{confluence framework}, a novel framework for
proving and disproving confluence using a divide-and-conquer modular strategy,
and its implementation in \CONFident{}. Using this approach, we are able to
automatically prove and disprove confluence of \emph{Generalized Term Rewriting Systems},
where (i) only selected arguments of function symbols can be rewritten
and (ii) a rather general class of conditional rules can be used.
This includes, as particular cases, several variants of rewrite systems
such as (context-sensitive) \emph{term rewriting systems},
\emph{string rewriting systems}, and
(context-sensitive) \emph{conditional term rewriting
systems}.
The divide-and-conquer modular strategy allows us to combine  in a proof tree
different techniques
for proving confluence, including
modular decompositions,
checking joinability of (conditional) critical and variable pairs, transformations, etc., and auxiliary tasks
required by them, e.g., joinability of terms, joinability of conditional
pairs, etc.

\end{abstract}


\section{Introduction}

Reduction relations $\to$ are pervasive in computer science and semantics of programming languages as suitable means to describe computations $s\to^* t$, where $\to^*$ denotes zero or more steps issued with $\to$.
In general, $s$ and $t$  are abstract values (i.e., elements of an arbitrary set $A$),
but often denote program expressions: terms, lambda expressions, configurations in imperative programming, etc.
If $s\to^*t$ holds, we say that $s$ \emph{reduces} to $t$ or that $t$ is a \emph{reduct} of $s$.
Confluence is the property of reduction  relations guaranteeing that whenever
$s$ has two different reducts $t$ and $t'$ (i.e., $s \to^* t$ and $s \to^* t'$),
both $t$ and $t'$ are \emph{joinable}, i.e., they have a common reduct $u$
(hence $t \to^* u$ and $t' \to^* u$ holds for some $u$).
Confluence is one of the most important properties of reduction relations: for instance, 
(i) it ensures that  for all expressions $s$, \emph{at most} one irreducible reduct  $t$
of $s$ can be obtained (if any); and
(ii)  it ensures that two divergent computations can always join in the future. Thus,
the semantics and implementation of rewriting-based languages is 
less dependent on specific strategies to implement reductions.

In this paper, rewriting steps are specified by using \emph{Generalized Term Rewriting Systems},
\gtrs{s} $\cR=(\Symbols,\SPredicates,\mu,H,R)$ \cite{Lucas_LocalConferenceOfConditionalAndGeneralizedTermRewritingSystems_JLAMP24},
where
(i) $\Symbols$ is a signature of function symbols;
(ii) $\SPredicates$ is a signature of predicate symbols;
(iii) replacement restrictions on selected arguments $i$ of $k$-ary function symbols $f\in\Symbols$
can be specified by means of a \emph{replacement map} $\mu$ as in \emph{context-sensitive rewriting} \cite{Lucas_ContextSensitiveRewriting_CSUR20};
also,
(iv) \emph{atomic} conditions $A$ can be included in the conditional part $\gencond$ of
conditional rules $\lhsr\to\rhsr\IF\gencond$ in $R$, provided that
(v) predicates $P$ occurring in such atomic conditions $P(t_1,\ldots,t_n)$ for terms $t_1,\ldots,t_n$  are \emph{defined} by means of Horn clauses in $H$.

Since
Term Rewriting Systems (TRSs \cite{BaaNip_TermRewAllThat_1998}),
Conditional TRSs \cite{Kaplan_ConditionalRewriteSystems_TCS84},
Context-Sensitive TRSs\linebreak (CS-TRSs \cite{Lucas_ContextSensitiveRewriting_CSUR20}),
and
Context-Sensitive CTRSs \cite[Section 8.1]{Lucas_ApplicationsAndExtensionsOfContextSensitiveRewriting_JLAMP21}
are particular cases of \gtrs{s} (see~\cite[Section 7.3]{Lucas_LocalConferenceOfConditionalAndGeneralizedTermRewritingSystems_JLAMP24}), our results apply to all of them.

\begin{example}
\label{ExCOPS_387_CS_CTRS}
Consider the
CTRS $\cR$ in \cite[Example 10]{Gmeiner_TransformationalApproachesForConditionalTermRewriteSystems_PhD14}
(COPS/387.trs\footnote{\emph{Confluence Problems Data Base}: \url{http://cops.uibk.ac.at/}.}), over a
signature $\Symbols$ consisting of function symbols
$\fS{f}$, $\fS{g}$, and $\fS{s}$, with
$R$ consisting of the \emph{conditional rules}
\begin{eqnarray}
\fS{g}(\fS{s}(x)) & \to & \fS{g}(x)\label{ExCOPS_387_CS_CTRS_rule1}\\
\fS{f}(\fS{g}(x)) & \to & x \IF x \cto \fS{s}(\fS{0})\label{ExCOPS_387_CS_CTRS_rule2}
\end{eqnarray}
where predicate $\cto$ represents a \emph{reachability test} which is performed as part of the preparation of a rewriting
step using rule (\ref{ExCOPS_387_CS_CTRS_rule2}).
This qualifies $\cR$ as an \emph{oriented} CTRS, see, e.g., \cite[Definition 7.1.3]{Ohlebusch_AdvTopicsTermRew_2002}.
This system is \emph{not} confluent, as we have the following (local) \emph{peak}\footnote{We use $\underrightarrow{\text{under}}$ and $\overleftarrow{\text{over}}$ arrows to highlight which parts of
the expression are rewritten and how.}:
\begin{eqnarray}
\fS{f}(\fS{g}(\fS{0}))\leftrew{(\ref{ExCOPS_387_CS_CTRS_rule1})}\underrightarrow{\fS{f}(\overleftarrow{\fS{g}(\fS{s}(\fS{0}))})}\rew{(\ref{ExCOPS_387_CS_CTRS_rule2})}\fS{s}(\fS{0})
\label{ExCOPS_387_CS_CTRS_peakOCTRS}
\end{eqnarray}
but $\fS{s}(\fS{0}))$ and $\fS{f}(\fS{g}(\fS{0}))$ are \emph{irreducible} (no rule applies on them)
and hence \emph{not} joinable.
Now consider the \emph{replacement map} $\muBot(f)=\emptyset$ for all symbols $f$, which forbids reductions on \emph{all} arguments of function symbols. Then, we obtain a CS-CTRSs $\cR_\bot$,
where (\ref{ExCOPS_387_CS_CTRS_peakOCTRS}) is \emph{not} possible as the leftmost reduction
is forbidden due to $\muBot(\fS{f})=\emptyset$, which disables the rewriting step on $\fS{g}(\fS{s}(\fS{0}))$.
 Using our framework we are able  to prove that $\cR$ \emph{is not confluent}, and also that $\cR_\bot$ \emph{is confluent}.
 \end{example}

Confluence has been investigated for several reduction-based formalisms and systems, see, e.g.,
\cite[Chapters 4, 7.3, 7.4 \& 8]{Ohlebusch_AdvTopicsTermRew_2002} and the references therein.
Confluence is \emph{undecidable} already for TRSs, see, e.g., \cite[Section 4.1]{Ohlebusch_AdvTopicsTermRew_2002}.
Since TRSs are \gtrs{s}, it is undecidable for \gtrs{s} too.
Thus, 
no algorithm is able to prove or disprove confluence
of the reduction relation associated to all such systems.
Hence, existing techniques for proving and disproving confluence are
\emph{partial}, i.e., they succeed on some kinds of systems and fail on others.
However, the combination of techniques in a certain order or the use of
auxiliary properties can help to prove or disprove confluence.

In this paper,
we introduce
a \emph{confluence framework} for proving confluence of \gtrs{s},
inspired by the \emph{Dependency Pair Framework}, originally
developed for proving (innermost) termination of
TRSs~\cite{GieThiSch_DPFramework_LPAR04,GieThiSchFal_MechImpDPs_JAR06}
and then generalized and extended to cope with other kind of termination
problems
\cite{BlaGenHer_DependencyPairsTerminationInDependentTypeTheoryModuloRewriting_FSCD19,%
FalKap_DependencyPairsForRewritingWithNonFreeConstructors_CADE07,%
FalKap_DPsForRewritingWithBuiltInNumbersAndSemanticDataStructures_RTA08,%
GutLuc_ProvingTerminationInTheContextSensitiveDependencyPairFramework_WRLA10,
LucMesGut_The2DDependencyPairFrameworkForConditionalRewriteSystemsPartI_JCSS18},
including termination of \gtrs{s}
\cite{Lucas_TerminationOfGeneralizedTermRewritingSystems_FSCD24},
and also to other properties like \emph{feasibility}, which tries to find a substitution
satisfying a combination of atoms with respect to a first-order theory~\cite{GutLuc_AutomaticallyProvingAndDisprovingFeasibilityConditions_IJCAR20}.

In the confluence framework, we define two kinds of problems:
three variants of \emph{confluence
problems} and
two variants of \emph{joinability problems}. \emph{Confluence problems}
encapsulate the system $\cR$ whose confluence is tested.
Such confluence problems are transformed, decomposed, simplified,
etc., into
other (possibly different) problems by using the so-called \emph{processors}
that can be plugged in and out in a proof strategy, allowing us
to find the best place to apply a proof technique in practice.
Processors embed existing results about confluence of (variants of) rewrite systems,
confluence-preserving transformations, etc.
Besides, \emph{joinability problems} are
produced by some processors acting on confluence problems.
They are used to prove or disprove the joinability of conditional pairs (including conditional critical pairs \cite[Definition 3.2]{Kaplan_FairConditionalTermRewritingSystemsUnificationTerminationAndConfluence_ADT84}
and
conditional variable pairs \cite[Definition 59]{Lucas_LocalConferenceOfConditionalAndGeneralizedTermRewritingSystems_JLAMP24}).
They are also treated by appropriate processors.
The obtained proof is depicted as a labeled \emph{proof tree} from which the (non-)confluence
of the targetted rewrite system can be proved.
Processors apply on the obtained problems until
(i) a trivial problem is obtained (which is then labeled with \yesTree) and the proof either
continues by considering pending problems, or else \emph{finishes} and \yesTree{} is returned if no problem remains to be solved; or
(ii) a counterexample is obtained and the problem is then labeled with \noTree{} and the proof \emph{finishes} as well but \noTree{} is returned; or
(iii) the successive application of all available processors finishes  \emph{unsuccessfully} and then
the whole proof finishes unsuccessfully (and `MAYBE' is returned);
or
(iv) the ongoing proof is eventually interrupted
due to a \emph{timeout}, which is usually prescribed in this kind of proof processes whose
termination is not guaranteed or could take too much time, and the whole proof fails.
The use of processors often requires calls to external tools to solve proof obligations like \emph{termination},
\emph{feasibility},
\emph{theorem proving},
etc.

\medskip
This paper is an extended and revised version of ~\cite{GutVitLuc_ConfluenceFrameworkProvingConfluenceWithCONFident_LOPSTR22}.
The main differences are:
\begin{enumerate}
\item The confluence framework has been extended to cope with Generalized Term Rewriting Systems, thus extending the scope of \cite{GutVitLuc_ConfluenceFrameworkProvingConfluenceWithCONFident_LOPSTR22}, where only \trs{s}, \cstrs{s}, and \ctrs{s} were treated.
In particular, we can treat \csctrs{s} now as particular cases of \gtrs{s}.
\item In \cite{GutVitLuc_ConfluenceFrameworkProvingConfluenceWithCONFident_LOPSTR22}
only \emph{confluence}
and \emph{joinability} problems
were considered.
Additional related problems are considered now:
\emph{local} and \emph{strong} confluence problems,
and \emph{strong joinability} problems.
They permit a better organization of confluence proofs.
Of course, they also permit the use of the framework for (dis)proving local confluence and strong confluence of \gtrs{s}.
\item 16 processors applicable to \gtrs{s} are described in this paper  (versus 10 in \cite{GutVitLuc_ConfluenceFrameworkProvingConfluenceWithCONFident_LOPSTR22}).
They apply on (local, strong) confluence problems for \gtrs{s}.
\item Details about the \emph{implementation} of the confluence framework in \CONFident{} 
are given now, including a more precise description of \CONFident{} proof strategy and its implementation.
\item Updated information about the participation of \CONFident{} in the 2023 International Confluence Competition, CoCo 2023, is provided.
\end{enumerate}

\medskip
\noindent
After some preliminaries in Section \ref{SecPreliminaries},  Section
\ref{SecGeneralizedTermRewritingSystems} describes Generalized Term Rewriting Systems.
Section~\ref{SecConfFramework} defines the
problems and processors used in the confluence framework.
Section
\ref{SecListOfProcessors} gives a list of processors that can be used in the framework.
Section~\ref{SecStrategy} presents the proof strategy of \CONFident{}.
Section~\ref{SecStructureCONFident} provides some details about the general implementation of \CONFident.
Section~\ref{SecExperimentalResults}
provides an experimental evaluation of the  tool, including an analysis of the use of processors in proofs of (non-)confluence.
Section~\ref{SecRelatedWork} discusses related work.
Section~\ref{SecConclusions} concludes.

\section{Preliminaries}\label{SecPreliminaries}

In the following, \emph{w.r.t.} means \emph{with respect to} and \emph{iff} means \emph{if and only if}.
We assume some familiarity with the basic notions of term rewriting \cite{BaaNip_TermRewAllThat_1998,Ohlebusch_AdvTopicsTermRew_2002,Terese_TermRewritingSystems_2003}
and first-order logic \cite{Fitting_FirstOrderLogicAndAutomatedTheoremProving_1997,Mendelson_IntroductionToMathematicalLogicFourtEd_1997}, where missing definitions can be found.
For the sake of readability, though, here we summarize the main notions and notations we use.

\paragraph{Abstract Reduction Relations.}
Given a binary relation $\genrelation\:\subseteq A\times A$ on a set $A$,
we often write $a\:\genrelation\:b$ or $b\:\genrelation^{-1}a$
instead of $(a,b)\in\:\genrelation$.
The \emph{reflexive} closure of $\genrelation$ is denoted by $\genrelation^=$;
the \emph{transitive} closure of $\genrelation$ is denoted by
$\genrelation^+$;
and the \emph{reflexive and transitive} closure by $\genrelation^*$.
Given $a\in A$, and a relation $\genrelation$, let $\successorsBy{\genrelation}(a)=\{b\mid a\:\genrelation^* b\}$ be the set of \emph{$\genrelation$-successors} of $a$ (and then adding $a$ itself) and $\directSuccessorsBy{\genrelation}(a)= \{b\mid a\:\genrelation^=b\}$ be the set of
\emph{direct $\genrelation$-successors} of $a$ (and also adding $a$).
An element $a\in A$ is \emph{irreducible},
if there is no $b$ such that $a\:\genrelation\:b$;
we say that $b$ is an $\genrelation$-normal form of $a$ (written $a\:\genrelation^!\:b$),
if  $a~\genrelation^*b$ and $b$ is an $\genrelation$-normal form.
We say that $b\in A$ is $\genrelation$-reachable from $a\in A$ if $a~\genrelation^*b$.
We say that $a,b\in A$ are \emph{$\genrelation$-joinable}
if there is $c\in A$ such that  $a~\genrelation^*c$ and $b~\genrelation^*c$.
We say that $a,b\in A$ are \emph{strongly $\genrelation$-joinable}
if there are $c,c'\in A$ such that  $a~\genrelation^=c$, $b~\genrelation^*c$,
$a~\genrelation^*c'$, and $b~\genrelation^=c'$.
Also, $a,b\in A$ are  \emph{$\genrelation$-convertible}
if there is $c\in A$ such that  $a~(\genrelation\cup\genrelation^{-1})^*b$.
Given $a\in A$, if there is no infinite sequence $a=a_1~\genrelation~a_2~\genrelation~\cdots~\genrelation~a_n~\genrelation\cdots$, then $a$ is $\genrelation$-\emph{terminating};
$\genrelation$ is \emph{terminating} if $a$ is $\genrelation$-terminating for all $a\in A$.
We say that $\genrelation$ is (locally) {\em confluent} if, for every $a,b,c\in A$, whenever $a~\genrelation^*b$ and $a~\genrelation^*c$ (resp.\ $a~\genrelation\: b$ and
$a~\genrelation\: c$), $b$ and $c$ are $\genrelation$-joinable.
Also, $\genrelation$ is {\em strongly confluent} if, for every $a,b,c\in A$, whenever $a\:\genrelation\:b$ and $a\:\genrelation\:c$, $b$ and $c$ are strongly $\genrelation$-joinable.
\paragraph{Signatures, Terms, Positions.}

In this paper, $\Variables$ denotes a
countable set of \emph{variables}.
A \emph{signature of symbols} is a set of \emph{symbols}
each with a fixed \emph{arity}.
We use $\Symbols$ to denote
a \emph{signature of function symbols} $f, g, \ldots,$
whose arity is given by a
mapping $ar:\Symbols\rightarrow \mathbb{N}$.
The set of
terms built from $\Symbols$ and $\Variables$ is $\Terms$.
The set of variables occurring in $t$ is $\Var(t)$;
we often use $\Var(t,t',\ldots)$ to denote the set of variables occurring in a sequence
of terms.
Terms are viewed as labeled trees in the usual way.
\emph{Positions} $p$
are represented by chains of positive natural numbers used to address subterms $t|_p$
of $t$.
The \emph{set of positions} of a term $t$ is $\Pos(t)$.
The set of positions of a subterm $s$ in $t$ is denoted $\Pos_s(t)$.
The set of positions of non-variable symbols in $t$ are denoted as $\Pos_\Symbols(t)$.

\paragraph{Replacement Maps.}

Given a signature $\Symbols$,
a \emph{replacement map} is a mapping
$\mu$
satisfying that, for all  symbols $f$ in $\Symbols$, $\mu(f)\subseteq \{1,\ldots,ar(f)\}$
 \cite{Lucas_ContextSensitiveRewriting_CSUR20}.
The set of  replacement maps for the signature $\Symbols$ is $\Rmaps{\Symbols}$.
Extreme cases
are $\muBot$,
 disallowing replacements in all arguments of function symbols:
$\muBot(f)=\emptyset$ for all $f\in\Symbols$;
and $\muTop$, restricting no replacement: $\muTop(f)=\{1,\ldots,k\}$ for all $k$-ary $f\in\Symbols$.
The
set $\Pos^\mu(t)$ of {\em $\mu$-replacing (or \emph{active}) positions}  of $t$ is
$\Pos^\mu(t)=\{\toppos\}$, if $t\in\Variables$, and
$\Pos^\mu(t)=\{\toppos\}\cup\{i.p\mid i\in\mu(f), p\in\Pos^\mu(t_i)\}$, if
$t=f(t_1,\ldots,t_k)$.
The set of {\em non-$\mu$-replacing} (or \emph{frozen}) positions of $t$ is $\ol{\Pos^\mu}(t)=\Pos(t)-
\Pos^\mu(t)$.
Positions of \emph{active} non-variable symbols in $t$ are denoted as $\Pos^\mu_\Symbols(t)$.
Given a term $t$, $\Var^\mu(t)$ (resp.\ $\NVar{\mu}(t)$)
is the set of variables occurring at active (resp.\ frozen)
positions
in $t$: $\Var^\mu(t)=\{x\in\Var(t)\mid\exists p\in\Pos^\mu(t),x=t|_p\}$
and
 $\NVar{\mu}(t)=\{x\in\Var(t)\mid\exists p\in\NPos{\mu}(t),x=t|_p\}$.
 In general, $\Var^\mu(t)$ and $\NVar{\mu}(t)$  are not disjoint: $x\in\Var(t)$ may occur active and also
 frozen in $t$.

\paragraph{Unification.}
A renaming $\rho$ is a bijection from $\Variables$ to $\Variables$.
A substitution $\sigma$ is  a mapping
$\sigma:\Variables\to\Terms$
from variables
into terms which is homomorphically extended to a mapping (also
denoted $\sigma$) $\sigma:\Terms\to\Terms$.
It is standard to assume that substitutions $\sigma$
satisfy $\sigma(x)=x$ except for a
\emph{finite} set
of variables.
Thus, we often write $\sigma=\{x_1\mapsto t_1,\ldots,x_n\mapsto t_n\}$ where $t_i\neq x_i$ for $1\leq i\leq n$
to denote a substitution.
Terms $s$ and $t$ \emph{unify} if
there is a substitution $\sigma$ (i.e., a \emph{unifier}) such that $\sigma(s)=\sigma(t)$. If $s$ and $t$ unify, then there is a (unique, up to renaming) \emph{most general unifier} (\emph{mgu})
$\theta$ of $s$ and $t$ satisfying that, for any other unifier $\sigma$ of $s$ and $t$, there is a
substitution $\tau$ such that, for all $x\in\Variables$, $\sigma(x)=\tau(\theta(x))$.

\paragraph{First-Order Logic.}

Here, $\SPredicates$ denotes a signature of \emph{predicate symbols}.
Atoms and
first-order formulas
are built using such
function and predicate symbols, variables in $\Variables$,
quantifiers $\forall$ and $\exists$ and
logical connectives for conjunction $(\wedge)$,
disjunction $(\vee)$,
negation $(\neg)$,
and implication $(\Rightarrow)$,
in the usual way.
A first-order theory (FO-theory for short) $\GLtheory$ is a set of sentences
(formulas whose variables are all \emph{quantified}).
In the following, given an FO-theory
$\GLtheory$ and a formula $\varphi$,
$\GLtheory\vdash\varphi$ means that $\varphi$ is \emph{deducible} from
(or a \emph{logical consequence} of) $\GLtheory$ by using
a correct and complete deduction procedure.

\paragraph{Feasibility Sequences.}

An
\emph{f-condition} $\fcond$
is an atom
 \cite{GutLuc_AutomaticallyProvingAndDisprovingFeasibilityConditions_IJCAR20}.
Sequences
$\fseq=(\fcond_i)_{i=1}^n=(\fcond_1,\ldots,\fcond_n)$
of f-conditions
are called \emph{f-sequences}.
We often drop `f-' when no confusion arises.
Given an FO-theory $\GLtheory$, a condition $\fcond$ is $\GLtheory$-\emph{feasible}
(or just \emph{feasible} if no confusion arises)
if $\deductionInThOf{\GLtheory}{\sigma(\fcond)}$
holds for some  substitution $\sigma$;
otherwise, it is \emph{infeasible}.
A sequence $\fseq$ is $\GLtheory$-feasible (or just \emph{feasible})  if there is a substitution $\sigma$ satisfying all conditions in the sequence, i.e.,
for all $\fcond\in\fseq$, $\deductionInThOf{\GLtheory}{\sigma(\fcond)}$
holds.

\paragraph{Grounding variables.}
Let $\Symbols$ be a signature and $\Variables$ be a set of variables such that $\Symbols\cap\Variables=\emptyset$.
Let $\Symbols_\Variables=\Symbols\cup C_\Variables$ where variables
$x\in\Variables$ are considered as (different) \emph{constant} symbols $c_x$ of
$C_\Variables=\{c_x\mid x\in\Variables\}$ and $\Symbols$ and $C_\Variables$ are disjoint
\cite{GutLucVit_ConfluenceOfConditionalRewritingInLogicForm_FSTTCS21}, see also
\cite[page 224]{AveLor_ConditionalRewriteSystemsWithExtraVariablesAndDeterministicLogicPrograms_LPAR94}.
Given a term $t\in\Terms$, a ground term $t^\downarrow$ is
obtained by replacing each occurrence of $x\in\Variables$ in $t$ by $c_x$.
Given a substitution $\sigma=\{x_1\mapsto t_1,\ldots,x_n\mapsto t_n\}$, we define
$\sigma^\downarrow=\{x_1\mapsto t^\downarrow_1,\ldots,x_n\mapsto t^\downarrow_n\}$.

\section{Generalized term rewriting systems}
\label{SecGeneralizedTermRewritingSystems}

The material in this section is taken from \cite[Section 7]{Lucas_LocalConferenceOfConditionalAndGeneralizedTermRewritingSystems_JLAMP24}.
We consider definite Horn clauses $\alpha:A\IF \gencond$ (with label $\alpha$)
where $\gencond$ is a sequence $A_1,\ldots,A_n$ of atoms.
If $n=0$, then $\alpha$ is written $A$ rather than $A\IF$.
Let $\Symbols$ be a signature of function symbols,
$\SPredicates$ be a signature of predicate symbols,
$\mu\in\Rmaps{\Symbols}$ be a replacement map,
$H$ be a set of clauses $A\Leftarrow \gencond$ where $\rootTerm(A)\notin\{\rew{},\rews{}\}$,
and
$R$ be a set of \emph{rewrite rules} $\lhsr\to\rhsr\IF\gencond$ such that $\lhsr$ is not a variable
(in both cases, $\gencond$ is a sequence $A_1,\ldots,A_n$ of atoms).
The tuple
$\cR=(\Symbols,\SPredicates,\mu,H,R)$ is called
a \emph{Generalized Term Rewriting System} (\gtrs{},
\cite[Definition 51]{Lucas_LocalConferenceOfConditionalAndGeneralizedTermRewritingSystems_JLAMP24}).
As in \cite[Definition 6.1]{MidHam_CompletenessResultsForBasicNarrowing_AAECC94}, rules
$\lhsr\to \rhsr \IF\gencond\in R$
are classified according to the distribution of
variables:
type 1,  if $\Var(r)\cup\Var(c)\subseteq\Var(\ell)$;
type 2, if $\Var(r)\subseteq\Var(\ell)$;
type 3, if $\Var(r)\subseteq\Var(\ell)\cup\Var(c)$; and
type 4, otherwise.
A rule of type $n$ is often called an $n$-rule.
A \gtrs{} $\cR$ is called an $n$-\gtrs{} if all its rules are of type $n$;
if $\cR$ contains at least one $n$-rule which is \emph{not} an $m$-rule for some $m<n$, then
we say that $\cR$ is a \emph{proper} $n$-\gtrs{}.
The FO-theory of a \gtrs{} $\cR=(\Symbols,\SPredicates,\mu,H,R)$ is
\[\ol{\cR}=\{(\RuleReflexivity),(\RuleCompatibility)\}
\cup
\{(\RulePropagation)_{f,i}\mid f\in\Symbols,i\in\mu(f)\}
\cup
\{(\RuleHornClause)_\alpha\mid \alpha\in H\cup R\}
\]
where, as displayed in Table \ref{TableFOSentencesCSCTRSs},
\begin{table}[!h]
\caption{Generic sentences of the first-order theory of rewriting}
\begin{center}
\begin{tabular}{l@{~}l}
Label & Sentence\\
\hline
$(\RuleReflexivity)$
&
$(\forall x)~x \tos{} x$
\\
$(\RuleCompatibility)$
&
$(\forall x,y,z)~x\to y\wedge y \tos{} z\Rightarrow x\tos{} z$
\\
$(\RulePropagation)_{f,i}$
&
$(\forall x_1,\ldots,x_k,y_i)~x_i\to y_i\Rightarrow f(x_1,\ldots,x_i,\ldots,x_k)\to{} f(x_1,\ldots,y_i,\ldots,x_k)$
\\
$(\RuleHornClause)_{A\IF A_1,\ldots,A_n}$
&
$(\forall x_1,\ldots,x_p)~A_1\wedge\cdots\wedge A_n\Rightarrow A$\\
&\hspace{0.5cm} where $x_1,\ldots,x_p$ are the variables occurring in $A_1,\ldots,A_n$ and $A$\\
\hline
\end{tabular}
\end{center}
\label{TableFOSentencesCSCTRSs}
\end{table}
$(\RuleReflexivity)$ expresses \emph{reflexivity} of many-step rewriting;
$(\RuleCompatibility)$ expresses \emph{compatibility} of one-step
and many-step rewriting;
for each $k$-ary function symbol $f$, $i\in\mu(f)$, and $x_1,\ldots,x_k$ and $y_i$ distinct variables,
$(\RulePropagation)_{f,i}$ enables the \emph{propagation} of rewriting steps in the $i$-th immediate \emph{active} subterm of a term with root symbol  $f$;
finally,
for each Horn clause $\alpha\in H\cup R$,
$(\RuleHornClause)_\alpha$ makes explicit
 the
relationship between Horn clause symbol $\IF$
(also used in rewrite rules which are particular Horn clauses, actually) and
logic implication $\Rightarrow$.
\begin{definition}[Rewriting as deduction]
\label{DefRewritingAsDeduction}
Let $\cR$ be a \gtrs{}.
For all terms $s$ and $t$, we write
$s\rew{\cR}t$ (resp.\ $s\rews{\cR}t$)
if $\ol{\cR}\vdash s\to t$
(resp.\ $\ol{\cR}\vdash s\to^* t$).
\end{definition}
\begin{figure}[t]\small
\vspace*{-2mm}
\begin{tabular}{ll}
Join
&
$
\begin{array}{crl}
(\RuleJO) & (\forall x,y,z) & x \to^* z\wedge y\to^*z\Rightarrow x\cto y
\end{array}
$
 \\[0.3cm]
Oriented
&
$
\begin{array}{crl}
(\RuleOR) & (\forall x,y)& x \to^* y\Rightarrow x\cto y\hspace{1.2cm}
\end{array}
$
 \\[0.3cm]
Semi-equational
&
$
\begin{array}{crl}
(\RuleSE_1) & (\forall x) & x\cto x
\\
(\RuleSE_2) & (\forall x,y,z) & x \rew{} y\wedge y \cto z\Rightarrow x\cto z
\\
(\RuleSE_3) & (\forall x,y,z) & y \rew{} x\wedge y \cto z\Rightarrow x\cto z
\end{array}
$
\end{tabular}
\caption{Sentences for different semantics of \csctrs{s}}
\label{FigSentencesForSemanticsOfCTRSs}\vspace*{-2mm}
\end{figure}
Figure \ref{FigSentencesForSemanticsOfCTRSs} displays some sentences
which can be included in $H$ to make
the  \emph{meaning} of predicate $\cto$ often used in the conditions of rules explicit;
see, e.g., \cite[Definition 7.1.3]{Ohlebusch_AdvTopicsTermRew_2002}:
$(\RuleJO)$ interprets $\cto$ as \emph{$\rew{\cR}$-joinability} of terms;
$(\RuleOR)$ interprets $\cto$ as \emph{$\rew{\cR}$-reachability};
and
$(\RuleSE)_1$, $(\RuleSE)_2$, $(\RuleSE)_3$ provide the interpretation of
$\cto$ as \emph{$\rew{\cR}$-conversion}.
Accordingly, we let
$H_\cto=\{(\RuleJO)\}$, or
$H_\cto=\{(\RuleOR)\}$, or
$H_\cto=\{(\RuleSE)_1,(\RuleSE)_2,(\RuleSE)_3\}$
and then we include $H_\cto$ in $H$.
Given a \gtrs{} $\cR=(\Symbols,\SPredicates,\mu,H,R)$,
a number of well-known classes of rule-based systems is obtained:
\begin{itemize}
\itemsep=0.9pt
\item if $\SPredicates=\{\rew{},\rews{}\}$,
$\mu=\muTop$, and $R$ consists of  unconditional $2$-rules only, then $\cR$ is a
TRS and we often just refer to it as $(\Symbols,R)$.
\item if $\SPredicates=\{\rew{},\rews{}\}$
and $R$ consists of  unconditional $2$-rules only, then $\cR$ is a
CS-TRS
\cite{Lucas_ContextSensitiveRewriting_CSUR20} and we often just refer to it as
$(\Symbols,\mu,R)$.
\item if $\SPredicates=\{\rew{},\rews{},\cto\}$
and for all $\lhsr\to\rhsr\IF \gencond\in R$,
(i)  $\gencond$ consists of conditions
$s\cto t$ (for some terms $s$ and $t$),
and (ii) $H=H_\cto$
then, depending on $H_\cto$, as explained above, $\cR$ is a
(J-,O-,SE-)\csctrs{};
if $\mu=\muTop$, then $\cR$ is a (J-,O-,SE-)CTRS.
 If no confusion arises, we often use $(\Symbols,\mu,R)$ and $(\Symbols,R)$
instead of $(\Symbols,\SPredicates,\mu,H_\cto,R)$ or
$(\Symbols,\SPredicates,\muTop,H_\cto,R)$, although these notations are more self-contained as $H$ embeds the evaluation semantics of $\cto$ \cite[Remark 53]{Lucas_LocalConferenceOfConditionalAndGeneralizedTermRewritingSystems_JLAMP24}.
\end{itemize}
\begin{definition}[Confluence and termination of \gtrs{s}]
\label{DefConfluenceOfGTRSs}
\label{DefTerminationOfGTRSs}
A \gtrs{} $\cR$ is (locally) confluent (resp.\ terminating)
if $\rew{\cR}$ is (locally) confluent (resp.\ terminating).
\end{definition}

\begin{remark}[Confluence and termination of CS-TRSs and CS-CTRSs]
As remarked above, TRSs, CS-TRSs, CTRSs, and CS-CTRSs are particular cases of \gtrs{s}.
In the realm of context-sensitive rewriting, it is often useful to make explicit the replacement map $\mu$ when
referring to the context-sensitive rewriting relation (by writing $\csrew{\cR,\mu}$ and $\csrews{\cR,\mu}$, or just $\csrew{}$ and $\csrews{}$)
and computational properties of CS-TRSs and CS-CTRSs $\cR$
using a replacement map $\mu$, i.e., we usually talk of $\mu$-termination or (local, strong) $\mu$-confluence of $\cR$, see \cite{Lucas_ContextSensitiveRewriting_CSUR20}.
This is useful to \emph{compare} properties of context-sensitive systems and the corresponding properties of unrestricted systems (TRSs and CTRSs).
In this paper, though, we use a uniform notation, $\rew{\cR}$, for the rewrite relation associated to a \gtrs{} $\cR$, and also a uniform designation of the properties,
just following Definition \ref{DefConfluenceOfGTRSs}.
\end{remark}

A rule  $\alpha:\lhsr\to\rhsr\IF \gencond\in\cR$
is \emph{(in)feasible} if $\gencond$ is $\ol{\cR}$-(in)feasible.
Two terms $s$ and $t$ are $\rew{\cR}$-joinable iff $s^\downarrow$ and $t^\downarrow$
are $\rew{\cR}$-joinable, cf.\ \cite[Proposition 6]{GutLucVit_ConfluenceOfConditionalRewritingInLogicForm_FSTTCS21}.
As in \cite[Section 5]{Lucas_LocalConferenceOfConditionalAndGeneralizedTermRewritingSystems_JLAMP24} we consider \emph{conditional pairs} $\langle s,t\rangle \IF  A_1,\ldots,A_n$, where $s,t$ are terms and $A_1,\ldots,A_n$ are atoms.
For a \gtrs{} $\cR$, we say that
$\pi:\langle s,t\rangle\IF\gencond$ is \emph{(in)feasible}  if $\gencond$ is $\ol{\cR}$-(in)feasible.
Also,  $\pi$ is  (strongly) \emph{joinable}
if for all
substitutions $\sigma$, whenever $\ol{\cR}\vdash\sigma(\fcond)$ holds for all $\fcond\in\gencond$,
 terms $\sigma(s)$
and $\sigma(t)$ are (strongly) joinable.
A conditional pair $\pi$ is \emph{trivial} if $s=t$.
Trivial and infeasible con\-di\-tion\-al pairs are both joinable.

\begin{definition}[Extended critical pairs of a \gtrs{}, {\cite[Definitions 59 \& 60]{Lucas_LocalConferenceOfConditionalAndGeneralizedTermRewritingSystems_JLAMP24}}]
\label{DefConditionalPairsGTRS}

Let $\cR=(\Symbols,\SPredicates,\mu,H,R)$ be a  \gtrs{} and
$\alpha:\lhsr\to\rhsr\IF\gencond,\alpha':\lhsr'\to\rhsr'\IF\gencond'\in R$
feasible rules sharing no variable (rename if necessary).
\begin{itemize}
\itemsep=0.85pt
\item
Let $p\in\Pos^\mu_\Symbols(\lhsr)$ be a nonvariable position of $\lhsr$
such that $\lhsr|_p$ and $\lhsr'$ unify with
\emph{mgu} $\theta$.
Then,
\begin{eqnarray}
\langle\theta(\lhsr[\rhsr']_p),\theta(\rhsr)\rangle\IF\theta(c),\theta(c') \label{LblConditionalCriticalPair}
\end{eqnarray}
is a \emph{conditional critical pair} (\CCP)
of $\cR$.
If $p=\toppos$ and $\alpha'$ is a renamed version of $\alpha$, then (\ref{LblConditionalCriticalPair}) is called \emph{improper};
otherwise, it is called \emph{proper}.
\item Let $x\in\Var^\mu(\lhsr)$,
$p\in\Pos^\mu_x(\lhsr)$,
and $x'$ be a fresh variable.
Then,
\begin{eqnarray}
\langle\lhsr[x']_p,\rhsr\rangle\IF x\to x',c \label{LblVariableCriticalPair}
\end{eqnarray}
is a \emph{conditional variable pair}
(\CVP) of $\cR$.
Variable $x$ is called the \emph{critical variable} of the pair.
\end{itemize}
In both cases,
$p$ is called the \emph{critical position}. We use the following notation:
\begin{itemize}
\itemsep=0.85pt
\item  $\pCCP(\cR)$ is the set of feasible
\emph{proper  conditional critical pairs} of $\cR$;
\item $\iCCP(\cR)$ is the set of
feasible \emph{improper
conditional critical pairs}  of $3$-rules
in $\cR$ (as improper critical pairs of 2-rules are joinable); and
\item $\CVP(\cR)$ is the set of all
\emph{feasible conditional variable pairs} in $\cR$.
\end{itemize}
Then,
\[\ECCP(\cR)=\pCCP(\cR)\cup\iCCP(\cR)\cup\CVP(\cR)\]
is the set of \emph{extended conditional critical pairs} of $\cR$.
\end{definition}
For unconditional systems (\cstrs{s} and \trs{s}) we can focus on
smaller sets of possibly conditional pairs to analize confluence:
\begin{itemize}
\itemsep=0.85pt
\item For \cstrs{s} $\cR$, we have $\iCCP(\cR)=\emptyset$ ($\cR$ contains no $3$-rule),
and $\pCCP(\cR)$ is written $\CP(\cR,\mu)$ or just $\CP(\cR)$
if no confusion arises.
Also, following \cite[Definition 20]{LucVitGut_ProvingAndDisprovingConfluenceOfContextSensitiveRewriting_JLAMP22}, instead of $\CVP(\cR)$, we can use the set of $\LH{\mu}$-critical pairs
\[\begin{array}{rcl}
\LHCP(\cR,\mu) &\!\! =\!\! & \{\langle \lhsr[x']_p,\rhsr\rangle\IF x\rew{}x'\mid \lhsr\to\rhsr\in\cR,x\in\Var^\mu(\lhsr)\cap(\NVar{\mu}(\lhsr)\cup\NVar{\mu}(\rhsr)),\\
&& \hspace{1.8cm} p\in\Pos^\mu_x(\lhsr)\}\\
&\!\! \subseteq\!\! & \CVP(\cR)
\end{array}
\]
which
provides a more specific set of conditional pairs
obtained from an unconditional rule $\lhsr\to\rhsr$
capturing \emph{possibly harmful} peaks coming from variables which are \emph{active} in the left-hand side $\lhsr$ but are also \emph{frozen} in the same $\lhsr$ or else in the right-hand side $\rhsr$ of the rule
(see \cite[Section 8]{Lucas_LocalConferenceOfConditionalAndGeneralizedTermRewritingSystems_JLAMP24} for a comparison of $\LH{\mu}$-critical pairs and conditional variable pairs for a \cstrs{} $\cR$).
\item For TRSs $\cR$ (where
$\muTop$ can be assumed to view it as a CS-TRS),
$\iCCP(\cR)=\emptyset$ (no $3$-rules), $\LHCP(\cR,\muTop)=\emptyset$ (as $\NVar{\muTop}(t)=\emptyset$ for all terms $t$) and $\pCCP(\cR)$ is written $\CP(\cR)$.
\end{itemize}

\eject

\begin{example}
\label{ExCOPS_387_CS_CTRS_pCCPs_iCCPs_CVPs}
Consider the CTRS $\cR$ in Example \ref{ExCOPS_387_CS_CTRS}.
With rule (\ref{ExCOPS_387_CS_CTRS_rule2}),
i.e., $\fS{f}(\fS{g}(x)) \to x \IF x \cto \fS{s}(\fS{0})$,
$1\in\Pos_\Symbols(\lhsr_{(\ref{ExCOPS_387_CS_CTRS_rule2})})$ and
(\ref{ExCOPS_387_CS_CTRS_rule1})', i.e.,
$\fS{g}(\fS{s}(x')) \to \fS{g}(x')$, since
$\fS{f}(\fS{g}(x))|_1=\fS{g}(x)$ and $\fS{g}(\fS{s}(x'))$ unify with $\theta=\{x\mapsto \fS{s}(x')\}$,
we obtain the following (feasible) \emph{proper conditional critical pair}
\begin{eqnarray}
\langle\fS{f}(\fS{g}(x')) , \fS{s}(x')\rangle \IF \fS{s}(x') \cto\fS{s}(\fS{0})
\label{ExCOPS_387_CS_CTRS_pCCP}
\end{eqnarray}
Thus, $\pCCP(\cR)=\{(\ref{ExCOPS_387_CS_CTRS_pCCP})\}$.
Since $\cR$ is a 1-CTRS, $\iCCP(\cR)=\emptyset$.
With (\ref{ExCOPS_387_CS_CTRS_rule2}) and $x\in\Var^\muTop(\fS{f}(\fS{g}(x)))=\Var(\fS{f}(\fS{g}(x)))$, we obtain
$\CVP(\cR)=\{(\ref{ExCOPS_387_CS_CTRS_CVP})\}$ for the
following \emph{conditional variable pair}:
\begin{eqnarray}
\langle\fS{f}(\fS{g}(x')),x\rangle \IF x\rew{}x', x \cto \fS{s}(\fS{0})\label{ExCOPS_387_CS_CTRS_CVP}
\end{eqnarray}
Thus, $\ECCP(\cR)=\{(\ref{ExCOPS_387_CS_CTRS_pCCP}),(\ref{ExCOPS_387_CS_CTRS_CVP})\}$.
\end{example}

\begin{example}
\label{ExCOPS_387_CS_CTRS_Rbot_NoECCPs}
For the \csctrs{} $\cR_\bot$ in Example \ref{ExCOPS_387_CS_CTRS} there is no conditional critical pair because the only active position of the left-hand sides
$\lhsr_{(\ref{ExCOPS_387_CS_CTRS_rule1})}$ and
$\lhsr_{(\ref{ExCOPS_387_CS_CTRS_rule2})}$ of rules
(\ref{ExCOPS_387_CS_CTRS_rule1}) and  (\ref{ExCOPS_387_CS_CTRS_rule2}) is
$\toppos$.
However, $\lhsr_{(\ref{ExCOPS_387_CS_CTRS_rule1})}$ and
$\lhsr_{(\ref{ExCOPS_387_CS_CTRS_rule2})}$  do \emph{not} unify.
Hence $\pCCP(\cR_\bot)=\emptyset$ and, again, $\iCCP(\cR_\bot)=\emptyset$, being a
1-\csctrs.
Also,
$\CVP(\cR_\bot)=\emptyset$ because all variables in the left-hand sides of the rules in
 $\cR_\bot$ are \emph{frozen} due to $\muBot$.
 Thus, $\ECCP(\cR_\bot)=\emptyset$.
 \end{example}

\section{Confluence framework}\label{SecConfFramework}

This section describes our confluence framework for proving and disproving
(local, strong) confluence of \gtrs{s}.
As mentioned in the introduction, the
framework is inspired by existing frameworks for proving termination of (variants of)
TRSs, starting from \cite{GieThiSch_DPFramework_LPAR04}.
We encapsulate the different stages of confluence proofs for a given \gtrs{} $\cR$ as
\emph{problems} and the techniques used to treat them and develop the proof
as \emph{processors}. Proofs are organized in a \emph{proof tree} whose nodes are
the aforementioned problems and whose branches are defined by the (possibly repeated
and parallel)
use of processors. Section \ref{SecListOfProcessors} describes a list of processors that can be used in the confluence
framework.
Forthcoming techniques can often be implemented as a new processor
of the framework and then included in an existing strategy for a practical use.
In general, it is hard to find a single technique that is able to obtain
a complete proof at once. In practice, most proofs are a
combination of different (repeatedly used) techniques.
This requires the definition of \emph{proof strategies} as a combination of processors.
Section \ref{SecStrategy} describes \CONFident's proof strategy.

\subsection{Problems}
\label{SecProblems}

A \emph{problem} is just a structure that contains information used to prove
the property we want to analyze. We define some variants of
\emph{Confluence} and \emph{Joinability} problems.
\emph{Confluence problems} are used to (dis)prove confluence of
\gtrs{s} (and also local and strong confluence),
and \emph{joinability problems} are used to (dis)prove
(strong) joinability of conditional pairs.
In the following, we often use $\tau$ to refer to a problem when no confusion arises.

\begin{definition}[Confluence Problems]\label{def:CProblem}
Let $\cR$ be a \gtrs{}.
A (local, strong) \emph{confluence problem},  denoted
	$\CRproblem{\cR}$ (resp.\ $\WCRproblem{\cR}$, $\SCRproblem{\cR}$),   is \emph{positive} if
	$\cR$ is (locally, strongly) confluent; otherwise, it is \emph{negative}.
\end{definition}

\begin{definition}[Joinability Problems]\label{def:JProblem}
Let $\cR$ be a \gtrs{} and $\pi$ be a \emph{conditional pair}.
	A (strong) \emph{joinability problem}, denoted
	$\JOproblem{\cR,\pi}$ (resp.\ $\SJOproblem{\cR,\pi}$), is \emph{positive}
	if $\pi$ is (strongly) joinable; otherwise, it is
	\emph{negative}.
\end{definition}
In the following, unless established otherwise, our definitions and results pay no attention to the
specific type (confluence or joinability) of problems at stake. We just refer to them as ``problems''.

\begin{remark}[Relationship between problems]
\label{RemRelationshipBetweenProblems}
From well-known results, see, e.g.,
\cite[Section 2.7]{BaaNip_TermRewAllThat_1998}, the following relations hold for these problems:
\begin{center}
\begin{tabular}{l}
If $\SCRproblem{\cR}$ is positive, then $\CRproblem{\cR}$ is positive.\\[2pt]
If $\CRproblem{\cR}$ is positive, then $\WCRproblem{\cR}$ is positive.\\[2pt]
If $\SJOproblem{\cR,\pi}$ is positive, then $\JOproblem{\cR,\pi}$ is positive.
\end{tabular}
\end{center}
It is well-known that, in general, these implications cannot be reversed.
\end{remark}

\subsection{Processors}\label{SecProcessors}

A processor $\Proc$ is a partial function that takes a problem
$\tau$ as an input and, if $\Proc$ is defined for $\tau$, then it returns either
a (possibly empty) set of problems $\tau_1,\ldots,\tau_n$ for some $n\geq 0$ or ``\noTree''.
Usually,
$\tau_1,\ldots,\tau_n$ are (hopefully) \emph{simpler} problems.
We say that a processor is \emph{sound} if it propagates \emph{positiveness} of \emph{all} returned problems
$\tau_1,\ldots,\tau_n$ upwards as \emph{positiveness} of the input problem $\tau$.
Symmetrically, a processor is \emph{complete} if \emph{negativeness} of \emph{some} returned problem is propagated upwards as \emph{negativeness} of the input problem $\tau$.
Furthermore, if a
complete processor returns ``\noTree'', it tells us that $\tau$ is negative
and if a sound processor returns an empty set of problems, then $\tau$ is
trivially positive.

\begin{definition}
\label{DefProcessorSoundnessAndCompleteness}
A \emph{processor} $\Proc{}$ is a partial function from
problems into sets of problems; alternatively it can return ``\noTree''.
The domain of $\Proc{}$ (i.e., the set of problems on which $\Proc$ is defined)
is denoted $\dom(\Proc{})$.
We say that $\Proc{}$ is
\begin{itemize}
\item \emph{sound} if for all $\tau \in \dom(\Proc{})$, $\tau$ is positive
whenever $\Proc{}(\tau) \neq$ ``\noTree'' and all $\tau'\in\Proc{}(\tau)$ are
positive.
\item \emph{complete} if for all $\tau \in \dom(\Proc{})$,
$\tau$ is negative whenever $\Proc{}(\tau) =$ ``\noTree'' or some $\tau'\in\Proc{}(\tau)$  is negative.
\end{itemize}
\end{definition}
Roughly speaking, soundness is used for \emph{proving} problems \emph{positive},
and completeness is used to prove them \emph{negative}.
Sound and complete processors are obviously desirable, as they can be used for both purposes.
However, it is often the case that processors which are sound but not complete are
available and heavily used (and vice versa) as they implement important techniques for (dis)proving
confluence.
Section \ref{SecListOfProcessors} describes several  processors
and their use in the confluence framework.

\subsection{Proofs in the confluence framework}
\label{SecProofsInTheConfluenceFramework}

Confluence problems can be proved positive or negative by using a proof tree as
follows.
Our definitions and results are given, in particular, for confluence problems $\CRproblem{\cR}$ for a \gtrs{} $\cR$.
They straightforwardly adapt to $\WCRproblem{\cR}$, $\SCRproblem{\cR}$, $\JOproblem{\cR,\pi}$,
and $\SJOproblem{\cR,\pi}$.

\begin{definition}[Confluence Proof Tree] {\rm }\label{DefConfProofTree}
	Let $\cR$ be a \gtrs{}. A confluence proof tree $\cT$ for $\cR$
	 is a tree whose root label is $\CRproblem{\cR}$, whose inner (i.e., non-leaf) nodes are labeled with problems $\tau$,
	and whose leaves are labeled either with problems $\tau$, or with ``\yesTree'' or ``\noTree''. For every inner
	node $\mathsf{n}$ labeled with $\tau$, there is a processor $\Proc{}$ such
	that $\tau \in \dom(\Proc{})$ and:
	\begin{enumerate}
		\item if $\Proc{}(\tau)=$``\noTree'' then $\mathsf{n}$ has just one child, labeled
		with ``\noTree''.
		\item if $\Proc{}(\tau)=\emptyset$ then $\mathsf{n}$ has just one child, labeled with
		``\yesTree''.
		\item if $\Proc{}(\tau) = \{\tau_1,\ldots,\tau_m\}$ with $m > 0$, then
		$\mathsf{n}$ has $m$ children labeled with the problems
		$\tau_1,\ldots,\tau_m$.
	\end{enumerate}
\end{definition}

\noindent In this way, a confluence proof tree is obtained by the combination of
different  processors.
The proof of the following result is obvious from the previous definitions.

\begin{theorem}[Confluence Framework] {\rm }\label{ThConfFramework}
	Let $\cR$ be a \gtrs{} and $\cT$ be a confluence proof tree for $\cR$. Then:
	\begin{enumerate}	
		\item if all leaves in $\cT$ are labeled with {\sf ``yes''} and all
		involved processors are sound for the problems they
		are applied to, then $\cR$ is confluent.
		\item if $\cT$ has a leaf labeled with {\sf ``no''} and all
		processors in the path from the root to such a leaf are complete for the problems they are applied to, then $\cR$ is not confluent.
	\end{enumerate}
\end{theorem}
Figures~\ref{FigProofTreesInTheConfluenceFramework} and \ref{FigProofTreesInTheConfluenceFramework2}
at the end of Setion \ref{SecListOfProcessors}
display examples of proofs obtained by using the confluence framework for some of the examples discussed
in this paper.

\section{List of processors}
\label{SecListOfProcessors}

In this section, we enumerate some processors for use in the confluence
framework, organized according to their functionality.
Table \ref{TableListOfProcessors} displays the complete list, which we develop in the following sections.
Table \ref{TableFinishingProcessorsInTheConfluenceFramework}
shows which processors (according to their definitions)
are able to \emph{finish} a proof branch
in a proof tree by either returning an empty set (which is translated by labeling with ``\yesTree''
a leaf of the tree) or by directly returning ``\noTree''.

\begin{table}[!ht]
\vspace*{-3mm}
\caption{Available processors}
\begin{center}
\scalebox{0.95}{
\begin{tabular}{lll}
Group & Section & Processors\\
\hline
Cleansing & \ref{SecCleansingProcessors} & $\ProcExtraVars$, $\ProcSimplify$, $\ProcInlining$\\
Modular decomposition & \ref{SectionModularDecomposition} & $\ProcModularDecomp$\\
Local/Strong confluence & \ref{SecProcessorForLocalStrongConfluenceProblems} & $\ProcHuetExtended$\\
Confluence of \ctrs{s} by transformation & \ref{SecProcTrU}--\ref{SecProcTrUconf} &
$\ProcTrU$,
$\ProcTrUconf$\\
Confluence and orthogonality & \ref{SecOrthogonality} & $\ProcOrthogonality$\\
Confluence and local/strong confluence & \ref{SecConfluenceAsLocalStrongConfluence} & $\ProcConfluence$, $\ProcLocalConfluence$, $\ProcStrongConfluence$\\
Confluence and (local) confluence of \csr & \ref{SecConfluenceAsLocalConfluenceOfCSR}--\ref{SecConfluenceAsCanonicalConfluenceOfCSR} & $\ProcCanJoinability$, $\ProcCnvJoinability$, $\ProcCanCSR$\\
Confluence by termination and local confluence & \ref{SecKnuthBendix} & $\ProcKnuthBendix$\\
Joinability & \ref{SecJoinabilityProcessor} & $\ProcJoinability$
\end{tabular} }
\end{center}
\label{TableListOfProcessors}
\end{table}

\begin{table}[ht]
\vspace*{-3mm}
\caption{Ending Processors in the Confluence Framework}
\begin{center}
\scalebox{0.95}{
\begin{tabular}{ll@{\hspace{1.5cm}}|ll@{\hspace{1.5cm}}|ll}
Proc.\ & May end with &
Proc.\ & May end with &
Proc.\ & May end with
\\
\hline
$\ProcExtraVars$ & \noTree
&
$\ProcOrthogonality$  & \yesTree
&
$\ProcKnuthBendix$  & \yesTree\\
$\ProcHuetExtended$  & \yesTree &
$\ProcJoinability$  & \yesTree\:/\:\noTree
\end{tabular} }
\end{center}
\label{TableFinishingProcessorsInTheConfluenceFramework}
\end{table}

\subsection{Cleansing processors}
\label{SecCleansingProcessors}

In this section we present a number of processors implementing
simple tests to detect and correct particular situations (extra variables, trivial rules, trivial conditions in rules, infeasible conditions, etc.) leading to simplifications
of rules or even to an immediate answer.
Sometimes, this is done on the input \gtrs{}, just before attempting a proof;
sometimes after
applying processors that split the system into components, or that
transform the rules to produce
new ones exhibiting such problems.

\medskip
We consider processors
$\ProcExtraVars$ which checks whether rules with extra variables may definitely imply a non-confluent behavior;
$\ProcSimplify$ which removes trivial (components of) rules to simplify them; and
$\ProcInlining$ which tries to obtain substitutions that can be used to remove conditions in rules.

\subsubsection{Extra variables check: $\ProcExtraVars$}

An obvious reason for non-confluence is the presence of extra variables in rules. For instance, a rule $\lhsr\to x$ for some \emph{extra}
variable $x\notin\Var(\lhsr)$  `produces' non-confluence: the following peak is always possible: $x'\leftrew{}\lhsr\rew{}x$
for some fresh variable $x'\notin\Var(\lhsr)\cup\{x\}$, but both $x$ and $x'$ are irreducible.
In general, as a simple generalization of this fact,
given a  \gtrs{} $\cR=(\Symbols,\SPredicates,\mu,H,R)$, if there is
a \emph{feasible} conditional rule $\lhsr\to\rhsr\IF\gencond\in R$, and a variable
$x\in\Var(\rhsr)-\Var(\lhsr,\gencond)$
such that $p\in\Pos_x(\rhsr)$ and for all $q<p$,  $\rootTerm(\rhsr|_q)$ is \emph{not a defined symbol},
then $\cR$ is not (locally, strongly) confluent. Hence, we define the following processor:

\[\begin{array}{rcl}
\ProcExtraVars(\CRproblem{\cR})=\noTree\\
\ProcExtraVars(\WCRproblem{\cR})=\noTree\\
\ProcExtraVars(\SCRproblem{\cR})=\noTree
\end{array}\]
if $\cR$ contains a rule as above. Then, $\ProcExtraVars$ is complete and (trivially) sound.

\subsubsection{Simplification: $\ProcSimplify$}
\label{SecPreprocessing}

The following simplifications of \gtrs{s}
are often useful in proofs of confluence problems.
\begin{enumerate}
\itemsep=0.9pt
\item
\emph{Removing trivial rules.}
All rules $ t\to t$ or $ t\to t\IF\gencond$ for some term $t$ are \emph{removed}.
\item
\emph{Removing trivial conditions.}
Conditions $t\cto t$ in the conditional part $\gencond$ of rules $\lhsr\to\rhsr\IF\gencond$ of J-, O-, or SE-\csctrs{s} are \emph{removed}.
\item \emph{Removing infeasible conditional rules. }
Conditional rules $\lhsr\to\rhsr\IF\gencond$ with an infeasible condition $c$ are \emph{removed}, as they
will not be applied in reduction steps.
\item \emph{Removing ground atoms.} Ground atoms $B$ occurring in \emph{feasible} conditions $\gencond$ in clauses $A\IFhorn\gencond\in H$ or rules $\lhsr\to\rhsr\IF\gencond\in R$ can be \emph{removed} without affecting
the role of the so-simplified clause or rule in computations.
\end{enumerate}
These are applied as much as possible (to each rule in the input system)
by means of a \emph{simplifying processor}
$\ProcSimplify$:
\[\begin{array}{rcl}
\ProcSimplify(\CRproblem{\cR}) & = & \{\CRproblem{\cR'}\}\\
\ProcSimplify(\WCRproblem{\cR}) & = & \{\WCRproblem{\cR'}\}\\
\ProcSimplify(\SCRproblem{\cR}) & = & \{\SCRproblem{\cR'}\}
\end{array}\]
where $\cR'=(\Symbols,\SPredicates,\mu,H',R')$ is obtained by using the previous transformations to obtain $H'$ and $R'$ from $H$ and $R$, respectively.
Since $\rew{\cR}$ and $\rew{\cR'}$ coincide, (local, strong) confluence of
$\cR$ and $\cR'$ also coincide.
Thus, $\ProcSimplify$ is \emph{sound} and \emph{complete}.

\begin{example}\label{COPS409Ex}
  The following example (\verb$#409$ in COPS\footnote{Confluence Problems database, see \url{https://cops.uibk.ac.at/}}) displays
  an \emph{oriented} CTRS.\footnote{In the following, in order to keep a close connection with the original sources,
rather than call them \gtrs{s}, we use TRS, CTRS, \csctrs{}, etc., when citing external
examples.}
\begin{eqnarray}
 \fS{b} & \to &  \fS{b}\label{COPS409Ex_rule1}\\
  \fS{g}( \fS{s}(x)) & \to &  x\label{COPS409Ex_rule2}\\
  \fS{h}(\fS{s}(x)) & \to &  x\label{COPS409Ex_rule3}\\
  \fS{f}(x,y) & \to &    \fS{g}(\fS{s}(x)) \IF  \fS{c}(\fS{g}(x)) \cto   \fS{c}(\fS{a})\label{COPS409Ex_rule4}\\
 \fS{f}(x,y) & \to &  \fS{h}(\fS{s}(x)) \IF \fS{c}(\fS{h}(x)) \cto \fS{c}(\fS{a})\label{COPS409Ex_rule5}
\end{eqnarray}
Since rule (\ref{COPS409Ex_rule1}) fits the first case above, we can remove it.
Thus, we have:
$\ProcSimplify(\CRproblem{\cR})=\{\CRproblem{\cR'}\}$
where $\cR'$ consists of the rules $(\ref{COPS409Ex_rule2}),\ldots,(\ref{COPS409Ex_rule5})$.
\end{example}

\subsubsection{Inlining: $\ProcInlining$}
\label{SecProcInlining}

As in \cite[Definition
	9.4 \& Lemma 9.5]{Sternagel_ReliableConfluenceAnalysisOfConditionalTermRewriteSystems_PhD17},
the so-called \emph{inlining of rules} is useful to shrink the conditions of
\emph{O-rules} in a \gtrs{} $\cR=(\Symbols,\SPredicates,\mu,H,R)$.
By an \emph{O-rule} we mean a rule $\alpha:\lhsr\to\rhsr\IF\gencond\in R$ where $\gencond$ consists of
conditions
which are given the usual \emph{reachability} semantics,
by explicitly writing $s\rews{}t$, or, in an indirect way, as $s\cto{}t$
with $\cto$
defined by a single
clause $x\cto y\IF x\rews{}y$ in $H$.
For simplicity, in
the remainder of the section we assume that the last format is used.

\begin{definition}[Inlining]
\label{DefInlining}
Let
$\alpha:\lhsr\to\rhsr\IF s_1\cto t_1, \cdots,s_n\cto t_n$ be an O-rule and
$t_i=x$ for some variable
\begin{eqnarray}
x\notin\Var(\lhsr,s_i,t_1,\ldots,t_{i-1},t_{i+1},\ldots t_n)\cup\NVar{\mu}(\rhsr,s_1\ldots,s_n),
\label{LblVariableRestrictionInlining}
\end{eqnarray}
and $1\leq i\leq n$. Let $\sigma=\{x\mapsto s_i\}$. The \emph{inlining} of the $i$-th condition of $\alpha$ with $x$ is
\begin{eqnarray}
\alpha_{x,i}:\lhsr\to\sigma(\rhsr)\IF \sigma(s_1)\cto t_1, \cdots, \sigma(s_{i-1})\cto
	t_{i-1}, \sigma(s_{i+1})\cto t_{i+1}, \cdots,\sigma(s_n)\cto t_n
\end{eqnarray}
Given a \gtrs{} $\cR=(\Symbols,\SPredicates,\mu,H,R\uplus\{\alpha\})$ where $\alpha$ is an O-rule,
the
 \emph{inlining} of  the $i$-th condition of $\alpha$ in $\cR$ with $x$ is
 $\cR_{\alpha,x,i}=(\Symbols,\SPredicates,\mu,H,R\uplus\{\alpha_{x,i}\})$.
\end{definition}
For the sake of readability, the proof of the following result is in  Appendix \ref{ApProofPropInliningGTRSs}.
\begin{proposition}
\label{PropInliningGTRSs}
Let $\cR=(\Symbols,\SPredicates,\mu,H,R)$ be a \gtrs{},
$\alpha\in R$,
 $i$, and $x$ as in Definition \ref{DefInlining}. Let $s$ and $t$ be terms.
\begin{enumerate}
\item  If $s\rew{\cR}t$,
then
$s\rews{\cR_{\alpha,x,i}}t$.
\item If $s\rew{\cR_{\alpha,x,i}}t$,
then
$s\rew{\cR}t$.
\end{enumerate}
\end{proposition}
\begin{corollary}
\label{CoroInliningGTRSs}
Let $\cR=(\Symbols,\SPredicates,\mu,H,R)$ be a \gtrs{},
$\alpha\in R$,
 $i$, and $x$ as in Definition \ref{DefInlining}. Then, $\rews{\cR}$ and $\rews{\cR_{\alpha,x,i}}$ coincide.
\end{corollary}
Corollary \ref{CoroInliningGTRSs} entails that confluence of $\cR$ and $\cR_{\alpha,x,i}$ coincide.
By Proposition \ref{PropInliningGTRSs}, local (resp.\ strong)
confluence of $\cR$ implies local (resp.\ strong)
confluence of $\cR_{\alpha,x,i}$.
In general, though, the opposite direction does \emph{not} hold.

\begin{example}
\label{ExHindleyCTRS}
Consider the following O-CTRS $\cR$:
\begin{eqnarray}
\fS{b} & \rew{} & \fS{a}\label{ExHindleyCTRS_rule1}\\
\fS{b} & \rew{} & x \IF \fS{c}\cto x\label{ExHindleyCTRS_rule2}\\
\fS{c} & \rew{} & \fS{b}\label{ExHindleyCTRS_rule3}\\
\fS{c} & \rew{} & \fS{d}\label{ExHindleyCTRS_rule4}
\end{eqnarray}
Note that $\cR$ is \emph{not} locally confluent:
\[
\fS{a}\leftrew{(\ref{ExHindleyCTRS_rule1})} \fS{b}\rew{(\ref{ExHindleyCTRS_rule2})}\fS{d}
\]
because the conditional part of (\ref{ExHindleyCTRS_rule2}) is satisfied by
$\fS{c}\rews{\cR}\fS{d}$.
However, the \emph{inlining} of (\ref{ExHindleyCTRS_rule2}) yields the rule
$\fS{b}\rew{}\fS{c}$ which, together with
(\ref{ExHindleyCTRS_rule1}),
(\ref{ExHindleyCTRS_rule3}), and (\ref{ExHindleyCTRS_rule4})
form the well-known \emph{locally confluent} TRS
$\cR_{(\ref{ExHindleyCTRS_rule2}),x,1}=\{\fS{b} \rew{} \fS{a},\fS{b} \rew{} \fS{c},\fS{c} \rew{} \fS{b},\fS{c} \rew{} \fS{d}\}$.
\end{example}
As for $\ProcSimplify$, with $\ProcInlining$ we assume that all rules in the input system $\cR$ have
been inlined as much as possible to obtain $\cR'$. Then we have
\[\begin{array}{rcl}
\ProcInlining(\CRproblem{\cR}) & = & \{\CRproblem{\cR'}\}\\
\ProcInlining(\WCRproblem{\cR}) & = & \{\WCRproblem{\cR'}\}\\
\ProcInlining(\SCRproblem{\cR}) & = & \{\SCRproblem{\cR'}\}
\end{array}\]
By Corollary \ref{CoroInliningGTRSs}, $\ProcInlining$ is sound and complete for confluence problems
$\CRproblem{\cR}$.
By Proposition \ref{PropInliningGTRSs}, $\ProcInlining$ is \emph{complete} (but in general \emph{not sound}, see Example \ref{ExHindleyCTRS}) for local (resp.\ strong) confluence problems
$\WCRproblem{\cR}$ ($\SCRproblem{\cR}$).

\subsection{Modular decomposition: $\ProcModularDecomp$}
\label{SectionModularDecomposition}

The \emph{decomposition} of a confluence problem $\CRproblem{\cR}$ into
 two problems $\CRproblem{\cR_1}$ and $\CRproblem{\cR_2}$,
 by splitting up the input \gtrs{} as $\cR=\cR_1\cup\cR_2$ for appropriate
 components or \emph{modules} $\cR_1$ and $\cR_2$,
can be useful in breaking down confluence problems into smaller ones.
On this basis, in this section we discuss a processor $\ProcModularDecomp$
implementing this approach.
Soundness and completeness of $\ProcModularDecomp$ can be proved
by using existing results about
\emph{modularity} of (local, strong) confluence.
As modularity of \gtrs{s} has not been investigated yet,
in this section we focus on TRSs (as particular \gtrs{s}), see \cite{Gramlich_ModularityInTermRewritingRevisited_TCS12}
and also \cite[Section 8]{Ohlebusch_AdvTopicsTermRew_2002}, which we follow here.
Our discussion would easily generalize to CTRSs, for which a number of modularity
results for confluence are also available, see \cite{Middeldorp_ModularPropertiesOfConditionalTermRewritingSystems_IC93}.
In general, a property $\cP$ of rewriting-based systems is \emph{modular}
if for all systems $\cR_1$ and $\cR_2$ satisfying $\cP$,
the union $\cR$ of $\cR_1$ and $\cR_2$ also satisfies $\cP$, see \cite[Definition 8.1.1]{Ohlebusch_AdvTopicsTermRew_2002}.
Usually, properties can be proved modular only if  $\cR_1$ and
$\cR_2$ fulfill
particular \emph{combination} conditions that are parameterized by $\fS{Comb}$:
we write $\combinationOfModules{\cR_1}{\cR_2}$ to express that
 $\cR_1$ and $\cR_2$ satisfy the requirements of a particular combination
$\fS{Comb}$ of modules.
\begin{definition}[Modularity of confluence]
\label{DefModularityOfConfluence}
(Local, Strong) Confluence is called \emph{modular with respect to a given combination
$\fS{Comb}$} of TRSs (\emph{$\fS{Comb}$-modular} for short)
if for all TRSs
$\cR_1=(\Symbols_1,R_1)$ and $\cR_2=(\Symbols_2,R_2)$
satisfying the condition $\combinationOfModules{\cR_1}{\cR_2}$, the following holds:
\emph{if} both $\cR_1$ and $\cR_2$ are (locally, strongly) confluent, then the union
$\cR_1\cup\cR_2=(\Symbols_1\cup\Symbols_2,R_1\cup R_2)$
 is also (locally, strongly) confluent.
\end{definition}
For TRSs $\cR$, processor $\ProcModularDecomp$ tries to find such a decomposition:
\[\begin{array}{rcl}
\ProcModularDecomp(\CRproblem{\cR}) & = & \{\CRproblem{\cR_1},\CRproblem{\cR_2}\}\\
\ProcModularDecomp(\WCRproblem{\cR}) & = & \{\WCRproblem{\cR_1},\WCRproblem{\cR_2}\}\\
\ProcModularDecomp(\SCRproblem{\cR}) & = & \{\SCRproblem{\cR_1},\SCRproblem{\cR_2}\}
\end{array}
\]
if $\cR=\cR_1\cup\cR_2$, $\combinationOfModules{\cR_1}{\cR_2}$ holds and
(local, strong) confluence is $\fS{Comb}$-modular.
This definition guarantees \emph{soundness} of $\ProcModularDecomp$ on (local, strong) confluence problems as it is implied by the modularity of the corresponding property.
Thus,  $\ProcModularDecomp$ is \emph{sound} on (local, strong) confluence problems
if $\combinationOfModules{\cR_1}{\cR_2}$ holds and
(local, strong) confluence is $\fS{Comb}$-modular.

\subsubsection{Modular combinations and modularity results for \trs{s}}
\label{SecModularityResults}

In the literature, a number of \emph{combinations}
$\combinationOfModules{\cR_1}{\cR_2}$
of TRSs $\cR_1$ and $\cR_2$
have been considered to prove modularity.
In particular,
\begin{itemize}
\item \emph{disjoint} combinations
(where $\cR_1$ and $\cR_2$ share no function symbol \cite{Toyama_OnTheCRPropDirSumTRS_JACM87}),
\item \emph{constructor-sharing} combinations
(where $\cR_1$ and $\cR_2$ may share constructor symbols only \cite{KurOhu_ModularityOfSimpleTerminationOfTermRewritingSystemsWithSharedConstructors_TCS92}),
\item \emph{composable} combinations
(where $\cR_1$ and $\cR_2$ may share constructor symbols and also defined symbols
provided that they also share the rules defining them \cite{MidToy_CompletenessOfCombinationsOfConstructorSystems_JSC93}).
\end{itemize}
See also  \cite[Definition 8.1.4]{Ohlebusch_AdvTopicsTermRew_2002} for definitions of all these
combinations of TRSs,
which we  refer as
$\disjointComb{\cR_1}{\cR_2}$,
$\constructorSharingComb{\cR_1}{\cR_2}$,
and
$\composableComb{\cR_1}{\cR_2}$, respectively.
Note that
\begin{eqnarray}
\disjointComb{\cR_1}{\cR_2}\Rightarrow\constructorSharingComb{\cR_1}{\cR_2}\Rightarrow\composableComb{\cR_1}{\cR_2}\label{LblHierarchyOfCombinations}
\end{eqnarray}
see \cite[Figure 8.2]{Ohlebusch_AdvTopicsTermRew_2002}.
In Table \ref{TableModularityResultsForTRSs} we show an excerpt of
the results displayed
in \cite[Table 8.1]{Ohlebusch_AdvTopicsTermRew_2002}
regarding modularity of confluence, local confluence, and strong confluence of TRSs.
In this table, we use the notion of \emph{layer-preserving TRS} (LP) \cite[Definition 5.5]{Ohlebusch_ModularPropertiesOfComposableTermRewritingSystems_JSC95}:
in a composable combination, i.e., $\composableComb{\cR_1}{\cR_2}$ holds,
let $\cB=\Symbols_1\cap\Symbols_2$ be the set of \emph{shared function symbols}
and, for $i\in\{1,2\}$, $\cA_i=\Symbols_i-\cB$ be the \emph{alien} symbols for $\cR_{3-i}$.
Let $i\in\{1,2\}$. Then, $\cR_i=(\Symbols_i,R_i)$ is
called \emph{layer-preserving} if for all $\lhsr\to\rhsr\in\cR_i$,
we have $\rootTerm(\rhsr)\in\cA_i$ whenever $\rootTerm(\lhsr)\in\cA_i$.
A \emph{constructor-sharing} union is \emph{layer-preserving} if $\cR_1$ and $\cR_2$ contain neither collapsing rules nor constructor-lifting rules (i.e., rules $\lhsr\to\rhsr$
such that $\rootTerm(\rhsr)$ is a shared constructor).

\begin{table}[!h]
\caption{Excerpt of \cite[Table 8.1]{Ohlebusch_AdvTopicsTermRew_2002}. Here L is \emph{linearity},
 LL is \emph{left-linearity}, and LP is \emph{layer preservation}}
\begin{center}
\scalebox{0.95}{
\begin{tabular}{llll}
Property & Disjoint union & Constructor-sharing & Composable\\
\hline
Confluence
&
\cite[Coro.\ 4.1]{Toyama_OnTheCRPropDirSumTRS_JACM87}
& +LP
\cite[Coro.\ 5.11]{Ohlebusch_OnTheModularityOfConfluenceOfConstructorSharingTRS_CAAP94}
& +LP
\cite{Ohlebusch_ModularPropertiesOfComposableTermRewritingSystems_PhD94}
\\
&& +LL \cite{RaoVui_OperationalAndSemanticEquivalenceBetweenRecursivePrograms_JACM80},
see \cite[Coro.\ 8.6.38(1)]{Ohlebusch_AdvTopicsTermRew_2002}
\\
Local Conf.\ &
\cite{Middeldorp_ModularPropertiesOfTermRewritingSystems_PhD90}
&
\cite{Middeldorp_ModularPropertiesOfTermRewritingSystems_PhD90}
&
\cite{Middeldorp_ModularPropertiesOfTermRewritingSystems_PhD90}
\\
Strong Conf. &
+L \cite{Huet_ConfReduc_JACM80}
&
+L \cite{Huet_ConfReduc_JACM80}
&
+L \cite{Huet_ConfReduc_JACM80}
\\
\hline
\end{tabular} }
\end{center}
\label{TableModularityResultsForTRSs}
\end{table}
\begin{remark}[Comments on Table \ref{TableModularityResultsForTRSs}]
\label{RemCommentsOnTableModularityResultsForTRSs}
As remarked in \cite[Sections 8.2.1 \& 8.6.3]{Ohlebusch_AdvTopicsTermRew_2002},
for composable combinations (hence for disjoint unions and constructor-sharing combinations, see (\ref{LblHierarchyOfCombinations})),
the following \emph{non interfering}  property \cite{Middeldorp_ModularPropertiesOfTermRewritingSystems_PhD90} holds:
\begin{eqnarray}
\CP(\cR_1\cup\cR_2) & = & \CP(\cR_1)\cup\CP(\cR_2)\label{LblCPsOfUnionIsUnionOfCPs}
\end{eqnarray}
that is,
the set of critical pairs of the union is the union of the critical pairs of the components.

\medskip
Accordingly, \cite{Ohlebusch_AdvTopicsTermRew_2002} makes the following observations that justify the last two rows in Table \ref{TableModularityResultsForTRSs}:
\begin{itemize}
\item
Middeldorp proved that \emph{local confluence is
modular} for
disjoint unions of TRSs, see, e.g.,
\cite[Theorem 2.4]{Middeldorp_ModularPropertiesOfConditionalTermRewritingSystems_IC93},
originally in \cite{Middeldorp_ModularPropertiesOfTermRewritingSystems_PhD90}.
Ohlebusch observes that local confluence is modular
for any combination of TRSs satisfying (\ref{LblCPsOfUnionIsUnionOfCPs}) \cite[page 249, penultimate paragraph]{Ohlebusch_AdvTopicsTermRew_2002}.
Thus, local confluence is modular for constructor-sharing and composable combinations too, see also \cite[Corollary 8.6.41(1)]{Ohlebusch_AdvTopicsTermRew_2002}.
\item  
As explained in the paragraph below \cite[Example 8.2.2]{Ohlebusch_AdvTopicsTermRew_2002}, the results about
\emph{modularity of strong confluence}, not explicit
in \cite{Huet_ConfReduc_JACM80}, are a consequence of
\cite[Lemma 3.2]{Huet_ConfReduc_JACM80} (\emph{A linear TRS is strongly closed iff it is strongly confluent}) and the \emph{non interfering}  property (\ref{LblCPsOfUnionIsUnionOfCPs}) for composable (hence disjoint, sharing constructor) combinations.
\end{itemize}
\end{remark}
\subsubsection{Soundness of $\ProcModularDecomp$}

As a consequence of the discussion in Section \ref{SecModularityResults},
(see Table \ref{TableModularityResultsForTRSs} and
Remark \ref{RemCommentsOnTableModularityResultsForTRSs}),
given $\cR$, $\cR_1$, and $\cR_2$ such that $\cR=\cR_1\cup\cR_2$,
$\ProcModularDecomp$ is \emph{sound} for
\begin{itemize}
\item
$\CRproblem{\cR}$
if
(i) $\disjointComb{\cR_1}{\cR_2}$ holds, or
(ii) $\constructorSharingComb{\cR_1}{\cR_2}$ holds and $\cR$ is \emph{left-linear},
or
(iii) $\constructorSharingComb{\cR_1}{\cR_2}$ or $\composableComb{\cR_1}{\cR_2}$ holds and $\cR_1$ and $\cR_2$ are \emph{layer preserving}.
\item
$\WCRproblem{\cR}$ if $\composableComb{\cR_1}{\cR_2}$ (and hence $\disjointComb{\cR_1}{\cR_2}$ or $\constructorSharingComb{\cR_1}{\cR_2}$) holds.
\item
$\SCRproblem{\cR}$ if
$\cR$ is linear, and $\disjointComb{\cR_1}{\cR_2}$ or $\constructorSharingComb{\cR_1}{\cR_2}$
or $\composableComb{\cR_1}{\cR_2}$ holds.
\end{itemize}

\begin{example}\label{Ex4_IWC23}
Consider the following TRS \cite[Example 4]{Oostrom_TheZpropertyForLeftLinearTermRewritingViaConvectiveContextSensitiveCompleteness_IWC23}:
\begin{eqnarray}
\fS{nats} & \to & \fS{from}(\fS{0})\label{Ex4_IWC23_rule1}\\
\fS{inc}(x\fS{:}y)& \to &  \fS{s}(x)\fS{:}\fS{inc}(y)\label{Ex4_IWC23_rule2}\\
\fS{hd}(x\fS{:}y) & \to &  x\label{Ex4_IWC23_rule3}\\
\fS{tl}(x\fS{:}y)  & \to &  y\label{Ex4_IWC23_rule4}\\
\fS{from}(x) & \to &  x\fS{:}\fS{from}(\fS{s}(x))\label{Ex4_IWC23_rule5}\\
\fS{inc}(\fS{tl}(\fS{from}(x))) & \to & \fS{tl}(\fS{inc}(\fS{from}(x)))\label{Ex4_IWC23_rule6}
\end{eqnarray}
With $\cR_1=\{(\ref{Ex4_IWC23_rule3})\}$ and
$\cR_2=\{(\ref{Ex4_IWC23_rule1}),(\ref{Ex4_IWC23_rule2}),(\ref{Ex4_IWC23_rule4}),
(\ref{Ex4_IWC23_rule5}),(\ref{Ex4_IWC23_rule6})\}$, with $\_\fS{:}\_$ the only shared constructor
symbol, $\constructorSharingComb{\cR_1}{\cR_2}$ holds.
Since, $\cR$ is \emph{left-linear},
$
\ProcModularDecomp(\CRproblem{\cR})  =  \{\CRproblem{\cR_1},\CRproblem{\cR_2}\}
$.
\end{example}

\subsubsection{Completeness of $\ProcModularDecomp$}

Regarding modularity of disjoint unions $\cR=\cR_1\cup\cR_2$,
where $\disjointComb{\cR_1}{\cR_2}$ holds,
\cite[Corollary 4.1]{Toyama_OnTheCRPropDirSumTRS_JACM87} proves that
\emph{$\cR$ is confluent iff both $\cR_1$ and $\cR_2$ are}.
This entails that $\ProcModularDecomp$ is \emph{complete} on confluence problems for disjoint unions.

\medskip
Also, $\ProcModularDecomp$ is \emph{complete} on
local confluence problems for disjoint unions: \cite[Theorem 2.4]{Middeldorp_ModularPropertiesOfConditionalTermRewritingSystems_IC93} establishes
modularity of local confluence
and, according to \cite[Definition 2.2]{Middeldorp_ModularPropertiesOfConditionalTermRewritingSystems_IC93},
local confluence is modular for a disjoint union of
TRSs if the following equivalence holds: $\cR_1\cup\cR_2$ is locally confluent iff both $\cR_1$ and $\cR_2$ are locally confluent.
The point is that, for disjoint unions, joinability of critical pairs in $\CP(\cR_i)$ for $i\in\{1,2\}$
cannot depend on reductions using $\cR_{3-i}$.
Since (\ref{LblCPsOfUnionIsUnionOfCPs}) holds for disjoint unions, it follows that local confluence of $\cR=\cR_1\cup\cR_2$ (i.e., joinability of all critical pairs in $\CP(\cR_1\cup\cR_2)$)
implies local confluence of both $\cR_1$ and $\cR_2$.
For composable combinations (hence for constructor sharing combinations), Ohlebusch provides a similar treatment in \cite{Ohlebusch_ModularPropertiesOfComposableTermRewritingSystems_JSC95}:
\cite[Proposition 5.3(1)]{Ohlebusch_ModularPropertiesOfComposableTermRewritingSystems_JSC95} establishes modularity of local confluence for composable combinations, and, for composable combinations,
$\cR_1\cup\cR_2$ is locally confluent iff both $\cR_1$ and $\cR_2$ are locally confluent
\cite[Definition 3.2]{Ohlebusch_ModularPropertiesOfComposableTermRewritingSystems_JSC95}.
Similarly, since strong confluence is characterized by the strong joinability of critical pairs,
$\ProcModularDecomp$ is also \emph{complete} on
strong confluence problems for disjoint unions.

\subsection{Processor for local/strong confluence problems: $\ProcHuetExtended$}
\label{SecProcessorForLocalStrongConfluenceProblems}

Processor $\ProcHuetExtended$ treats local and strong confluence problems
as (strong) joinability of conditional critical and variable pairs.
\paragraph{Local confluence.} Extended conditional critical pairs $\ECCP(\cR)$ enable the following characterizations of local confluence of
\gtrs{s} (and, in particular, of
\cstrs{s}, \ctrs{s}, \csctrs{s}, etc.), thus
extending the well-known result for TRSs by Huet \cite[Lemma 3.1]{Huet_ConfReduc_JACM80}.

\begin{theorem}[Local confluence of \gtrs{s}]
\label{TheoLocalConfluenceCSCTRSsWithExtendedCCPs}
Let $\cR$ be a \gtrs{}.
Then,
\begin{enumerate}
\item\label{TheoLocalConfluenceCSCTRSsWithExtendedCCPs_CSCTRS}
 $\cR$ is locally confluent iff each $\pi\in\ECCP(\cR)$ is joinable \cite[Theorem 62]{Lucas_LocalConferenceOfConditionalAndGeneralizedTermRewritingSystems_JLAMP24}.
\item\label{TheoLocalConfluenceCSCTRSsWithExtendedCCPs_SECSCTRS}
 If $\cR$ is an SE-\csctrs{} such that
 ($\dagger$) for all
$\langle s,t\rangle\IF x\rew{}x',\gencond\in\CVP(\cR)$,
and $u\cto v\in\gencond$, $x\notin\NVar{\mu}(u)\cup\NVar{\mu}(v)$,
then, $\cR$ is locally confluent iff each $\pi\in\pCCP(\cR)\cup\iCCP(\cR)$ is joinable
\cite[Corollary 63]{Lucas_LocalConferenceOfConditionalAndGeneralizedTermRewritingSystems_JLAMP24}.
\item\label{TheoLocalConfluenceCSCTRSsWithExtendedCCPs_CSTRS}
 If $\cR$ is a CS-TRS,
then $\cR$ is locally confluent iff each $\pi\in\CP(\cR,\mu)\cup\LHCP(\cR,\mu)$ is joinable
\cite[Theorem 30]{LucVitGut_ProvingAndDisprovingConfluenceOfContextSensitiveRewriting_JLAMP22}.
\end{enumerate}
\end{theorem}
Accordingly, for \gtrs{s} $\cR=(\Symbols,\SPredicates,\mu,H,R)$,
\[\begin{array}{r@{\:}c@{\:}l}
\ProcHuetExtended(\WCRproblem{\cR}) & = & \{\JOproblem{\cR,\pi_1},\ldots,\JOproblem{\cR,\pi_n}\}, \text{where } \\
  &&\hspace*{-10mm} \{\pi_1,\ldots,\pi_n\} =
\left \{
\begin{array}{ll}
\CP(\cR) & \text{if $\cR$ is a TRS}\\
\CP(\cR,\mu)\cup\LHCP(\cR,\mu) & \text{if $\cR$ is a CS-TRS}\\
\pCCP(\cR)\cup\iCCP(\cR) & \text{if $\cR$ is an SE-\csctrs{} satisfying ($\dagger$)}\\
\ECCP(\cR) & \text{otherwise}
\end{array}
\right .
\end{array}
\]
\begin{example}\label{ExCOPS_387_CS_CTRS_PHuet}
Consider the CTRS $\cR$ in Example \ref{ExCOPS_387_CS_CTRS} with $\ECCP(\cR)=\{(\ref{ExCOPS_387_CS_CTRS_pCCP}),(\ref{ExCOPS_387_CS_CTRS_CVP})\}$ (see Example \ref{ExCOPS_387_CS_CTRS_pCCPs_iCCPs_CVPs}).
Thus, we have $\ProcHuetExtended(\WCRproblem{\cR})  = \{\JOproblem{\cR,(\ref{ExCOPS_387_CS_CTRS_pCCP})},\JOproblem{\cR,(\ref{ExCOPS_387_CS_CTRS_CVP})}\}$.
\end{example}
\begin{example}[Continuing Examples \ref{ExCOPS_387_CS_CTRS} and \ref{ExCOPS_387_CS_CTRS_Rbot_NoECCPs}]
\label{ExCOPS_387_CS_CTRS_PHuet_Rbot_PKnuthBendix}
For the \csctrs{}
$\cR_\bot$ in Example \ref{ExCOPS_387_CS_CTRS},
$\ECCP(\cR_\bot)=\emptyset$. Thus,
$\ProcHuetExtended(\WCRproblem{\cR_\bot})=\emptyset$.
\end{example}
Note that,  by
Theorem \ref{TheoLocalConfluenceCSCTRSsWithExtendedCCPs}, $\ProcHuetExtended$ is \emph{sound and complete} on local confluence problems.

\paragraph{Strong confluence.} Linear TRSs whose critical pairs are \emph{strongly joinable} are strongly confluent,
\cite[Lemma 6.3.3]{BaaNip_TermRewAllThat_1998}.
Accordingly, for \emph{linear} TRSs $\cR$,
\[\begin{array}{r@{\:}c@{\:}l}
\ProcHuetExtended(\SCRproblem{\cR}) & = & \{\SJOproblem{\cR,\pi_1},\ldots,\SJOproblem{\cR,\pi_n}\}, \text{where }
  \{\pi_1,\ldots,\pi_n\} =
\CP(\cR)
\end{array}
\]
This processor is \emph{sound} and \emph{complete}
on strong confluence problems for linear TRSs: the existence of a
non-strongly-joinable conditional pair witnesses non-strong-joinabilty.

\medskip
\noindent
Table~\ref{TableUseOfCleansingAndModularProcessorsInTheConfluenceFramework}~summarizes the use of $\ProcExtraVars$, $\ProcSimplify$, $\ProcInlining$,
$\ProcModularDecomp$, and $\ProcHuetExtended$
 in the confluence framework.

\begin{table}[ht]
\vspace*{-4mm}
\caption{Use of $\ProcExtraVars$, $\ProcSimplify$, $\ProcInlining$, $\ProcModularDecomp$ and $\ProcHuetExtended$
 in the confluence framework}
\begin{center}
\scalebox{0.94}{
\begin{tabular}
{l@{~~}c@{~\:}c@{~\:}c@{~\:}c@{~\:}c@{~\:}c@{~\:}c@{~\:}c@{~\:}c@{~\:}c@{~\:}c@{~\:}c@{~\:}|c@{~\:}|c@{~\:}c}
Problem
& $\ProcExtraVars$
 & $\ProcSimplify$
 & $\ProcInlining$
 & $\ProcModularDecomp$
& $\ProcHuetExtended$
\\
 \hline
CR
& $\checkmark$
& $\checkmark$
& $\checkmark$
& $\checkmark$
 &
\\
WCR
& $\checkmark$
& $\checkmark$
& $\checkmark$
& $\checkmark$
& $\checkmark$
\\
SCR
& $\checkmark$
& $\checkmark$
& $\checkmark$
& $\checkmark$
& $\checkmark$
\\
 \hline
 TRSs
 & $\checkmark$
& $\checkmark$
&
& $\checkmark$
 & $\checkmark$
\\
 CS-TRSs
& $\checkmark$
& $\checkmark$
&
  &
& $\checkmark$
\\
CTRSs
& $\checkmark$
& $\checkmark$
& $\checkmark$
 &
& $\checkmark$
\\
CS-CTRSs
& $\checkmark$
& $\checkmark$
& $\checkmark$
 &
 & $\checkmark$
\\
\gtrs{s}
& $\checkmark$
& $\checkmark$
& $\checkmark$
 &
& $\checkmark$
 \\
 \hline
\end{tabular} }
\end{center}
{\small{
\noindent $\ProcExtraVars$, $\ProcSimplify$, and $\ProcHuetExtended$ are sound and complete. $\ProcInlining$ is complete, but it is sound on CR problems only.
$\ProcModularDecomp$ is sound;
if $\ProcModularDecomp$ is
used with \emph{disjoint unions} of TRSs,
then it is \emph{complete} on all problems it applies to;
if $\ProcModularDecomp$ is used with \emph{composable combinations} of TRSs (hence, disjoint unions and constructor-sharing combinations), then it is \emph{complete} on  WCR problems. } }
\label{TableUseOfCleansingAndModularProcessorsInTheConfluenceFramework}
\end{table}

\begin{table}[!hb]
\vspace*{-3mm}
\caption{Use of Processors for Confluence Problems in the confluence framework}
\begin{center}
\scalebox{0.95}{
\begin{tabular}
{l@{~~}c@{~\:}c@{~\:}c@{~\:}c@{~\:}c@{~\:}c@{~\:}c@{~\:}c@{~\:}c@{~\:}c@{~\:}c@{~\:}c@{~\:}|c@{~\:}|c@{~\:}c}
Problem
   & $\ProcTrU$
 & $\ProcTrUconf$
 & $\ProcOrthogonality$
 & $\ProcConfluence$
 & $\ProcLocalConfluence$
 & $\ProcStrongConfluence$
 & $\ProcCanJoinability$
 & $\ProcCnvJoinability$
 & $\ProcCanCSR$
  & $\ProcKnuthBendix$
 \\
 \hline
CR
 & $\checkmark$
 & $\checkmark$
 & $\checkmark$
 &
 & $\checkmark$
 & $\checkmark$
 & $\checkmark$
 & $\checkmark$
 & $\checkmark$
 & $\checkmark$
\\
WCR
 & $\checkmark$
 & $\checkmark$
 & $\checkmark$
 & $\checkmark$
 &
 & $\checkmark$
 & $\checkmark$
 & $\checkmark$
 & $\checkmark$
 &
\\
SCR
 &
&
 &
 & $\checkmark$
 & $\checkmark$
 &
 &
 &
 &
 &
 \\
 \hline
 TRSs
 &
 &
 & $\checkmark$
 & $\checkmark$
 & $\checkmark$
 & $\checkmark$
 & $\checkmark$
 & $\checkmark$
 & $\checkmark$
 & $\checkmark$
\\
 CS-TRSs
 &
 &
 & $\checkmark$
 & $\checkmark$
 & $\checkmark$
 & $\checkmark$
 &
 &
 &
 & $\checkmark$
\\
CTRSs
 & $\checkmark$
 & $\checkmark$
 & $\checkmark$
 & $\checkmark$
 & $\checkmark$
 & $\checkmark$
 &
 &
 &
 & $\checkmark$
\\
CS-CTRSs
 &
 &
 &
 & $\checkmark$
 & $\checkmark$
 & $\checkmark$
 &
 &
 &
 & $\checkmark$
\\
\gtrs{s}
 &
&
 &
 & $\checkmark$
 & $\checkmark$
 & $\checkmark$
 &
 &
 &
 & $\checkmark$
\\
\hline

Sound
 & $\checkmark$
 & $\checkmark$
 & $\checkmark$
 &
 &
 & $\checkmark$
 & $\checkmark$
 & $\checkmark$
 & $\checkmark$
 & $\checkmark$
\\
 Complete
 &
 &
 & $\checkmark$
 &
 & $\checkmark$
 &
 &
 &
 &
 & $\checkmark$
\end{tabular} }
\\[5pt]
$\ProcConfluence$ is \emph{sound} on WCR problems,
and \emph{complete} on SCR-problems.
\end{center}
\label{TableUseOfConfluenceProcessorsInTheConfluenceFramework}
\end{table}

\subsection{Processors for confluence problems}

In this section we discuss a number of processors to be used with confluence problems.
Processors $\ProcTrU$
and
$\ProcTrUconf$ treat confluence problems for \ctrs{s} $\cR$ by transforming them into
\trs{s} 
$\cU(\cR)$
and
$\Uconf(\cR)$ and then solving the corresponding confluence problem.
Following the results obtained by Huet for left-linear TRSs without critical pairs
\cite{Huet_ConfReduc_JACM80}, processor $\ProcOrthogonality$ exploits (variants of) \emph{orthogonality} investigated for \cstrs{s} and \ctrs{s} to treat confluence problems.
Processors $\ProcConfluence$, $\ProcLocalConfluence$ and $\ProcStrongConfluence$ relate confluence problems and local/strong confluence problems.
Processors $\ProcCanJoinability$ and $\ProcCnvJoinability$ translate confluence problems for \trs{s} into local confluence problems of \cstrs{s}.
Processor $\ProcCanCSR$ translates confluence problems for \trs{s} into confluence problems for \cstrs{s}.
Finally, for terminating \gtrs{s} $\cR$, $\ProcKnuthBendix$ either translates
confluence problems  into joinability problems of an appropriate set of conditional critical pairs, or else into a local confluence problem for $\cR$.
Table~\ref{TableUseOfConfluenceProcessorsInTheConfluenceFramework}~summarizes the use of processors for confluence problems in the confluence framework.

 \subsubsection{Confluence of terminating CTRSs as confluence of TRSs: $\ProcTrU$}
\label{SecProcTrU}

An oriented \ctrs{} $\cR$ is
\emph{deterministic} (\dctrs{}) if for every rule $\lhsr\to\rhsr\IF s_1\cto t_1,\ldots,s_n\cto t_n$  in $\cR$ and every $1\leq i\leq n$,
we have $\var(s_i)\subseteq\var(\lhsr)\cup\bigcup_{j=1}^{i-1}\Var(t_j)$
\cite[Def.\ 7.2.33]{Ohlebusch_AdvTopicsTermRew_2002}.
In the following, $\cU$ is the transformation in
\cite[Def.\ 7.2.48]{Ohlebusch_AdvTopicsTermRew_2002} that, given a  3-\dctrs{s}
$\cR=(\Symbols,R)$ obtains a \trs{} $\cU(\cR)$.
Each rule
$
\alpha:\lhsr\to\rhsr \IF s_1\cto t_1,\ldots,s_n\cto t_n\in R
$
is transformed into $n+1$ unconditional rules:
\begin{eqnarray*}
\lhsr &\to& U^\alpha_1(s_1,\vec{x}_1) \label{eq:ul}\\
U^\alpha_{i-1}(t_{i-1},\vec{x}_{i-1}) & \to & U^\alpha_i(s_i,\vec{x}_i) \qquad 2\leq
i \leq n \label{eq:ui}\\
U^\alpha_n(t_n,\vec{x}_n) & \to & r \label{eq:ur}
\end{eqnarray*}
where $U^\alpha_i$ are fresh new symbols and $\vec{x}_i$ are sequences
of
variables in
$\Var(\lhsr) \cup \Var(t_1) \cup \cdots \cup \Var(t_{i-1})
$
for $1\leq i\leq n$.
Unconditional rules remain unchanged.
Then, $\cU(\cR)=(\cU(\Symbols),\cU(R))$, where
$\cU(\Symbols)$ is  the signature $\Symbols$ extended with the new symbols
introduced by transformation $\cU$ and $\cU(R)$ is the new set of rules obtained from $R$ as explained above.
\begin{example}[Transformation $\cU$]
\label{ExTrU}
For the
\dctrs{}
$\cR=(\Symbols,R)$ with $R=\{\fS{a}\to\fS{b}\IF\fS{a}\cto\fS{b},\fS{a}\to\fS{b}\IF\fS{a}\cto\fS{c}\}$,
$\cU(R)$ consists of:
\begin{eqnarray}
\fS{a} & \to & U(\fS{a})\\
U(\fS{b}) & \to & \fS{b}\\
\fS{a} & \to &  U'(\fS{a})\\
U'(\fS{c}) & \to & \fS{b}
\end{eqnarray}
Note that $\cU(\Symbols)=\{\fS{a},\fS{b},\fS{c},U,U'\}$.
\end{example}
\begin{remark}[Simulating conditional rewriting by unconditional rewriting]
Starting from the pioneering work by Marchiori \cite{Marchiori_UnravelingsAndUltraProperties_ALP96},
several authors have investigated the ability of
transformations from \ctrs{s} $\cR$ to \trs{s} to \emph{simulate} conditional rewriting
by means of unconditional rewriting, see, e.g., \cite{GmeGraSch_SoundnessCondUnravelingCTRS_RTA12,NisSakSak_SoundnessOfUnravelingsForCTRSsViaUltraProperties_LMCS12} and the references therein.
It is  well-known that $\cU$ is \emph{simulation-complete}, i.e.,
$\rew{\cR}\:\subseteq\:\rews{\cU(\cR)}$ see \cite[Section 3]{NisSakSak_SoundnessOfUnravelingsForCTRSsViaUltraProperties_LMCS12};
however, $\cU$ is \emph{not} \emph{simulation-sound}:
there are terms $s,t\in\Terms$ such that $s\rews{\cU(\cR)}t$ but
$s\not\rews{\cR}t$, see, e.g., \cite[Example 3.3]{NisSakSak_SoundnessOfUnravelingsForCTRSsViaUltraProperties_LMCS12}.
\end{remark}
The following definition prepares Theorem \ref{TheoConfluenceOfCTRSsByUTransformation} below, which enables the
use of transformation $\cU$ in the confluence framework.

\begin{definition}
Let $\cR$ be a 3-\dctrs{}.
We say that $\cU$ preserves $\rew{\cR}$-irreducibility of a 3-\dctrs{}
$\cR=(\Symbols,R)$
if for all terms
$t\in\Terms$, if $t$ is $\rew{\cR}$-irreducible, then $t$ is $\rew{\cU(\cR)}$-irreducible.
\end{definition}
In general, this property does \emph{not} hold for all 3-\dctrs{s}. 
\begin{example}
Let $\cR$ and $\cU(\cR)$ as in Example \ref{ExTrU}.
Although term $\fS{a}$ is $\rew{\cR}$-irreducible, it is \emph{not}
$\rew{\cU(\cR)}$-irreducible as, e.g., $\fS{a}\rew{\cU(\cR)}U(\fS{a})$.
Thus, $\cU$ does \emph{not} preserve $\rew{\cR}$-irreducibility of $\cR$.
\end{example}
\begin{theorem}
\label{TheoConfluenceOfCTRSsByUTransformation}
Let $\cR$ be a 3-\dctrs{}.
If $\cR$ is terminating, $\cU$ preserves $\rew{\cR}$-irreducibility, and
$\cU(\cR)$ is
confluent, then $\cR$ is confluent.\footnote{This result fixes the buggy \cite[Theorem 8]{Lucas_ApplicationsAndExtensionsOfContextSensitiveRewriting_JLAMP21}.}
\end{theorem}

\begin{proof}
By contradiction.
If $\cR$ is not confluent, then  there are terms $s,t,t'\in\Terms$ such that
$s\rews{\cR}t$ and $s\rews{\cR}t'$ but $t$ and $t'$ are not $\to^*_\cR$-joinable.
By termination of $\rew{\cR}$, we can assume that $t$ and $t'$ are different
$\rew{\cR}$-irreducible terms, i.e.,
 $t\neq t'$.
By confluence of $\cU(\cR)$,
$t\rews{\cU(\cR)} u$ and $t'\rews{\cU(\cR)} u$ for some term $u$.
Since $\cU$ preserves $\rew{\cR}$-irreducibility,
both $t$ and $t'$ are $\rew{\cU(\cR)}$-irreducible.
Thus,
$t=u=t'$, a contradiction.
\end{proof}
For a practical use of Theorem \ref{TheoConfluenceOfCTRSsByUTransformation}, we introduce a sufficient condition for preservation of irreducibility.
In the following, given a sequence of expressions $\gencond$ and a set of variables $V$, $\groundingTermOn{\gencond}{V}$ is the \emph{partial grounding} of
$\gencond$ where all variables $x\in\Var(\gencond)\cap V$ are replaced by $c_x$.
The following result provides a \emph{sufficient condition}
for preservation of $\rew{\cR}$-irreducibility.
Intuitively, requiring \emph{feasibility} of $\groundingTermOn{\gencond}{\Var(\lhsr)}$
for all rules $\lhsr\to\rhsr\IF\gencond$ guarantees that no particular
 instantiation of variables due to pattern
matching against $\lhsr$ plays a role in the satisfaction of the conditional part $\gencond$
of the rule.

\begin{proposition}
\label{PropSufficientConditionForURpreservingRIrreducibility}
Let
$\cR=(\Symbols,R)$
be a 3-\dctrs{}.
If for all $\lhsr\to\rhsr\IF\gencond\in R$, $\groundingTermOn{\gencond}{\Var(\lhsr)}$ is $\rewtheoryOf{\cR}$-feasible,
then
$\cU$ preserves $\rew{\cR}$-irreducibility
\end{proposition}

\begin{proof}
If $t\in\Terms$ is $\rew{\cR}$-irreducible, we consider two cases:
First, if there is no $p\in\Pos(t)$ such that $t|_p=\sigma(\lhsr)$ for
some $\lhsr\to\rhsr\IF\gencond\in R$ and substitution $\sigma$,
then $t$ is  also $\rew{\cU(\cR)}$-irreducible, as the left-hand sides of rules
in $\cU(\cR)$ either coincide with those in $\cR$ or include a symbol in $\cU(\Symbols)-\Symbols$.
Otherwise,
since $t$ is $\rew{\cR}$-irreducible, $\sigma'(\gencond)$ must be $\rewtheoryOf{\cR}$-infeasible
for all substitutions $\sigma'$  such that $\sigma'(x)=\sigma(x)$ for all $x\in\Var(\lhsr)$.
However, $\groundingTermOn{\gencond}{\Var(\lhsr)}$ is $\rewtheoryOf{\cR}$-feasible, i.e., there is a substitution $\varsigma$ such that $\varsigma(\groundingTermOn{\gencond}{\Var(\lhsr)})$ holds.
Thus the substitution $\sigma'$ defined by
$\sigma'(x)=\sigma(x)$ for all $x\in\Var(\lhsr)$ and $\sigma'(y)=\varsigma(y)$ for all $y\notin\Var(\lhsr)$ also satisfies $\sigma(\gencond)$, a contradiction.
\end{proof}
\begin{example}
\label{ExUconfFailsAndURsucceeds}
Consider the following  3-\dctrs{} $\cR$:
\begin{eqnarray}
\fS{g}(x) & \to & x\label{ExUconfFailsAndURsucceeds_rule1}\\
\fS{f}(x) & \to & x  \IF  \fS{g}(x)\cto x\label{ExUconfFailsAndURsucceeds_rule2}
\end{eqnarray}
Since $\fS{g}(c_x)\cto c_x$ is clearly \emph{feasible} (we have $\ul{\fS{g}(c_x)}\rew{(\ref{ExUconfFailsAndURsucceeds_rule1})} c_x$),
$\cU$ preserves $\rew{\cR}$-irreducibility.
\end{example}
Processor $\ProcTrU$ transforms a confluence problem for a  terminating 3-\dctrs{}
$\cR$ into a
confluence problem for a TRS $\cU(\cR)$.
If $\cR$ is a terminating 3-\dctrs{} such that $\cU$ preserves $\rew{\cR}$-irreducibility,
then
\[\begin{array}{rcl}
\ProcTrU(\CRproblem{\cR})
& = & \{\CRproblem{\cU(\cR)}\}
\end{array}
\]
By relying on
Theorem \ref{TheoConfluenceOfCTRSsByUTransformation},
$\ProcTrU$ is \emph{sound}.
However, it is \emph{not} complete.

\begin{example}[$\ProcTrU$ is not complete]
\label{ExProcTrUisNotComplete}
The  3-\dctrs{} $\cR$ in Example \ref{ExTrU}
is
(locally) confluent: the rules are infeasible; thus,
the one-step reduction relation is \empty{empty}.
However, $\cU(\cR)$ defines a non-joinable peak $U(\fS{a})\leftrew{\cU(\cR)}\fS{a}\rew{\cU(\cR)}U'(\fS{a})$.
\end{example}

\subsubsection{Confluence of CTRSs as confluence of TRSs using an improved transformation: $\ProcTrUconf$}
\label{SecProcTrUconf}

Gmeiner, Nishida and Gramlich
\cite{GmeNisGra_ProvingConfluenceOfConditionalTermRewritingSystemsViaUnravelings_IWC2013}
introduced transformation $\Uconf$,
which can also be used in proofs of confluence of \emph{$3$-\dctrs{s}} in the confluence framework.
Each rule
$
\alpha:\lhsr\to\rhsr \IF s_1\cto t_1,\ldots,s_n\cto t_n
$
is transformed into $n+1$ unconditional rules
\cite[Definition 6]{GmeNisGra_ProvingConfluenceOfConditionalTermRewritingSystemsViaUnravelings_IWC2013}:
\begin{eqnarray*}
\lhsr &\to& U_{\lhsr,s_1}(s_1,\vec{x}_1) \label{eq:ul}\\
U_{\lhsr,s_1}(t_1,\vec{x}_{1}) & \to & U_{\lhsr,s_1,t_1,s_2}(s_2,\vec{x}_2)\\
\vdots
\\
U_{\lhsr,s_1,t_1,\ldots,s_{n}}(t_n,\vec{x}_n) & \to & r
\end{eqnarray*}
where
new symbols $U_{\lhsr,\ldots} $ depending on the the left-hand side
$\lhsr$ of  $\alpha$ and also on the terms occurring in the conditions of $\alpha$ are introduced.
For each $1\leq i\leq n$,
$\vec{x}_i$ is as in
transformation $\cU$.
Then, $\Uconf(\cR)=(\Uconf(\Symbols),\Uconf(R))$, where
$\Uconf(\Symbols)$ is  the signature $\Symbols$ extended with the new symbols
introduced by transformation $\Uconf$ and $\Uconf(R)$ is the new set of rules obtained from $R$.
\begin{example}[Transformation $\Uconf$]
\label{ExTrUconf}
For
$\cR=(\Symbols,R)$ in Example \ref{ExTrU},
$\Uconf(R)$ consists of the rules:
\begin{eqnarray}
\fS{a} & \to & U(\fS{a})\\
U(\fS{b}) & \to & \fS{b}\\
U(\fS{c}) & \to & \fS{b}
\end{eqnarray}
Compared with $\cU(\cR)$ in Example \ref{ExTrU}, note that $U'$ is missing thanks to the refined definition of $\Uconf$.
\end{example}
A \dctrs{} is \emph{weakly left-linear} if ``variables that occur more than once in the lhs of a conditional rule and the rhs’s of conditions should not occur at all in lhs’s of conditions or the rhs of the conditional rule'' \cite[Definition 3.17]{GmeGraSch_SoundnessCondUnravelingCTRS_RTA12}.
Processor $\ProcTrUconf$ transforms a confluence problem for a
3-\dctrs{} $\cR$ into a
confluence problem for a \trs{} $\Uconf(\cR)$,
where
$\Uconf$ is the transformation in
\cite[Definition 6]{GmeNisGra_ProvingConfluenceOfConditionalTermRewritingSystemsViaUnravelings_IWC2013}:
\[ \ProcTrUconf(\CRproblem{\cR})=
\{\CRproblem{\Uconf(\cR)}\} \]
if $\cR$ is a  weakly left-linear  3-\dctrs{}.
By  \cite[Theorem 9]{GmeNisGra_ProvingConfluenceOfConditionalTermRewritingSystemsViaUnravelings_IWC2013},
$\ProcTrUconf$ is sound; however, it is \emph{not} complete.

\begin{example}[$\ProcTrUconf$ is not complete]
\label{ExProcTrUisNotComplete}
For the 3-\dctrs{} $\cR=(\Symbols,R)$ with $R=\{\fS{a}\to\fS{b}\IF\fS{b}\cto\fS{a},\fS{a}\to\fS{b}\IF\fS{c}\cto\fS{a}\}$, we have
\[\cU(\cR)=\Uconf(\cR)=\{\fS{a}\to U(\fS{b}), U(\fS{a})\to \fS{b},\fS{a}\to U'(\fS{c}), U'(\fS{a})\to \fS{b}\}.\]
The rules of $\cR$ are infeasible and can be removed. Thus, $\rew{\cR}$ is \emph{empty},
hence confluent.
However, the peak
$U(\fS{b})\leftrew{\Uconf(\cR)}\fS{a}\rew{\Uconf(\cR)}U'(\fS{c})$ is  not joinable, as both $U(\fS{b})$ and $U'(\fS{c})$ are
$\Uconf(\cR)$-irreducible.
\end{example}
Note that $\cR$ in Example \ref{ExTrU} can be proved confluent using
$\Uconf(\cR)$: as shown in Example \ref{ExTrUconf}, $\Uconf(\cR)$ is orthogonal, hence confluent, which proves confluence of $\cR$ by \cite[Theorem 9]{GmeNisGra_ProvingConfluenceOfConditionalTermRewritingSystemsViaUnravelings_IWC2013}.
However, the following example shows that $\ProcTrU$ can be used to prove
confluence when
$\ProcTrUconf$ fails.

\begin{example}
\label{ExUconfFailsAndURsucceeds_check}
The 3-\dctrs{} $\cR$ in Example \ref{ExUconfFailsAndURsucceeds}
is clearly terminating and, as shown in the example, $\cU$ preserves $\rew{\cR}$-irreducibility.
The \trs{} $\cU(\cR)$:
\begin{eqnarray}
\fS{g}(x) & \to & x\label{ExUconfFailsAndURsucceeds_U_rule1}\\
\fS{f}(x) & \to & U(\fS{g}(x),x)\label{ExUconfFailsAndURsucceeds_U_rule2}\\
U(x,x) & \to & x\label{ExUconfFailsAndURsucceeds_U_rule3}
\end{eqnarray}
is terminating and has no critical pair.
Hence, $\cU(\cR)$ is confluent.
By Theorem \ref{TheoConfluenceOfCTRSsByUTransformation}, $\cR$ is confluent.
Note that $\cR$ is \emph{not} weakly left-linear.
Thus, $\Uconf(\cR)$ cannot be used to prove confluence of $\cR$.
\end{example}
Thus, $\ProcTrU$ and $\ProcTrUconf$ are complementary.
However, our experiments show that $\ProcTrUconf$ applies more frequently than $\ProcTrU$,
see Table \ref{TableUseOfProcessorsInExperiments}.

\subsubsection{Orthogonality: $\ProcOrthogonality$}
\label{SecOrthogonality}

A \gtrs{} $\cR$ is \emph{left-linear} if for all rules $\lhsr\to\rhsr\IF\gencond$,  the left-hand side $\lhsr$ is linear.
\begin{itemize}
\item A left-linear TRS $\cR$ whose critical pairs
are all trivial
is called
\emph{weakly orthogonal}.
For left-linear TRSs, Huet provided several results  \cite[Section 3.3]{Huet_ConfReduc_JACM80}
leading, in particular, to conclude that \emph{weakly orthogonal TRSs are confluent},
see \cite[Section
6.4]{BaaNip_TermRewAllThat_1998}.
\item A left-linear CS-TRS $(\cR,\mu)$ is $\mu$-orthogonal if $\CP(\cR,\mu)=\LHCP(\cR,\mu)=\emptyset$
\cite[Definition 35]{LucVitGut_ProvingAndDisprovingConfluenceOfContextSensitiveRewriting_JLAMP22}. By \cite[Corollary 36]{LucVitGut_ProvingAndDisprovingConfluenceOfContextSensitiveRewriting_JLAMP22},
$\mu$-orthogonal CS-TRSs are confluent.
\item A left-linear CTRS is (almost) \emph{orthogonal} if $\pCCP(\cR)=\emptyset$ (resp.\ $\pCCP(\cR)$ consists of \emph{trivial}
pairs $\langle t,t\rangle\IF\gencond$ with critical position $p=\toppos$)
\cite[Definition 7.1.10(1 \& 2)]{Ohlebusch_AdvTopicsTermRew_2002}.
By \cite[Theorem 7.4.14]{Ohlebusch_AdvTopicsTermRew_2002},\footnote{Originally in \cite[Theorem 4.6]{SuzMidIda_LevelConfluenceOfConditionalRewriteSystemsWithExtraVariablesInRightHandSides_RTA95}, this result concerns \emph{level-confluence}, which implies confluence.
}
orthogonal, properly oriented, and right-stable $3$-CTRS are confluent, where
\begin{itemize}
\item A $3$-CTRS $\cR$ is
\emph{properly oriented} if every rule $\lhsr\to\rhsr\IF s_1\cto t_1,\ldots,s_n\cto t_n$
satisfies: if $\Var(\rhsr)\notin\Var(\lhsr)$, then $\Var(s_i)\subseteq\Var(\lhsr)\cup\bigcup_{j=1}^{i-1}\Var(t_j)$ for all $1\leq i\leq n$.
\cite[Definition 7.4.13]{Ohlebusch_AdvTopicsTermRew_2002}.
\item A CTRS is \emph{right-stable} if for every rule $\lhsr\to\rhsr\IF s_1\cto t_1,\ldots,s_n\cto t_n\in\cR$
and for all $1\leq i\leq n$,
(a) $(\Var(\lhsr)\cup\bigcup_{j=1}^{i-1}\Var(s_j\cto t_j)\cup\Var(s_i))\cap\Var(t_i)=\emptyset$,
and (b) $t_i$ is either a linear constructor term or a ground $\cR_u$-irreducible term
\cite[Definition 7.4.8]{Ohlebusch_AdvTopicsTermRew_2002},
where $\cR_u$ is
obtained from the rules of $\cR$ by just dropping the conditional part: $\cR_u=\{\lhsr\to\rhsr\mid\lhsr\to\rhsr\IF\gencond\in\cR\}$.
\cite[Definition 7.1.2]{Ohlebusch_AdvTopicsTermRew_2002}.
\end{itemize}
By \cite[Corollary 7.4.11]{Ohlebusch_AdvTopicsTermRew_2002}, almost orthogonal and almost normal (i.e.,  right-stable and oriented
\cite[Definition 7.4.8(2)]{Ohlebusch_AdvTopicsTermRew_2002}) $2$-\ctrs{s} are confluent.
\end{itemize}
Thus,
we let (using the fact that confluence implies local confluence):
\[\ProcOrthogonality(\CRproblem{\cR})=\ProcOrthogonality(\WCRproblem{\cR})=\emptyset \text{ if $\cR$ is }
\left \{
\begin{array}{l}
\text{a weakly orthogonal TRS, or}\\
\text{a $\mu$-orthogonal CS-TRS, or}\\
\text{an almost orthogonal and almost normal}\\[-0.1cm]
\text{~~$2$-CTRS, or}\\
\text{an orthogonal, properly oriented, and right-stable}\\[-0.1cm]
\text{~~$3$-CTRS}
\end{array}
\right .
\]
In all these uses, $\ProcOrthogonality$ is \emph{sound} and (trivially) \emph{complete}.

\begin{example}[Continuing Example \ref{Ex4_IWC23}]
\label{Ex4_IWC23_PHuetLevy}
Since the TRS $\cR_1$ in Example \ref{Ex4_IWC23},  is orthogonal, we obtain $\ProcOrthogonality(\CRproblem{\cR_1})=\emptyset$.
\end{example}

\subsubsection{Confluence and local/strong confluence: $\ProcConfluence$, $\ProcLocalConfluence$, $\ProcStrongConfluence$}
\label{SecConfluenceAsLocalStrongConfluence}

Processors $\ProcConfluence$, $\ProcLocalConfluence$ and $\ProcStrongConfluence$ implement the
well-known relationships between confluence and local and strong confluence in Remark \ref{RemRelationshipBetweenProblems} from which soundness and completeness properties of these processors follow.

\paragraph{Local and strong confluence problems as confluence problems.}
The following processor transforms local confluence problems into confluence problems: for \gtrs{s} $\cR$,
\[
\ProcConfluence(\WCRproblem{\cR}) = \ProcConfluence(\SCRproblem{\cR}) = \{\CRproblem{\cR}\}.
\]
$\ProcConfluence$ is \emph{sound} on local confluence problems,
and \emph{complete} on strong confluence problems.

\paragraph{Confluence and strong confluence problems as local confluence problems.}
The following processors transforms confluence problems into local confluence problems: for \gtrs{s} $\cR$.
\[\begin{array}{rcl}
\ProcLocalConfluence(\CRproblem{\cR})  =  \ProcLocalConfluence(\SCRproblem{\cR})=\{\WCRproblem{\cR}\}.
\end{array}
\]
$\ProcLocalConfluence$ is \emph{complete}, but not \emph{sound}.

\paragraph{Confluence and local confluence problems as strong confluence problems.}

Strong confluence implies confluence and hence local confluence
(but not vice versa) \cite[Lemma 2.5]{Huet_ConfReduc_JACM80}, see also
\cite[Section 2.2]{Ohlebusch_AdvTopicsTermRew_2002}.
The following processor uses this fact: for \gtrs{s} $\cR$,
\[\begin{array}{rcl}
\ProcStrongConfluence(\CRproblem{\cR})=
\ProcStrongConfluence(\WCRproblem{\cR})=\{\SCRproblem{\cR}\}.
\end{array}
\]
$\ProcStrongConfluence$ is \emph{sound},
but \emph{not} complete.

\subsubsection{Confluence of a TRS as local confluence of a terminating CS-TRS: $\ProcCanJoinability$, $\ProcCnvJoinability$}
\label{SecConfluenceAsLocalConfluenceOfCSR}

Replacement maps can be used to prove confluence of a TRS $\cR$ by transforming
it the
into a CS-TRS $(\cR,\mu)$.
We consider two replacement maps:
\begin{itemize}
\item The \emph{canonical replacement map} $\muCan$ is \emph{the most restrictive replacement
map ensuring that the non-variable subterms of the left-hand sides of the rules of $\cR$ are all active}
\cite[Section 5]{Lucas_ContextSensitiveRewriting_CSUR20}.
\item The \emph{convective replacement map} $\muCnv$ is \emph{the most restrictive replacement map
that makes all critical positions $p$ of critical pairs $\langle\theta(\lhsr)[\theta(\rhsr')]_p,\theta(\rhsr)\rangle\in\CP(\cR)$ active}
\cite[Definition 3]{Oostrom_TheZpropertyForLeftLinearTermRewritingViaConvectiveContextSensitiveCompleteness_IWC23}.
\end{itemize}
Replacement maps $\mu$ that are \emph{less restrictive} than $\muCan$ (i.e., such that $\muCan(f)\subseteq\mu(f)$ for all symbols $f$ in the signature, written $\muCan\sqsubseteq\mu$) are collected in the set $\CRmaps{\cR}$
and similarly for $\muCnv$ with $\CnvRmaps{\cR}$.
Since $\muCan$ is not more restrictive than $\muCnv$,
i.e., $\muCnv\sqsubseteq\muCan$,
we have $\CRmaps{\cR}\subseteq\CnvRmaps{\cR}$.
\begin{example}\label{Ex42_COPS}
Consider the following TRS (COPS/42.trs):
\begin{eqnarray}
\fS{f}(\fS{g}(x)) & \to & \fS{f}(\fS{h}(x,x))\label{Ex42_COPS_rule1}\\
\fS{g}(\fS{a}) & \to & \fS{g}(\fS{g}(\fS{a}))\\
\fS{h}(\fS{a},\fS{a}) & \to & \fS{g}(\fS{g}(\fS{a}))
\end{eqnarray}
We have $\muCan=\muTop$.
The
critical position $1\in\Pos(\lhsr_{(\ref{Ex42_COPS_rule1})})$
of the only critical pair
$\langle\fS{f}(\fS{g}(\fS{g}(\fS{a}))),\fS{f}(\fS{h}(\fS{a},\fS{a}))\rangle$  becomes \emph{active} by just letting
$\muCnv(\fS{f})=\{1\}$ and $\muCnv(p)=\emptyset$ for any other symbol $p$.
\end{example}

\begin{example}\label{Ex4_IWC23_muCan_muCnv}
Consider the TRS $\cR$ in Example \ref{Ex4_IWC23}.We have:
\[\begin{array}{lll}
\muCan(\fS{inc})=\muCan(\fS{hd})=\muCan(\fS{tl})=\{1\}, & \muCan(\fS{\fS{:}})=\muCan(\fS{from})=\muCan(\fS{s})=\emptyset.
\end{array}
\]
Also, with rule (\ref{Ex4_IWC23_rule6}), i.e.,
$\fS{inc}(\fS{tl}(\fS{from}(x))) \to \fS{tl}(\fS{inc}(\fS{from}(x)))$,
critical position $p=1.1$ in the left-hand side $\fS{inc}(\fS{tl}(\fS{from}(x)))$
of the rule,
and rule (\ref{Ex4_IWC23_rule5}), i.e.,
$\fS{from}(x) \to  x\fS{:}\fS{from}(\fS{s}(x))$,
we obtain a critical pair
\begin{eqnarray}
\langle\fS{inc}(\fS{tl}(x\fS{:}\fS{from}(\fS{s}(x)))),\fS{tl}(\fS{inc}(\fS{from}(x)))\rangle\label{Ex4_IWC23_CP}
\end{eqnarray}
which is the only critical pair in $\CP(\cR)$.
We only need to make the arguments of $\fS{inc}$ and $\fS{tl}$ active to
guarantee that $p$ is active in $\fS{inc}(\fS{tl}(\fS{from}(x)))$. Thus, we have:
\[\begin{array}{lll}
\muCnv(\fS{inc})=\muCnv(\fS{tl})=\{1\}, & \muCnv(\fS{\fS{:}})=\muCnv(\fS{from})=\muCnv(\fS{hd})=\muCnv(\fS{s})=\emptyset.
\end{array}
\]
\end{example}
A CS-TRS $(\cR,\mu)$ is
\begin{itemize}
\item \emph{level-decreasing} if for all rules
$\lhsr\to\rhsr$ in $\cR$, the level of each variable in $\rhsr$ does not exceed its
level in $\lhsr$; the \emph{level}  $\level_\mu(t,p)$ of an occurrence
of variable $x$ at position $p$ in a term $t$ is obtained by adding the number of
frozen arguments that
are traversed from the root to the occurrence $t|_p=x$ of the variable.
Then, $\level_\mu(t,x)$ is the maximum level to which $x$ occurs in $t$
\cite[Definition 1]{GraLuc_GenNewmanLemma_RTA06}.
\item
\emph{$0$-preserving} \cite[page 6]{Oostrom_TheZpropertyForLeftLinearTermRewritingViaConvectiveContextSensitiveCompleteness_IWC23}
if for all rules $\lhsr\to\rhsr\in\cR$, if a variable occurs active in the left-hand side $\lhsr$ of rules $\lhsr\to\rhsr$,
then all its occurrences  in  $\rhsr$ are also active, i.e., $\Var^\mu(\lhsr)\cap\NVar{\mu}(\rhsr)=\emptyset$.
For left-linear TRSs, this property coincides with $\lhsr\to\rhsr$ having \emph{left-homogeneous $\mu$-replacing variables}
\cite[Section 8.1]{Lucas_ContextSensitiveRewriting_CSUR20}, written $\LHRV(\cR,\mu)$,
see, e.g., \cite[Proposition 3]{LucVitGut_ProvingAndDisprovingConfluenceOfContextSensitiveRewriting_JLAMP22}.
\end{itemize}
\begin{example}\label{Ex4_IWC23_NoLevelDecreasing_LHRV}
Consider the TRS $\cR$ in Example \ref{Ex4_IWC23}. 
\begin{itemize}
\item With $\mu=\muCan$, $\cR$ is \emph{not} level-decreasing, as for rule (\ref{Ex4_IWC23_rule5}),
\[\level_\mu(\lhsr_{(\ref{Ex4_IWC23_rule5})},x)=\level_\mu(\fS{from}(x),x)=0<2=\level_\mu(x\fS{:}\fS{from}(\fS{s}(x)),x)=\level_\mu(\rhsr_{(\ref{Ex4_IWC23_rule5})},x).\]
\item
With respect to both $\muCnv$ and $\muCan$, $\cR$ is $0$-preserving as all variable occurrences in left-hand sides of rules are frozen.
\end{itemize}
\end{example}
Let $\cR=(\Symbols,R)$ be a left-linear TRS and
assume $\cR^\mu=(\Symbols,\mu,R)$ terminating for $\mu$ as given below. We
define
\[\begin{array}{rcll}
\ProcCanJoinability(\CRproblem{\cR}) & = & \{\WCRproblem{\cR^\mu}\} &
\text{if
$\mu\in\CRmaps{\cR}$ and
$\cR$ is level-decreasing.}\\
\ProcCnvJoinability(\CRproblem{\cR}) & = & \{\WCRproblem{\cR^\mu}\} &
\text{if $\mu\in\CnvRmaps{\cR}$ and
$\LHRV(\cR,\mu)$ holds.
}
\end{array}
\]
\begin{example}[Continuing Example \ref{Ex4_IWC23}]
\label{Ex4_IWC23_useOfP_CnvJ}
Since
$\cR_2$ in Example \ref{Ex4_IWC23} is left-linear and $\LHRV(\cR_2,\muCnvOf{\cR_2})$ holds,
we obtain $\ProcCnvJoinability(\CRproblem{\cR_2})=\{\WCRproblem{\cR_2^\mu}\}$,
where $\cR_2^\mu=(\cR_2,\muCnvOf{\cR_2})$.
\end{example}
By \cite[Theorem 2]{GraLuc_GenNewmanLemma_RTA06}
(resp.\
\cite[Corollary 13]{Oostrom_TheZpropertyForLeftLinearTermRewritingViaConvectiveContextSensitiveCompleteness_IWC23}),
$\ProcCanJoinability$ (resp.\
$\ProcCnvJoinability$) is \emph{sound}.
However, they are \emph{not} complete.

\begin{example}[$\ProcCanJoinability$ and $\ProcCnvJoinability$ are not complete]
\label{ExProcCanJoinabilityIsNotComplete}
For $\cR$ consisting of the rules
\begin{eqnarray}
\fS{a} & \to & \fS{b}\label{ExProcCanJoinabilityIsNotComplete_rule1}\\
\fS{c} & \to & \fS{d}(\fS{a})\label{ExProcCanJoinabilityIsNotComplete_rule2}\\
\fS{c}& \to & \fS{d}(\fS{b})\label{ExProcCanJoinabilityIsNotComplete_rule3}
\end{eqnarray}
we have
$\muCan=\muCnv=\muBot$.
In particular, $\muCnv(\fS{d})=\emptyset$.
The (only) critical pair
\begin{eqnarray}
\langle\fS{d}(\fS{a}),\fS{d}(\fS{b})\rangle\label{LblExProcCanJoinabilityIsNotComplete_CP}
\end{eqnarray}
is \emph{not} $(\cR,\muBot)$-joinable,
hence $(\cR,\muBot)$ is \emph{not} locally confluent.
However,
$\cR$ is confluent as it is terminating and $\pi$ is clearly $\cR$-joinable.
\end{example}

\begin{remark}[$\ProcCnvJoinability$ subsumes $\ProcCanJoinability$]
Since $\muCnv\sqsubseteq\muCan$, and level-decreasingness implies the \LHRV{} property
(see \cite[page 72]{GraLuc_GenNewmanLemma_RTA06}), and hence $0$-decreasingness, all uses of $\ProcCanJoinability$ are covered by $\ProcCnvJoinability$.
In practice, though, this may depend on the specific choice(s) of $\mu$ when implementing the processor.
For instance, consider $\cR'$ consisting of rules
(\ref{ExProcCanJoinabilityIsNotComplete_rule1}),
(\ref{ExProcCanJoinabilityIsNotComplete_rule2}),
and
(\ref{ExProcCanJoinabilityIsNotComplete_rule3}),
together with
\begin{eqnarray}
\fS{d}(\fS{b}) & \to \fS{b}\label{ExProcCanJoinabilityIsNotComplete_addedRule}
\end{eqnarray}
For $\cR'$ and $\cR$ in Example \ref{ExProcCanJoinabilityIsNotComplete}, with
$\CP(\cR)=\CP(\cR')=\{(\ref{LblExProcCanJoinabilityIsNotComplete_CP})\}$,
we have
$\muCnv=\muCnvOf{\cR'}=\muBot$.
Still, (\ref{LblExProcCanJoinabilityIsNotComplete_CP}) is not $(\cR',\muBot)$-joinable.
However, $\muCanOf{\cR'}(\fS{d})=\{1\}$ now (due to rule (\ref{ExProcCanJoinabilityIsNotComplete_addedRule})) and hence (\ref{LblExProcCanJoinabilityIsNotComplete_CP}) is $(\cR',\muCanOf{\cR'})$-joinable.
Thus, $(\cR',\muCanOf{\cR'})$ is locally confluent and (by soundness of $\ProcCanJoinability$
and termination of $(\cR',\muCanOf{\cR'})$),
$\cR$ is confluent.
Typically, an implementation of $\ProcCnvJoinability$ would try $\muCnv$ \emph{only}, although
$\muCan\in\CnvRmaps{\cR}$ is a possible choice as well.
\end{remark}

\subsubsection{Confluence of a TRS as confluence of canonical \csr: $\ProcCanCSR$}
\label{SecConfluenceAsCanonicalConfluenceOfCSR}

By relying on \cite[Corollary 8.23]{Lucas_ContextSensitiveRewriting_CSUR20},
processor $\ProcCanCSR$ transforms a confluence problem $\CRproblem{\cR}$ for a TRS $\cR$
into a confluence problem for a CS-TRS:
\[
\ProcCanCSR(\CRproblem{\cR})= \{\CRproblem{\cR^\mu}\}
\]
if $\cR$ is a left-linear and \emph{normalizing} TRS (i.e., every term has a normal form),
and  $\cR^\mu=(\cR,\mu)$ for some $\mu\in\CRmaps{\cR}$.
$\ProcCanCSR$ is \emph{sound} but \emph{not} complete, see \cite[Sections 8.4 \& 8.5]{Lucas_ContextSensitiveRewriting_CSUR20} for a
discussion.

\subsubsection{Confluence of terminating systems as local confluence: $\ProcKnuthBendix$}
\label{SecKnuthBendix}

The characterization of confluence of terminating TRSs as the joinability of all its critical pairs
is a landmark, early result by Knuth and Bendix \cite{KnuBen_SimpWordProbUnivAlg_CPAA70}.
Thus, we let
\[
\begin{array}{r@{\:}c@{\:}l}
\ProcKnuthBendix(\CRproblem{\cR}) & = & \left \{
\begin{array}{ll}
\{\mathit{JO}(\cR,\pi_1),\ldots,\mathit{JO}(\cR,\pi_n)\} &
\text{if $\{\pi_1,\ldots,\pi_n\}=\pCCP(\cR)\cup\iCCP(\cR)$}\\
&
\text{are overlays and $\cR$ is a
J-CTRS}\\
\{\WCRproblem{\cR}\} & \text{
otherwise}
\end{array}
\right .
\end{array}\]
if
$\cR$ is a terminating \gtrs{} \cite{Lucas_TerminationOfGeneralizedTermRewritingSystems_FSCD24}.
By relying on \cite[Theorem 4]{DerOkaSiv_ConfluenceOfConditionalRewriteSystems_CTRS87} for
the first case of the application of $\ProcKnuthBendix$ and Theorem \ref{TheoLocalConfluenceCSCTRSsWithExtendedCCPs} (plus Newman's Lemma) for the second one,
$\ProcKnuthBendix$ is  \emph{sound} and \emph{complete}.

\begin{example}[Continuing Example \ref{COPS409Ex}]
\label{COPS409Ex_KBproc}
 For the \emph{oriented} 1-CTRS $\cR'$ obtained by $\ProcSimplify$
 in Example \ref{COPS409Ex}, there is a proper \CCP:
\begin{eqnarray}
\langle \fS{h}(\fS{s}(x)), \fS{g}(\fS{s}(x))\rangle \IF  \fS{c}(\fS{g}(x)) \cto   \fS{c}(\fS{a}),\fS{c}(\fS{h}(x)) \cto \fS{c}(\fS{a})\label{COPS409Ex_pCCP}
\end{eqnarray}
which is an \emph{overlay}, as the critical position is $p=\toppos$.
Since $\cR'$ is a $1$-CTRS, we dismiss improper critical pairs.
Thus, we have
$\ProcKnuthBendix(\CRproblem{\cR'})=\{\JOproblem{\cR',(\ref{COPS409Ex_pCCP})}\}$.
\end{example}
\begin{example}[Continuing Examples \ref{ExCOPS_387_CS_CTRS} and \ref{ExCOPS_387_CS_CTRS_Rbot_NoECCPs}]
\label{ExCOPS_387_CS_CTRS_PHuet_Rbot_PKnuthBendix}
The \csctrs{} $\cR_\bot$ in Example \ref{ExCOPS_387_CS_CTRS}, is clearly terminating. Thus,
$\ProcKnuthBendix(\CRproblem{\cR_\bot})=\{\WCRproblem{\cR_\bot}\}$.
\end{example}

\subsection{Joinability processor: $\ProcJoinability$}
\label{SecJoinabilityProcessor}

 For \gtrs{s} $\cR$
and conditional pairs $\pi$, we have the following processor:
\[\begin{array}{rcl}
\ProcJoinability(\JOproblem{\cR,\pi}) & = & \left \{
\begin{array}{cl}
\emptyset & \text{if $\pi$ is joinable}\\
\mathsf{no} & \text{otherwise}
\end{array}
\right .
\\[0.5cm]
\ProcJoinability(\SJOproblem{\cR,\pi}) & = & \left \{
\begin{array}{cl}
\emptyset & \text{if $\pi$ is strongly joinable}\\
\mathsf{no} & \text{otherwise}
\end{array}
\right .
\end{array}
\]
For both uses $\ProcJoinability$ is \emph{sound} and \emph{complete}.
As in \cite{GutLucVit_ConfluenceOfConditionalRewritingInLogicForm_FSTTCS21,
Lucas_LocalConferenceOfConditionalAndGeneralizedTermRewritingSystems_JLAMP24,
LucVitGut_ProvingAndDisprovingConfluenceOfContextSensitiveRewriting_JLAMP22},
we often prove (non)joinability of terms and critical pairs by proving the \emph{(in)feasibility} of sequences
(see Section \ref{SecPreliminaries}).

\begin{proposition}\label{PropJoinabilityOfCCPsInGTRSs}
(cf.\ \cite[Proposition 21 \& Section 7.5]{Lucas_LocalConferenceOfConditionalAndGeneralizedTermRewritingSystems_JLAMP24}) Let $\cR$ be a \gtrs{} and $\pi:\langle s,t\rangle\IF\gencond$ be a conditional pair
with variables $\vec{x}$, and $z$ be a variable not in $\vec{x}$.
 If (i) $\sigma(c)$ is feasible for some substitution $\sigma$, and
(ii) $\sigma(c),\sigma(s)\to^*z,\sigma(t)\to^*z$ is infeasible (for some
$z\notin\Var(\sigma(c),\sigma(s),\sigma(t))$),
then $\pi$ is not joinable.
\end{proposition}
\begin{remark}
\label{RemPropJoinabilityOfConditionalPairsBySatisfiability_nonJoinability}
As discussed in \cite[Remark 22]{Lucas_LocalConferenceOfConditionalAndGeneralizedTermRewritingSystems_JLAMP24}, in order to use Proposition \ref{PropJoinabilityOfCCPsInGTRSs},
the following \emph{heuristics} are useful.
\begin{enumerate}
\item[H1] The simplest choice is the empty substitution, i.e., $\sigma=\varepsilon$ or its \emph{grounded} version $\sigma=\grounding{\varepsilon}=\{x\mapsto c_x\mid x\in\Variables\}$.
This is easily mechanizable, see Examples \ref{ExCOPS_387_CS_CTRS_pCCP_notJoinable}
and \ref{ExCOPS_387_CS_CTRS_pCCPs_iCCPs_CVPs_joinability}.
\item[H2] Use $\sigma=\{x\mapsto \lhsr^\downarrow, x'\mapsto\rhsr^\downarrow\}$ for some unconditional rule $\lhsr\to\rhsr$ in $\cR$
if $\pi$ is a conditional variable
pair $\langle s,t\rangle\IF x\rew{}x',\gencond$.
\item[H3] Choosing another substitution $\sigma$, usually trying to
fulfill the conditions in the proposition.
\end{enumerate}
\end{remark}

\begin{example}[Continuing Example \ref{ExCOPS_387_CS_CTRS}]
\label{ExCOPS_387_CS_CTRS_pCCP_notJoinable}
We use Proposition \ref{PropJoinabilityOfCCPsInGTRSs} to show
 non-joinability of (\ref{ExCOPS_387_CS_CTRS_pCCP}): the sequence
$\fS{s}(x') \rews{} \fS{s}(\fS{0})$ is feasible, but
\[\fS{s}(x') \rews{} \fS{s}(\fS{0}), \fS{f}(\fS{g}(x'))\rews{}z,\fS{s}(x')\rews{}z\]
is
infeasible:
the only substitution satisfying
$\fS{s}(x') \rews{} \fS{s}(\fS{0})$ is $\sigma=\{x'\mapsto\fS{0}\}$;
furthermore, in order to satisfy the last condition $\fS{s}(x')\rews{}z$ we need $\sigma=\{x'\mapsto\fS{0},z\mapsto\fS{0}\}$;
however, $\sigma(\fS{f}(\fS{g}(x')))=\fS{f}(\fS{g}(\fS{0}))$ is irreducible; thus
$\fS{f}(\fS{g}(\fS{0}))\rews{}\fS{0}$ is \emph{not} satisfied.
By Proposition \ref{PropJoinabilityOfCCPsInGTRSs},
the conditional critical pair (\ref{ExCOPS_387_CS_CTRS_pCCP}) is not joinable.
By Theorem \ref{TheoLocalConfluenceCSCTRSsWithExtendedCCPs}.(\ref{TheoLocalConfluenceCSCTRSsWithExtendedCCPs_CSCTRS}), $\cR$
in Example \ref{ExCOPS_387_CS_CTRS}
is not locally $\ol{\cR}$-confluent nor $\ol{\cR}$-confluent.
As for $\cR_\bot$, since $\pCCP(\cR_\bot)=\iCCP(\cR_\bot)=\CVP(\cR_\bot)=\emptyset$ (see Example \ref{ExCOPS_387_CS_CTRS_Rbot_NoECCPs}),
by Theorem \ref{TheoLocalConfluenceCSCTRSsWithExtendedCCPs}.(\ref{TheoLocalConfluenceCSCTRSsWithExtendedCCPs_CSCTRS}),
$\cR_\bot$ is locally confluent.
Since $\cR_\bot$ is terminating, by Newman's Lemma, $\cR_\bot$ is confluent.
\end{example}
The following results, originally established in \cite{GutLucVit_ConfluenceOfConditionalRewritingInLogicForm_FSTTCS21} for \ctrs{s}, can also be used with \gtrs{s}.

\begin{proposition}\label{PropJoinabilityOfCCPs}
Let $\cR$ be a \gtrs{} and $\pi:\langle s,t\rangle\IF\gencond$ be a conditional pair
with variables $\vec{x}$, and $z$ be a variable not in $\vec{x}$.
Then,
\begin{enumerate}
\item\label{PropJoinabilityOfCCPs_Coro17_FSTTCS21}
 (cf.\ \cite[Corollary 17]{GutLucVit_ConfluenceOfConditionalRewritingInLogicForm_FSTTCS21})
If $\ol{\cR}\vdash(\forall\vec{x})(\exists z)\:\gencond\Rightarrow s\to^*z\wedge t\to^*z$
holds,
then $\pi$ is joinable.
 \item\label{PropJoinabilityOfCCPs_Coro18_FSTTCS21}
 (cf.\ \cite[Corollary 18]{GutLucVit_ConfluenceOfConditionalRewritingInLogicForm_FSTTCS21})
If $s^\downarrow\to^*z\wedge t^\downarrow\to^*z$
is feasible,
then $\pi$ is joinable.
\item\label{PropJoinabilityOfCCPs_Prop22_FSTTCS21}
 (cf.\ \cite[Proposition 22]{GutLucVit_ConfluenceOfConditionalRewritingInLogicForm_FSTTCS21})
 If $\Var(\gencond)\cap\Var(s,t)=\emptyset$, then $\pi$ is joinable iff
 $\gencond$ is infeasible or $s^\downarrow\to^*z,t^\downarrow\to^*z$ is feasible.
\end{enumerate}
\end{proposition}
The tool \infChecker{} \cite{GutLuc_AutomaticallyProvingAndDisprovingFeasibilityConditions_IJCAR20}
can be used to automatically prove (in)feasibility of sequences.

\begin{example}[Continuing Example \ref{Ex4_IWC23_useOfP_CnvJ}]
\label{Ex4_IWC23_joinabilityCP}
The (unconditional) critical pair (\ref{Ex4_IWC23_CP}), i.e.,
$\langle\fS{inc}(\fS{tl}(x\fS{:}\fS{from}(\fS{s}(x)))),\fS{tl}(\fS{inc}(\fS{from}(x)))\rangle$,
obtained in Example \ref{Ex4_IWC23_useOfP_CnvJ} for $\cR_2$ in
Example \ref{Ex4_IWC23} is $\ol{\cR}_2^\mu$-joinable:
\[
\begin{array}{rcl}
\fS{inc}(\ul{\fS{tl}(x\fS{:}\fS{from}(\fS{s}(x)))}) & \rew{\cR_2^\mu} & \fS{inc}(\fS{from}(\fS{s}(x)))
\end{array}
\]
and
\[\fS{tl}(\fS{inc}(\ul{\fS{from}(x)}))
\rew{\cR_2^\mu}
\fS{tl}(\ul{\fS{inc}(x\fS{:}\fS{from}(\fS{s}(x))}))
\rew{\cR_2^\mu}
\ul{\fS{tl}(\fS{s}(x)\fS{:}\fS{inc}(\fS{from}(\fS{s}(x))))}
 \rew{\cR_2^\mu}
 \fS{inc}(\fS{from}(\fS{s}(x))).\]
Thus, we have $\ProcJoinability(\JOproblem{\cR_2^\mu,(\ref{Ex4_IWC23_CP})})=\emptyset $.
\end{example}

\begin{example}[Continuing Example \ref{ExCOPS_387_CS_CTRS_pCCPs_iCCPs_CVPs}]
\label{ExCOPS_387_CS_CTRS_pCCPs_iCCPs_CVPs_joinability}
Consider the O-CTRS $\cR$ in Example \ref{ExCOPS_387_CS_CTRS} and the conditional critical pair (\ref{ExCOPS_387_CS_CTRS_pCCP}), i.e.,
$\langle\fS{f}(\fS{g}(x')) , \fS{s}(x')\rangle \IF \fS{s}(x') \cto{} \fS{s}(\fS{0})$.
This conditional pair is \emph{not} joinable as the only way to satisfy the reachability condition $\fS{s}(x') \cto{} \fS{s}(\fS{0})$ in the
conditional part is using $\sigma=\{x'\mapsto\fS{0}\}$.
However, $\sigma(\fS{f}(\fS{g}(x')))=\fS{f}(\fS{g}(\fS{0})) , $ and $\sigma(\fS{s}(x'))=\fS{s}(\fS{0})$ are both irreducible.
Alternatively, using Proposition \ref{PropJoinabilityOfCCPsInGTRSs}
and heuristic H1 in Remark \ref{RemPropJoinabilityOfConditionalPairsBySatisfiability_nonJoinability},
the non-joinability of (\ref{ExCOPS_387_CS_CTRS_pCCP}) can be proved as the $\ol{\cR}$-\emph{infeasibility} of
the sequence
\begin{eqnarray}
\fS{s}(x') \rews{} \fS{s}(\fS{0}), \fS{f}(\fS{g}(x'))  \rews{} z, \fS{s}(x') \rews{} z
\end{eqnarray}
where $z$ is a fresh variable.
This can be proved by \infChecker{} (use the
input in Figure~~ \ref{FigExCOPS_387_CS_CTRS_pCCP_nonJoinable}).

\begin{figure}[!h]
\vspace*{2mm}
\begin{verbatim}
(PROBLEM INFEASIBILITY)
(CONDITIONTYPE ORIENTED)
(VAR x)
(RULES
  f(g(x)) -> x | x == s(0)
  g(s(x)) -> g(x)
)
(VAR x' z)
(CONDITION s(x') ->* s(0), f(g(x')) ->* z, s(x') ->* z
)
\end{verbatim}\vspace*{-5mm}
\caption{Non-joinability as infeasibility in Example \ref{ExCOPS_387_CS_CTRS_pCCPs_iCCPs_CVPs_joinability}}
\label{FigExCOPS_387_CS_CTRS_pCCP_nonJoinable}
\end{figure}
\end{example}

\begin{remark}[Simple methods for joinability]
\label{RemSimpleMethodsForJoinability}
The previous methods, which are based on translating joinability proofs into (in)feasibility proofs,
heavily rely on the use of \infChecker{} to solve them.
Since this is costly, some exploration  techniques for concluding joinability are used.
For \emph{unconditional} pairs $\pi:\langle s,t\rangle$,
\begin{enumerate}
\item[J0] If $s$ and $t$ are \emph{irreducible} and $s\neq t$, then
$\pi$ is \emph{not} joinable.
\end{enumerate}
Given an unconditional pair $\pi:\langle s,t\rangle$ or a feasible
conditional pair $\pi:\langle s,t\rangle\IF\gencond$,
\begin{enumerate}
\item[J1] If $s$ and $t$ are  both ground and \emph{irreducible} and $s\neq t$, then $\pi$ is \emph{not} joinable.
\item[J2] If $\successorsBy{\rew{\cR}}(s)\cap\successorsBy{\rew{\cR}}(t)\neq\emptyset$, then
$\pi$ is joinable.
\item[J3] If $\directSuccessorsBy{\rew{\cR}}(t)\cap \successorsBy{\rew{\cR}}(s)\neq\emptyset$ and $\directSuccessorsBy{\rew{\cR}}(s)\cap \successorsBy{\rew{\cR}}(t)\neq\emptyset$, then $\pi$ is strongly joinable.
\end{enumerate}
\end{remark}

\begin{remark}[Checking irreducibility in J0 and J1]
Dealing with TRSs $\cR$, we can check \emph{irreducibility} of a term $u$
by just showing that $u$ contains
\emph{no redex}, i.e., no instance $\sigma(\lhsr)$ of a left-hand side $\lhsr$
of a rule $\lhsr\to\rhsr$ in $\cR$ occurs in $u$.
Dealing with \gtrs{s} $\cR$,
this simple test is not enough:
there can be terms including instances $\sigma(\lhsr)$ of the left-hand side $\lhsr$
of a \emph{conditional} rule
$\lhsr\to\rhsr\IF\gencond$
(they are called \emph{preredexes} \cite{BerKlo_ConditionalRewriteRulesConfluenceAndTermination_JCSS86}), which are also \emph{irreducible} if
$\deductionInThOf{\rewtheoryOf{\cR}}{\sigma(\gencond)}$ does \emph{not} hold.
This is implemented by using a model generator (e.g.,
\AGES{} \cite{GutLuc_AutomaticGenerationOfLogicalModelsWithAGES_CADE19}
or
\MaceFour{} \cite{McCune_Prove9andMace4_Unpublished10}) to (try to) find a
\emph{countermodel} of $(\forall\vec{x})\sigma(\gencond)$ in $\rewtheoryOf{\cR}$,
where $\vec{x}$ are the variables occurring in $\sigma(\gencond)$.
We proceed  in two steps:
\begin{enumerate}
\item If $u$ contains no (pre)redex, then it is irreducible.
\item If $u$ contains a preredex $\sigma(\lhsr)$  of a conditional rule $\lhsr\to\rhsr\IF A_1,\ldots,A_n$ and the theory
\[\rewtheoryOf{\cR}\cup\{\neg(\forall\vec{x})~\sigma(A_1)\wedge\cdots\wedge\sigma(A_n)\},\]
for
$\vec{x}$ the variables occurring in $\sigma(A_1),\cdots,\sigma(A_n)$ is
\emph{satisfiable}, then $u$ is irreducible.
\end{enumerate}
\end{remark}
\begin{example}[Continuing Example \ref{COPS409Ex_KBproc}]
\label{COPS409Ex_Joinability}
The  proper \CCP{} (\ref{COPS409Ex_pCCP}), i.e.,
$\langle \fS{h}(\fS{s}(x)), \fS{g}(\fS{s}(x))\rangle \IF  \fS{c}(\fS{g}(x)) \cto   \fS{c}(\fS{a}),\fS{c}(\fS{h}(x)) \cto \fS{c}(\fS{a})$,
is joinable, as we have
$\fS{h}(\fS{s}(x)) \rew{(\ref{COPS409Ex_rule3})} x$ and
$\fS{g}(\fS{s}(x)) \rew{(\ref{COPS409Ex_rule2})}  x$.
Thus, 
$\ProcJoinability(\JOproblem{\cR',(\ref{COPS409Ex_pCCP})})=\emptyset$.
\end{example}

\medskip
\noindent
Figures \ref{FigProofTreesInTheConfluenceFramework} and \ref{FigProofTreesInTheConfluenceFramework2} display the proof trees of the confluence framework for the examples in the paper.

  \begin{figure}[!ht]
  \vspace*{2mm}
    \scalebox{0.95}{
    \begin{tabular}{ccc}
   \begin{minipage}{7cm}
    \centering
      \begin{tikzpicture}[node distance = 6cm, auto]

      \node [blockW] (tIn) {$\CRproblem{\cR}$};
      \node [below right = 0.5cm and 0.2cm] at (tIn) {$\ProcLocalConfluence$};
      \node [below = 1.1cm] at (tIn) [blockW] (t0) {$\WCRproblem{\cR}$};
       \node [below left = 1.1cm and 0.7cm] at (t0) [blockW] (t1) {$\JOproblem{\cR,(\ref{ExCOPS_387_CS_CTRS_pCCP})}$}{
                              child {node [blockW] (t11) {\noTree}}
      };

      \node [below right = 1.1cm and 0.7cm] at (t0) [blockW] (t2) {$\JOproblem{\cR,(\ref{ExCOPS_387_CS_CTRS_CVP})}$} ;

     \node [below = 0.5cm] at (t0) {$ \ProcHuetExtended $};
    \node [below right = 0.5cm and 0.2cm] at (t1) {$\ProcJoinability$};

        \draw [line] (tIn) [bend right=0] to (t0) ;
        \draw [line] (t0) [bend right=0] to (t1) ;
    \draw [line] (t0) [bend right=0] to (t2) ;
    \draw [line] (t1) [bend right=0] to (t11) ;

    \end{tikzpicture}
     \end{minipage}
      &
     \begin{minipage}{7cm}
     \centering
      \begin{tikzpicture}[node distance = 6cm, auto]

      \node [blockW] (t1) {$\CRproblem{\cR_\bot}$};
      \node [below = 1.1cm] at (t1) [blockW] (t11) {$\WCRproblem{\cR_\bot}$};

      \node [below right = 0.5cm and 0.2cm] at (t1) {$\ProcKnuthBendix$};
    \node [below right = 0.5cm and 0.2cm] at (t11) {$\ProcHuetExtended$};
     \node [below = 1.1cm] at (t11) [blockW] (t111) {\yesTree};

        \draw [line] (t1) [bend right=0] to (t11) ;
        \draw [line] (t11) [bend right=0] to (t111) ;
     \end{tikzpicture}

\end{minipage}
\\
   CTRS $\cR$ in Example~\ref{ExCOPS_387_CS_CTRS};
&
     CS-CTRS $\cR_\bot$ in Example~\ref{ExCOPS_387_CS_CTRS};\\
   see also
   Examples \ref{ExCOPS_387_CS_CTRS_pCCPs_iCCPs_CVPs},
   \ref{ExCOPS_387_CS_CTRS_PHuet}, and \ref{ExCOPS_387_CS_CTRS_pCCP_notJoinable}
 &
      see also
     Examples \ref{ExCOPS_387_CS_CTRS_Rbot_NoECCPs} and
     \ref{ExCOPS_387_CS_CTRS_PHuet_Rbot_PKnuthBendix}
 \end{tabular}  }
     \caption{Proof trees for confluence problems of \ctrs{s} and \csctrs{s} in the confluence framework}
     \label{FigProofTreesInTheConfluenceFramework}

\vspace*{10mm}
\scalebox{0.95}{
  \begin{tabular}{cc}
     \begin{minipage}{7cm}
     \centering
      \begin{tikzpicture}[node distance = 6cm, auto]

     \node [blockW] (t1) {$\CRproblem{\cR}$};
     \node [below = 1.1cm] at (t1) [blockW] (t11) {$\CRproblem{\cR'}$};
     \node [below = 1.1cm] at (t11) [blockW] (t111) {$\JOproblem{\cR',(\ref{COPS409Ex_pCCP})}$};
     \node [below = 1.1cm] at (t111) [blockW] (t1111) {\yesTree};

      \node [below right = 0.5cm and 0.2cm] at (t1) {$\ProcSimplify$};
      \node [below right = 0.5cm and 0.2cm] at (t11) {$ \ProcKnuthBendix $};
     \node [below right = 0.5cm and 0.2cm] at (t111) {$\ProcJoinability$};

    \draw [line] (t1) [bend right=0] to (t11) ;
    \draw [line] (t11) [bend right=0] to (t111) ;
    \draw [line] (t111) [bend right=0] to (t1111) ;

    \end{tikzpicture}
    \end{minipage}
  &
     \begin{minipage}{7cm}
     \centering
      \begin{tikzpicture}[node distance = 6cm, auto]

      \node [blockW] (t0) {$\CRproblem{\cR}$};
       \node [below left = 1.1cm and 0.7cm] at (t0) [blockW] (t1) {$\CRproblem{\cR_1}$}{
                              child {node [blockW] (t11) {\yesTree}}
      };

      \node [below right = 1.1cm and 0.7cm] at (t0) [blockW] (t2) {$\CRproblem{\cR_2}$} ;
      \node [below = 1.1cm] at (t2)[blockW] (t21) {$\WCRproblem{\cR^\mu_2}$} ;
     \node [below = 1.1cm] at (t21)[blockW] (t211) {$\JOproblem{\cR^\mu_2,(\ref{Ex4_IWC23_CP})}$} ;
    \node [below = 1.1cm] at (t211)[blockW] (t2111) {\yesTree} ;

      \node [below = 0.5cm] at (t0) {$\ProcModularDecomp$};
      \node [below right = 0.5cm and 0.2cm] at (t1) {$ \ProcOrthogonality $};
      \node [below right = 0.5cm and 0.2cm] at (t2) {$ \ProcCnvJoinability $};
     \node [below right = 0.5cm and 0.2cm] at (t21) {$\ProcHuetExtended$};
     \node [below right = 0.5cm and 0.2cm] at (t211) {$\ProcJoinability$};

    \draw [line] (t0) [bend right=0] to (t1) ;
    \draw [line] (t1) [bend right=0] to (t11) ;
    \draw [line] (t0) [bend right=0] to (t2) ;
    \draw [line] (t2) [bend right=0] to (t21) ;
    \draw [line] (t21) [bend right=0] to (t211) ;
    \draw [line] (t211) [bend right=0] to (t2111) ;
    \end{tikzpicture}
\end{minipage}

    \\
    CTRS $\cR$ in Example~\ref{COPS409Ex};
&
     TRS $\cR$ in Example~\ref{Ex4_IWC23};\\
see also Examples \ref{COPS409Ex_KBproc} and \ref{COPS409Ex_Joinability}
&
see also Examples \ref{Ex4_IWC23_PHuetLevy}, \ref{Ex4_IWC23_useOfP_CnvJ},
and \ref{Ex4_IWC23_joinabilityCP}
\end{tabular} }
     \caption{Proof trees for confluence problems of \ctrs{s} and \trs{s} in the confluence framework}
     \label{FigProofTreesInTheConfluenceFramework2}\vspace*{-3mm}
   \end{figure}

\section{Strategy}\label{SecStrategy}

\CONFident{} implements several processors.
Given a \gtrs{} $\cR$, the processors enumerated in Section
\ref{SecListOfProcessors} are used to build a proof tree with root
$\CRproblem{\cR}$
to hopefully conclude confluence or non-confluence of $\cR$
(Theorem \ref{ThConfFramework}).
The selection and combination of processors to generate such a proof tree is
usually encoded as a  fixed \emph{proof strategy} which is
applied to the initial confluence problem $\CRproblem{\cR}$.
Choosing the appropriate proof strategy for an input problem is not a trivial task.
The strategy must take into account that many
proof obligations involved in the implementation of processors
are \emph{undecidable}.
For instance, joinability of conditional pairs (required, e.g., by $\ProcJoinability$)
is, in general, undecidable.
Also, calls to external tools trying to check undecidable properties like termination
(as required, e.g., by $\ProcKnuthBendix$) may fail or succeed in proving the \emph{opposite} property, i.e.,
non-termination.
For these reasons, the application of processors is usually constrained by a \emph{timeout} so that after a predefined
amount of time, the strategy may try a different processor
hopefully succeeding on the considered problem, or even  \emph{backtrack} to a previous problem in the proof tree.
Typically, a thorough
experimental analysis is needed to obtain a suitable strategy.

\begin{figure}[!ht]
\begin{center}
\scalebox{0.65}{
      \begin{tikzpicture}[node distance = 6cm, auto]
       \node [blockW] (CR) {Confluence};
      \node [below = 1cm] at (CR) [blockW] (ProcExtraVars) {$\ProcExtraVars$} ;
       \node [left = 1cm] at (ProcExtraVars) [blockW] (ProcExtraVarsNO) {\noTree} ;
      \node [below = 1cm] at (ProcExtraVars) [blockW] (ProcSimplify) {$\ProcSimplify$} ;
      \node [below = 1cm] at (ProcSimplify) [blockW] (ProcModularDecomp) {$\ProcModularDecomp$} ;
       \node [below = 1cm] at (ProcModularDecomp) [blockW] (ProcOrthogonality) {$\ProcOrthogonality$} ;
       \node [left = 1cm] at (ProcOrthogonality) [blockW] (ProcOrthogonalityYES) {\yesTree} ;
      \node [below = 1cm] at (ProcOrthogonality) [blockW] (ProcStrongConfluence) {$\ProcStrongConfluence$} ;
      \node [below = 1cm] at (ProcStrongConfluence) [blockW] (ProcKnuthBendix) {$\ProcKnuthBendix$} ;
      \node [below = 1cm] at (ProcKnuthBendix) [blockW] (ProcCnvJoinability) {$\ProcCnvJoinability$} ;
      \node [below = 1cm] at (ProcCnvJoinability) [blockW] (ProcCanCSR) {$\ProcCanCSR$} ;

    \draw [line] (ProcExtraVars) [bend right=0] to (ProcExtraVarsNO) ;
    \draw [line] (ProcExtraVars) [bend right=0] to (ProcSimplify) ;
    \draw [line] (ProcSimplify) [bend right=0] to (ProcModularDecomp) ;
    \draw [line] (ProcModularDecomp) [bend right=0] to (ProcOrthogonality) ;
    \draw [line] (ProcOrthogonality) [bend right=0] to (ProcOrthogonalityYES) ;
     \path [-angle 90]   (ProcOrthogonality) edge[dashed] (ProcStrongConfluence) ;
      \path [-angle 90] (ProcStrongConfluence)  edge[dashed] (ProcKnuthBendix) ;
     \path [-angle 90]  (ProcKnuthBendix) edge[dashed] (ProcCnvJoinability) ;
      \path [-angle 90]   (ProcCnvJoinability) edge[dashed] (ProcCanCSR) ;
      \path [-angle 90, bend left=70]   (ProcCanCSR) edge[dashed] (ProcOrthogonality) ;

            \node [right = 3.5cm] at (CR) [blockW] (WCR) {Local Confluence};
      \node [below = 1cm] at (WCR) [blockW] (ProcExtraVarsWCR) {$\ProcExtraVars$} ;
       \node [left = 1cm] at (ProcExtraVarsWCR) [blockW] (ProcExtraVarsWCRNO) {\noTree} ;
      \node [below = 1cm] at (ProcExtraVarsWCR) [blockW] (ProcSimplifyWCR) {$\ProcSimplify$} ;
      \node [below = 1cm] at (ProcSimplifyWCR) [blockW] (ProcConfluenceWCR) {$\ProcConfluence$} ;
      \node [below = 1cm] at (ProcConfluenceWCR) [blockW] (ProcHuetExtendedWCR) {$\ProcHuetExtended$} ;
      \node [below right = 1.5cm and 0.2cm] at (ProcHuetExtendedWCR) {$\neq\emptyset$};
     \node [above left = 0.2cm and 0.6cm] at (ProcHuetExtendedWCR) {$\emptyset$};
       \node [left = 1cm] at (ProcHuetExtendedWCR) [blockW] (ProcHuetExtendedWCRYES) {\yesTree} ;

    \draw [line] (ProcExtraVarsWCR) [bend right=0] to (ProcExtraVarsWCRNO) ;
    \draw [line] (ProcExtraVarsWCR) [bend right=0] to (ProcSimplifyWCR) ;
    \draw [line] (ProcSimplifyWCR) [bend right=0] to (ProcConfluenceWCR) ;
    \draw [line] (ProcConfluenceWCR) [bend right=0] to (ProcHuetExtendedWCR) ;
   \path [-angle 90]  (ProcConfluenceWCR) edge[dashed] (ProcModularDecomp) ;
     \draw [line] (ProcKnuthBendix) [bend right=10] to (ProcHuetExtendedWCR) ;
     \draw [line] (ProcCnvJoinability) [bend right=10] to (ProcHuetExtendedWCR) ;
     \draw [line] (ProcCanCSR) [bend right=10] to (ProcHuetExtendedWCR) ;
     \draw [line] (ProcHuetExtendedWCR) [bend right=0] to (ProcHuetExtendedWCRYES) ;

           \node [left = 3.5cm] at (CR) [blockW] (SCR) {Strong Confluence};
      \node [below = 1cm] at (SCR) [blockW] (ProcExtraVarsSCR) {$\ProcExtraVars$} ;
       \node [left = 1cm] at (ProcExtraVarsSCR) [blockW] (ProcExtraVarsSCRNO) {\noTree} ;
      \node [below = 1cm] at (ProcExtraVarsSCR) [blockW] (ProcSimplifySCR) {$\ProcSimplify$} ;
        \node [below = 1cm] at (ProcSimplifySCR) [blockW] (ProcConfluenceSCR) {$\ProcConfluence$} ;
      \node [below = 1cm] at (ProcConfluenceSCR) [blockW] (ProcHuetExtendedSCR) {$\ProcHuetExtended$} ;
      \node [below right = 1.5cm and 0.2cm] at (ProcHuetExtendedSCR) {$\neq\emptyset$};
     \node [above left = 0.2cm and 0.6cm] at (ProcHuetExtendedSCR) {$\emptyset$};
       \node [left = 1cm] at (ProcHuetExtendedSCR) [blockW] (ProcHuetExtendedSCRYES) {\yesTree} ;

    \draw [line] (ProcExtraVarsSCR) [bend right=0] to (ProcExtraVarsSCRNO) ;
    \draw [line] (ProcExtraVarsSCR) [bend right=0] to (ProcSimplifySCR) ;
    \draw [line] (ProcSimplifySCR) [bend right=0] to (ProcConfluenceSCR) ;
    \draw [line] (ProcConfluenceSCR) [bend right=0] to (ProcHuetExtendedSCR) ;
   \path [-angle 90]  (ProcConfluenceSCR) edge[dashed] (ProcModularDecomp) ;
   \draw [line] (ProcStrongConfluence) [bend left=20] to (ProcHuetExtendedSCR) ;
     \draw [line] (ProcHuetExtendedSCR) [bend right=0] to (ProcHuetExtendedSCRYES) ;

           \node [below = 10cm] at (WCR) [blockW] (JO) {Joinability};
      \node [above = 2cm] at (JO) [blockW] (ProcJoinabilityJO) {$\ProcJoinability$} ;
       \node [below left = 1cm and 1cm] at (ProcJoinabilityJO) [blockW] (ProcJoinabilityJOYES) {\yesTree} ;
       \node [below right = 1cm and 1cm] at (ProcJoinabilityJO) [blockW] (ProcJoinabilityJONO) {\noTree} ;

     \draw [line] (ProcJoinabilityJO) [bend right=0] to (ProcJoinabilityJOYES) ;
     \draw [line] (ProcJoinabilityJO) [bend right=0] to (ProcJoinabilityJONO) ;
     \draw [line] (ProcHuetExtendedWCR) [bend right=0] to (ProcJoinabilityJO) ;

            \node [below = 10cm] at (SCR) [blockW] (SJO) {Strong Joinability};
      \node [above = 2cm] at (SJO) [blockW] (ProcJoinabilitySJO) {$\ProcJoinability$} ;
       \node [below left = 1cm and 1cm] at (ProcJoinabilitySJO) [blockW] (ProcJoinabilitySJOYES) {\yesTree} ;
       \node [below right = 1cm and 1cm] at (ProcJoinabilitySJO) [blockW] (ProcJoinabilitySJONO) {\noTree} ;

     \draw [line] (ProcJoinabilitySJO) [bend right=0] to (ProcJoinabilitySJOYES) ;
     \draw [line] (ProcJoinabilitySJO) [bend right=0] to (ProcJoinabilitySJONO) ;
     \draw [line] (ProcHuetExtendedSCR) [bend right=0] to (ProcJoinabilitySJO) ;

   \end{tikzpicture}
   }
   \end{center}\vspace*{-4mm}
   \caption{Proof strategy for (CS-)TRSs}
    \label{FigProofStrategyCSTRSs}

\vspace*{4mm}
\begin{center}
\scalebox{0.65}{
      \begin{tikzpicture}[node distance = 6cm, auto]
        \node [blockW] (CR) {Confluence};
      \node [below = 1cm] at (CR) [blockW] (ProcExtraVars) {$\ProcExtraVars$} ;
       \node [left = 1cm] at (ProcExtraVars) [blockW] (ProcExtraVarsNO) {\noTree} ;
       \node [below = 1cm] at (ProcExtraVars) [blockW] (ProcInlining) {$\ProcInlining$} ;
      \node [below = 1cm] at (ProcInlining) [blockW] (ProcSimplify) {$\ProcSimplify$} ;
      \node [below = 1cm] at (ProcSimplify) [blockW] (ProcTrUconf) {$\ProcTrUconf$} ;
      \node [below = 1cm] at (ProcTrUconf) [blockW] (ProcOrthogonality) {$\ProcOrthogonality$} ;
       \node [left = 1cm] at (ProcOrthogonality) [blockW] (ProcOrthogonalityYES) {\yesTree} ;
      \node [below = 1cm] at (ProcOrthogonality) [blockW] (ProcKnuthBendix) {$\ProcKnuthBendix$} ;
     \node [above left = 0.1cm and 0.5cm] at (ProcKnuthBendix) {$\emptyset$};
      \node [below right = -0.2cm and 1cm] at (ProcKnuthBendix) {$\neq\emptyset$};
       \node [left = 1cm] at (ProcKnuthBendix) [blockW] (ProcKnuthBendixYES) {\yesTree} ;
      \node [below = 1cm] at (ProcKnuthBendix) [blockW] (ProcLocalConfluence) {$\ProcLocalConfluence$} ;

    \draw [line] (ProcExtraVars) [bend right=0] to (ProcExtraVarsNO) ;
    \draw [line] (ProcExtraVars) [bend right=0] to (ProcInlining) ;
    \draw [line] (ProcInlining) [bend right=0] to (ProcSimplify) ;
    \draw [line] (ProcSimplify) [bend right=0] to (ProcTrUconf) ;
        \draw [line] (ProcTrUconf) [bend right=0] to (ProcOrthogonality) ;
    \draw [line] (ProcOrthogonality) [bend right=0] to (ProcOrthogonalityYES) ;
     \path [-angle 90]  (ProcOrthogonality) edge[dashed] (ProcKnuthBendix) ;
      \draw [line] (ProcKnuthBendix) [bend right=0] to (ProcKnuthBendixYES) ;
       \path [-angle 90]  (ProcKnuthBendix) edge[dashed] (ProcLocalConfluence) ;

            \node [right = 3.5cm] at (CR) [blockW] (WCR) {Local Confluence};
      \node [below = 1cm] at (WCR) [blockW] (ProcExtraVarsWCR) {$\ProcExtraVars$} ;
       \node [left = 1cm] at (ProcExtraVarsWCR) [blockW] (ProcExtraVarsWCRNO) {\noTree} ;
       \node [below = 1cm] at (ProcExtraVarsWCR) [blockW] (ProcInliningWCR) {$\ProcInlining$} ;
     \node [below = 1cm] at (ProcInliningWCR) [blockW] (ProcSimplifyWCR) {$\ProcSimplify$} ;
       \node [below = 1cm] at (ProcSimplifyWCR) [blockW] (ProcConfluenceWCR) {$\ProcConfluence$} ;
      \node [below = 1cm] at (ProcConfluenceWCR) [blockW] (ProcHuetExtendedWCR) {$\ProcHuetExtended$} ;
      \node [below right = 0.5cm and 0.2cm] at (ProcHuetExtendedWCR) {$\neq\emptyset$};
     \node [above left = 0.2cm and 0.6cm] at (ProcHuetExtendedWCR) {$\emptyset$};
       \node [left = 1cm] at (ProcHuetExtendedWCR) [blockW] (ProcHuetExtendedWCRYES) {\yesTree} ;

    \draw [line] (ProcExtraVarsWCR) [bend right=0] to (ProcExtraVarsWCRNO) ;
    \draw [line] (ProcExtraVarsWCR) [bend right=0] to (ProcInliningWCR) ;
    \draw [line] (ProcInliningWCR) [bend right=0] to (ProcSimplifyWCR) ;
    \draw [line] (ProcSimplifyWCR) [bend right=0] to (ProcConfluenceWCR) ;
    \draw [line] (ProcConfluenceWCR) [bend right=0] to (ProcHuetExtendedWCR) ;
    \path [-angle 90]  (ProcConfluenceWCR) edge[dashed] (ProcTrUconf) ;
     \draw [line] (ProcLocalConfluence) [bend right=10] to (ProcHuetExtendedWCR) ;
     \draw [line] (ProcKnuthBendix) [bend right=10] to (ProcHuetExtendedWCR) ;
     \draw [line] (ProcHuetExtendedWCR) [bend right=0] to (ProcHuetExtendedWCRYES) ;

            \node [below = 11cm] at (WCR) [blockW] (JO) {Joinability};
      \node [above = 2cm] at (JO) [blockW] (ProcJoinabilityJO) {$\ProcJoinability$} ;
       \node [below left = 1cm and 1cm] at (ProcJoinabilityJO) [blockW] (ProcJoinabilityJOYES) {\yesTree} ;
       \node [below right = 1cm and 1cm] at (ProcJoinabilityJO) [blockW] (ProcJoinabilityJONO) {\noTree} ;

     \draw [line] (ProcJoinabilityJO) [bend right=0] to (ProcJoinabilityJOYES) ;
     \draw [line] (ProcJoinabilityJO) [bend right=0] to (ProcJoinabilityJONO) ;
     \draw [line] (ProcHuetExtendedWCR) [bend right=0] to (ProcJoinabilityJO) ;
     \draw [line] (ProcKnuthBendix) [bend right=10] to (ProcJoinabilityJO) ;
    \end{tikzpicture}
   }
   \end{center}\vspace*{-4mm}
   \caption{Proof strategy for (CS-)CTRSs}
      \label{FigProofStrategyCSCTRSs}\vspace*{-5mm}
   \end{figure}

The proof strategies for (CS-)TRSs and (CS-)CTRSs are depicted in Figures \ref{FigProofStrategyCSTRSs} and \ref{FigProofStrategyCSCTRSs}, respectively.
The diagrams show how to deal with CR, WCR, SCR, JO, and SJO problems for (CS-)TRSs and (CS-)CTRSs in \CONFident, according to currently implemented techniques.
Starting from the box identifying the problem at stake, the sequence of used processors is displayed
and the order of application is indicated by means of arrows from one processor to the next one.
Some processors have two possible continuations. Some of them correspond to \emph{ending} applications
of the processor leading to finish some branch (represented as boxes enclosing \yesTree{} or \noTree) and
the decision depends on the result of the processor as indicated as a label in the branches (e.g.,
for $\ProcHuetExtended$, $\ProcOrthogonality$, $\ProcKnuthBendix$, $\ProcJoinability$).
In other cases, the continuous one is chosen first and the \emph{dashed} one is followed after a
\emph{failure} in the first option.

\begin{example}[Use of $\ProcKnuthBendix$]
For (CS-)TRSs $\cR$ (see Figure \ref{FigProofStrategyCSTRSs}), if $\cR$ is terminating, then $\ProcKnuthBendix$
produces a call to $\ProcHuetExtended(\WCRproblem{\cR})$.
Otherwise, $\ProcCnvJoinability$ is used.
For (CS-)CTRSs (see Figure \ref{FigProofStrategyCSCTRSs}),  if $\cR$ is terminating, then
\begin{enumerate}
\item If $\cR$ is a J-\ctrs{}, then, if $\pCCP(\cR)\cup\iCCP(\cR)=\emptyset$, then a positive answer \yesTree{}
is given; otherwise, if all pairs in $\pCCP(\cR)\cup\iCCP(\cR)$ are overlays, then calls to
$\ProcJoinability$ for each $\pi\in\pCCP(\cR)\cup\iCCP(\cR)$
are made.
\item If $\cR$ is not a J-\ctrs{}, or  $\pCCP(\cR)\cup\iCCP(\cR)$ is not a set of overlays, then
$\ProcKnuthBendix$
produces a call to $\ProcHuetExtended(\WCRproblem{\cR})$.
\end{enumerate}
Otherwise,
$\ProcLocalConfluence$ is called.
\end{example}
Note that, in proofs involving (CS-)CTRSs $\cR$, after applying
$\ProcInlining$ and $\ProcSimplify$, it is possible to obtain a
(CS-)TRS $\cR'$. In this case, the proof would continue in the corresponding point of
Figure \ref{FigProofStrategyCSTRSs}.

\section{Structure of \CONFident}\label{SecStructureCONFident}

\CONFident\ is written in \Haskell{} and it has more than 100 \Haskell{} files
with almost 15000 lines of code (blanks and comments not included). The
tool is used online through its web interface
in:

\begin{center}
    \url{http://zenon.dsic.upv.es/confident/}
\end{center}
Nowadays, only \emph{confluence problems} $\CRproblem{\cR}$ for $\csctrs{s}$ $\cR$
can be (explicitly) solved, i.e., the input system $\cR$ is treated as a \csctrs{} and a proof of
\emph{confluence} of $\cR$ is attempted.
Direct proofs of local or strong confluence, or (strong) joinability are not
possible yet.

\subsection{Input format}
The main input format to introduce rewrite systems in \CONFident{} is the COPS
format,\footnote{http://project-coco.uibk.ac.at/problems/}
 the
official input format of the Confluence Competition
(CoCo\footnote{http://project-coco.uibk.ac.at/}). The (older) TPDB format\footnote{https://www.lri.fr/\textasciitilde{}marche/tpdb/}
can also be used.
As an example, Figure \ref{FigExCOPS_387_CS_CTRS_COPS_TPDBencoding} displays the COPS and TPDB encoding of the
O-\csctrs{} $\cR_\bot$ in Example \ref{ExCOPS_387_CS_CTRS}.
\begin{figure}[t]\small
\begin{tabular}{cc}
\begin{minipage}{7cm}
\begin{verbatim}
(CONDITIONTYPE ORIENTED)
(REPLACEMENT-MAP
  (f )
  (g )
  (s )
)
(VAR x)
(RULES
  g(s(x)) -> g(x)
  f(g(x)) -> x | x == s(0)
)
\end{verbatim}
\end{minipage}
&
\begin{minipage}{7cm}
\begin{verbatim}
(STRATEGY CONTEXTSENSITIVE
  (f )
  (g )
  (s )
)
(VAR x)
(RULES
  g(s(x)) -> g(x)
  f(g(x)) -> x | x ->* s(0)
)
\end{verbatim}
\vspace{0.2cm}
\end{minipage}
\end{tabular}
\caption{The \csctrs{} $\cR_\bot$ in Example \ref{ExCOPS_387_CS_CTRS} in COPS (left) and TPDB (right) format}
\label{FigExCOPS_387_CS_CTRS_COPS_TPDBencoding}\vspace*{-3mm}
\end{figure}
Both formats provide a similar display of the example,
as they organize the \mbox{information} about the rewrite system in several blocks:
the  type of CTRS according to the evaluation of conditions (only in the COPS format, as the TPDB format permits
\emph{oriented} CTRSs only);
replacement map specification (within a section \verb$REPLACEMENT-MAP$ in the COPS format
and within a section \verb$STRATEGY CONTEXTSENSITIVE$ in the TPDB format);
variable declaration;
and
rule description.

\medskip
In \CONFident{}, the user can combine both formats if needed (avoiding
inconsistencies, of course), for example, using specific rewriting relations in
conditions in the COPS format or using the reserved word
\texttt{REPLACEMENT-MAP} to define a replacement map instead of the
\texttt{STRATEGY} block in the TPDB format.
\subsection{Use of external tools}
\CONFident{} uses specialized tools to solve auxiliary proof
	obligations. For instance,
	\begin{itemize}
		\item \infChecker{} is used by $\ProcSimplify$  to prove infeasibility of
		conditional rules which are then discarded from the analysis. It is also used by $\ProcHuetExtended$, $\ProcOrthogonality$, and $\ProcKnuthBendix$ to remove infeasible conditional critical pairs.
		Finally, \infChecker{} is also used to implement the tests of joinability and $\mu$-joinability of (conditional) pairs required by $\ProcJoinability$. For this purpose, \MaceFour{} \cite{McCune_Prove9andMace4_Unpublished10} and \AGES{} \cite{GutLuc_AutomaticGenerationOfLogicalModelsWithAGES_CADE19} are used from \infChecker.
		\item \muterm{} is used by  $\ProcKnuthBendix$ and $\ProcCnvJoinability$ to check termination of CTRSs
		and CS-TRSs.
		\item \ProverNine{} \cite{McCune_Prove9andMace4_Unpublished10} is used by $\ProcJoinability$ to prove joinability of conditional pairs.
		\item \FORT~\cite{RapMid_FORT2_0_IJCAR18} is used by $\ProcCanCSR$ to check
		whether a TRS is normalizing.
	\end{itemize}

Tools like \muterm{} or \infChecker{} are connected as Haskell libraries that
are directly used by \CONFident{}, the rest of the tools are used by capturing
external calls.

\subsection{Implementing the confluence framework in \Haskell}

\paragraph{Problems.}
To implement the confluence framework, we start with the notion of problem. From
the implementation point of view, a problem is just a data structure containing
all the needed information to check its associated property. We can use a data
type also to describe each kind of problem.

{\small{\begin{lstlisting}{haskell}
-- | Problem
data Problem typ p = Problem typ p

-- | Confluence Problem Type
data ConfProblem = ConfProblem

-- | Joinability Problem Type
data JoinProblem = JoinProblem
\end{lstlisting} } }

\vspace*{1.5mm}
However, because we want to give more flexibility and deal with variants of
rewrite systems homogeneously, we prefer to define type classes to describe the
common characteristics of each problem and define each variant of rewrite system
as an proper instance.

{\small{\begin{lstlisting}{haskell}
-- | Problems contains a rewrite system
class IsProblem typ problem | typ -> problem  where
    getProblemType :: Problem typ problem -> typ

-- | Confluence Problems contains a rewrite system
class IsConfProblem typ problem trs | typ -> problem trs  where
    getConfProblem :: Problem typ problem -> trs

-- | Joinability Problems have a list of critical pairs
class IsJoinProblem typ problem cp | typ -> problem cp  where
    getCPair :: Problem typ problem -> cp
\end{lstlisting} } }

\vspace*{1.5mm}
In these definitions, functional dependencies restrict a problem from being
linked to multiple variants of rewrite systems and critical pairs. An example of
instance can be CTRS:

{\small{\begin{lstlisting}{haskell}
-- A CTRS is a Problem
instance IsProblem ConfProblem CR where
  getProblemType (Problem ConfProblem _) = ConfProblem

-- A CTRS is a Confluence Problem
instance IsConfProblem ConfProblem CR CR where
  getConfProblem (Problem ConfProblem trs) = trs
\end{lstlisting} } }

\vspace*{1.5mm}
By using type classes, we can reuse methods among problems. In this case, a
joinability problem is also a confluence problem because they share common
methods.

{\small{\begin{lstlisting}{haskell}
-- A Pair (CR,CP) is a Problem
instance IsProblem JoinProblem JO where
  getProblemType (Problem JoinProblem _) = JoinProblem

-- A Conditional TRS and a (possibly conditional) Critical Pair is a Confluence Problem
instance IsConfProblem JoinProblem JO CR where
  getConfProblem (Problem JoinProblem jopair) = josystem jopair

-- A Conditional TRS and a (possible conditional) Critical Pair is a Joinability Problem
instance IsJoinProblem JoinProblem JO InCritPair where
  getCPair (Problem JoinProblem jopair) = joccpair jopair
\end{lstlisting} } }

\eject

\paragraph{Processors.}
After defining our problems, we need to find the way of defining processors in
practice. The notion of processors is defined as abstractly as possible,
allowing us to encapsulate every possible technique applied to a problem. Also, the
implementation is kept as abstract as possible:

\begin{lstlisting}{haskell}
-- | Each processor is an instance of the class 'Processor'. The
-- output problem depends of the input problem and the applied processor
class Processor tag o d | tag o -> d where
apply       :: tag -> o -> Proof d
\end{lstlisting}

Each processor has its own name (\texttt{tag}). In the implementation, when a
processor is applied, it yields a proof node. A proof contains the resulting set
of problems or the refutation together with the infomation about the obtaind
proof.

\medskip
An example of proccessor can be the following dummy processor, a successful
processor that returns a empty list of subproblems. Each processor returns
information about its proof:

\begin{lstlisting}{haskell}
-- | Processor that returns an empty list
data SuccessProcessor  = SuccessProcessor

-- | The information of the proof is just the input problem
data SuccessProcInfo problem = SuccessProcInfo { inSuccessProcInfoProblem :: problem }

-- | Returns an empty list of processors
instance Processor SuccessProcessor (Problem typ problem) (Problem typ problem) where
  apply SuccessProcessor inP
    = andP (SuccessProcInfo inP) inP []
\end{lstlisting}

\vspace*{1.5mm}
The processor $\ProcHuetExtended$ presented above accepts confluence problems as an input but
returns joinability problems instead. Thus, the instances have the following
form:

\begin{lstlisting}{haskell}
  instance Processor HuetProcessor (Problem ConfProblem problem1)
                                     (Problem JoinProblem problem2) where
  ...

  instance Processor CanCRProcessor (Problem ConfProblem problem1)
                                    (Problem ConfProblem problem2) where
  ...
\end{lstlisting}

Our framework is designed to be flexible enough to accommodate such processors.

We do not delve into implementation details of processors as they can vary
significantly for each technique.

\paragraph{Proof trees.}
A proof node consists of the result of the applied processor in the form of an
info data structure, the input problem, and a list of subproblems if the
processor returns a set of subproblems. Consequently, we require a structure
similar to the following:

\begin{lstlisting}{haskell}
-- | Proof Tree constructors
data ProofF k  =
    And      { procInfo :: SomeProcInfo, problem :: SomeProblem, subProblems :: [k] }
  | Success  { procInfo :: SomeProcInfo, problem :: SomeProblem }
  | Refuted  { procInfo :: SomeProcInfo, problem :: SomeProblem }
  | DontKnow { procInfo :: SomeProcInfo, problem :: SomeProblem }
  | Search !([k])
\end{lstlisting}

\texttt{Success} is not necessary because it represents an \texttt{And} node
with empty subproblems, but it aids us during the process of processing the
solution. \texttt{DontKnow} is used to indicate a failure in the application of
the processor, while \texttt{Search} is associated with the possibility of
returning multiple solutions in the proof tree. \texttt{Search} is closely tied
to how we define our strategies. It is important to note that \Haskell{} utilizes
lazy evaluation. As a result, we initially apply a strategy to the initial
problem, and only when we traverse it to obtain the solution, the different
processors are applied. Consequently, we may have a list of solutions, but
ultimately we select the first successful one.

\medskip
In our implementation, our proof is constructed by combining \texttt{ProofF} and
\texttt{Problem} structures. By utilizing \emph{free
monads}~\cite{Swiestra_DataTypesAlaCarte}, we can integrate the two types of
nodes within the proof construction. The pure nodes (\texttt{Problem} instances)
are represented as pure values encapsulated within the free monad, while the
impure nodes (\texttt{ProofF} instances) are represented as impure effects
within the free monad.

\begin{lstlisting}{haskell}
-- | 'Proof' is a Free Monad.
type Proof a = Free ProofF a
\end{lstlisting}

\vspace*{1.5mm}
Because in our framework we have several kinds of problems and many types of
processor answers, in the \texttt{ProofF} node we hide the type of problems
using \texttt{SomeProblem} for the problems and
\texttt{SomeProcInfo} for the processor answers. For this purpose, we use
Generalized Algebraic Data Types (GADTs):

\begin{lstlisting}{haskell}
-- | 'SomeProblem' hides the type of the Problem
data SomeProblem where
    SomeProblem :: p -> SomeProblem

-- | 'SomeProcInfo' hides the type of the Process Output
data SomeProcInfo where
    SomeProcInfo :: p -> SomeProcInfo

-- | wraps the type of the problem
someProblem :: p -> SomeProblem
someProblem = SomeProblem

-- | wraps the type of the processor output
someProcInfo :: p -> SomeProcInfo
someProcInfo = SomeProcInfo
\end{lstlisting}

\vspace*{1.5mm}
We define functions to simplify the returning of proofs by the processors:

\begin{lstlisting}{haskell}
-- | Return a success node
success :: procInfo -> problem1 -> Proof problem2
success pi p0 = Impure (Success (someProcInfo pi) (someProblem p0))

-- | Return a refuted node
refuted :: procInfo -> problem1 -> Proof problem2
refuted pi p0 = Impure (Refuted (someProcInfo pi) (someProblem p0))

-- | Return a dontKnow node
dontKnow :: procInfo -> problem1 -> Proof problem2
dontKnow pi p0 = Impure (DontKnow (someProcInfo pi) (someProblem p0))

-- | Return a list of subproblems
andP :: procInfo -> problem1 -> [problem2] -> Proof problem2
andP pi p0 [] = success pi p0
andP pi p0 pp = Impure (And (someProcInfo pi) (someProblem p0) (map return pp))
\end{lstlisting}

\paragraph{Strategies.}
Finally, to define proof strategies, we use two strategy combinators, the
\emph{sequential} combinator and the \emph{alternative} combinator. This is
implemented using the two functions\footnote{The strategies used in an
alternative combinator could be also executed in parallel.}

\begin{lstlisting}{haskell}
-- | And strategy combinator
(.&.) :: (a -> Proof b) -> (b -> Proof c) -> a -> Proof c
(.&.) = (>=>)

-- | Or strategy combinator
(.|.) :: (a -> Proof b) -> (a -> Proof b) -> a -> Proof b
(f .|. g) p = (f p) `mplus` (g p)
\end{lstlisting}

In this section, we have introduced two types of problems: \texttt{ConfProblem}
and \texttt{JoinProblem}. To enhance our strategy in practice, we can create new
problem types that encapsulate desirable properties. By doing so, we can define
specific strategies tailored to those problem types. This is particularly
relevant when dealing with termination, where we can transition from
\texttt{ConfProblem} to \texttt{TermConfProblem} and apply specialized
strategies aimed at achieving finiteness.

\medskip
The execution of our framework in the Main module can be summarized in the
following three lines:
\begin{lstlisting}{haskell}
  let proof = crewstrat timeout problem
  let sol = runProof proof
  putStr . show . pPrint $ sol
\end{lstlisting}
where \texttt{csrewstrat} is a proof strategy of processors connected with
combinators and \texttt{runProof} traverses the proof tree hopefully obtaining
the solution. In our case, we follow a Breadth-First Search strategy, but other alternatives can be
implemented.

\section{Experimental results}\label{SecExperimentalResults}

\CONFident{} participated in the
2022\footnote{\url{http://project-coco.uibk.ac.at/2022/}} and 2023\footnote{\url{http://project-coco.uibk.ac.at/2023/}} International
Confluence Competition (CoCo) in the categories \trs{}, \srs{}, \csr{},
and \ctrs.
Table \ref{TableExperimentalResultsCoCo2022} summarizes the results obtained by the tool in the different categories.
Detailed information can be obtained from competition full run web page which
can be reached, for the different categories and years, from the following URL:
\begin{center}
\url{http://cops.uibk.ac.at/results/}
\end{center}
\begin{table}[!ht]\small
\caption{\CONFident{} at the CoCo 2022 Full Run (left) and at the CoCo 2023 Full Run (right)}
\begin{center}
\begin{tabular*}{\textwidth}{ccc}
\begin{tabular}{c@{~~~}}
\hline
COPS cat.\\ \hline
TRS\\
CTRS\\
CS-TRS\\
CS-CTRS\\ \hline
\end{tabular}
&
\begin{tabular}{c@{~~~}c@{~~~}c@{~~~}c@{~~~}c}
\hline
Yes & No & Maybe & Solved & Total \\ \hline
108   & 138  & 331 & 246 & 577 \\
82  & 43 & 36 & 125 & 161   \\
25   & 87  & 64 & 112 & 176 \\
88   & 41  & 158 & 129 & 287 \\ \hline
\end{tabular}
&
\begin{tabular}{c@{~~~}c@{~~~}c@{~~~}c@{~~~}c@{~~~}c}
\hline
Yes & No & Maybe & Solved & Total \\ \hline
109   & 138  & 330 & 247 & 577 \\
81  & 38 & 42 & 119 & 161   \\
 49   & 83  & 44 & 132 & 176 \\
94   & 47  & 146 & 141 & 287 \\ \hline
\end{tabular}
\end{tabular*}
\end{center}
\label{TableExperimentalResultsCoCo2022}
\end{table}
\CONFident{} obtained good results in the CTRS and CSR categories that we describe in the following.

\begin{table}[!h]\small
\vspace{3mm}
  \caption{CTRS Category: CoCo 2022 Full Run (left) and CoCo 2023 Full Run (right)}
  \begin{center}
    \begin{tabular*}{\textwidth}{ccc}
\begin{tabular}{c}
  \hline
  Tool\\ \hline
  ACP\\
  CO3\\
  CONFident\\ \hline
\end{tabular}
&
\begin{tabular}{c@{~~~}c@{~~~}c@{~~~}c@{~~~}c}
  \hline
  Yes & No & Maybe & Solved & Total \\ \hline
  54   & 26  & 81 & 80     & 161 \\
  57   & 32  & 72 & 89 & 161   \\
  82  & 43 & 36 & 125 & 161   \\ \hline
\end{tabular}
&
\begin{tabular}{c@{~~~}c@{~~~}c@{~~~}c@{~~~}c}
  \hline
  Yes & No & Maybe & Solved & Total \\ \hline
  54   & 25  & 82 & 79     & 161 \\
  57   & 35  & 69 & 92 & 161   \\
  81  & 38 & 42 & 119 & 161   \\ \hline
\end{tabular}
\end{tabular*}
\end{center}
\label{TableExperimentalResultsCTRS}\vspace*{-3mm}
\end{table}

\paragraph{\CONFident{} at the CTRS category of CoCo 2022 and 2023.}
Table \ref{TableExperimentalResultsCTRS} summarizes the results on the COPS
collection of CTRSs used in the CoCo 2022 and 2023 \emph{full-run} which we use below to
provide an analysis of use of our processors.
There was a bug in the 2022 version of the tool which was fixed in the 2023 version.
From the second row in Table \ref{TableExperimentalResultsCTRS} we can see that a CTRS (actually COPS \#1289, as can be seen on CoCo 2023 full-run web page) was proved confluent by \CONFident{} 2022, but could \emph{not} be
handled by \CONFident{} 2023.
Also, 5 examples
(actually, COPS \#311, \#312, \#353, \#524, and \#1138) were proved \emph{non-confluent} 
by \CONFident{} 2022, but could \emph{not} be
handled by \CONFident{} 2023.
The reason is that the 2023 version of the tool could not deliver a proof within the $60''$ timeout of the competition.
It can be checked that the five examples can be proved by the online version of \CONFident{} if the default $120''$ timeout is used; furthermore, if a timeout of $60''$ is selected, the proofs are \emph{not} obtained.
By lack of time, we could not appropriately tune the 2023 version to reproduce on CoCo 2023 all good results obtained in CoCo 2022.

\medskip
Taking CoCo 2023 as a reference, we see that
\CONFident{} solves almost $74\%$ of the 161 CTRS problems in COPS mainly using the techniques described in \cite{Lucas_LocalConferenceOfConditionalAndGeneralizedTermRewritingSystems_JLAMP24,GutLucVit_ConfluenceOfConditionalRewritingInLogicForm_FSTTCS21} orchestrated within the confluence framework described here.
From the 161 examples, 28 of them (more than 17\%) were proved
by \CONFident{} only.
There were 5 problems (COPS \#286, \#311, \#312, \#353 and
\#406) that \CONFident{} could not prove but were proved by ACP.

\medskip
Still,  \CONFident{} obtained the first place in the CTRS category.

\paragraph{\CONFident{} at the CSR category of CoCo 2022 and 2023.}

With respect to \csr, a \emph{demonstration} subcategory of confluence of CSR
was hosted as part of CoCo 2022 and a competitive subcategory was hosted
at CoCo 2023, see \url{http://project-coco.uibk.ac.at/2022/categories/csr.php}.
\CONFident{} and \ConfCSR{} participated in 2023.
The results on the CoCo 2023 full-run for the \csr{} category,
consisting of 176 CS-TRSs
(COPS problems \#1161\:--\:\#1164; \#1167\:--\:\#1274; and \#1298\:--\:\#1361)
 and 287
\csctrs{s}
(COPS problems \#1362\:--\:\#1648)
are summarized in Table
\ref{TableExperimentalResultsCSR}, see
\url{http://cops.uibk.ac.at/results/?y=2023-full-run&c=CSR} for complete details,
where both CS-TRSs and CS-CTRSs are displayed in a single table.
\CONFident{} was
buggy in the $\ProcJoinability$ processor in 2022, but it was fixed and improved
in the 2023 version.

\begin{table}[h!]
  \caption{CSR Category: CoCo 2023 Full Run}
  \begin{center}
    \begin{tabular*}{\textwidth}{cc}
   \scalebox{0.83}{
   \begin{tabular}{c}
  \hline
  Tool\\ \hline
  ConfCSR\\
  CONFident 2022\\
  CONFident 2023\\ \hline
\end{tabular} }
&
 \scalebox{0.83}{
 \begin{tabular}{c@{~~~}c@{~~~}c@{~~~}c@{~~~}c@{~~~}c}
  \hline
  Yes & No & Maybe & Erroneous & Solved & Total \\ \hline
  28   & 83  & 352 & 0 & 111 & 463 \\
  113   & 124 & 222 & 4 & 237 & 463   \\
  143  & 130 & 190 & 0 & 273 & 463   \\ \hline
\end{tabular} }
\end{tabular*}
\end{center}
\label{TableExperimentalResultsCSR}
\end{table}

\vspace*{-8mm}
\begin{table}[!h]
  \caption{Confluence of \csr: \cstrs{s} (left) and \csctrs{s} (right)}
  \begin{center}
     \begin{tabular*}{\textwidth}{c@{~~}c@{~~}c}
 \scalebox{0.83}{
 \begin{tabular}{c}
  \hline
  Tool\\ \hline
  ConfCSR\\
  CONFident 2024\\ \hline
\end{tabular} }
&
 \scalebox{0.83}{
 \begin{tabular}{c@{~~~}c@{~~~}c@{~~~}c@{~~~}c@{~~~}c}
  \hline
  Yes & No & Maybe  & Solved & Total \\ \hline
  28   & 83  & 65 & 111 & 176 \\
  49  & 83 & 44  & 132 & 176   \\ \hline
\end{tabular} }
&
 \scalebox{0.83}{
 \begin{tabular}{c@{~~~}c@{~~~}c@{~~~}c@{~~~}c@{~~~}c}
  \hline
  Yes & No & Maybe  & Solved & Total \\ \hline
  --   & --  & -- & -- & -- \\
  105  & 93 & 89  & 198 & 287   \\ \hline
\end{tabular}  }
\end{tabular*}
\end{center}
\label{TableExperimentalResultsConfluenceCSR}\vspace*{-2mm}
\end{table}

Unfortunately, due to a recently discovered
bug in the implementation of the 2023 version of \CONFident, processor $\ProcTrUconf$ was used in proofs of confluence of \csctrs{s}, even though no theoretical result gives support to  this yet.
As a consequence, a number of \csctrs{s} were reported as confluent
under no solid basis.
We have fixed this problem and reproduced the full-run benchmarks with the
new 2024 version of \CONFident.
Table \ref{TableExperimentalResultsConfluenceCSR} summarizes the
obtained results.
We separately show the results for
\cstrs{s}\footnote{Complete details in
\url{http://zenon.dsic.upv.es/confident/benchmarks/fi24/cstrs/}} and
\csctrs{s}.\footnote{Complete details in
\url{http://zenon.dsic.upv.es/confident/benchmarks/fi24/csctrs/}}
Although \ConfCSR{}
does not handle \csctrs{s}, for \cstrs{s} we also include the results obtained by \ConfCSR{}
in the \emph{full-run} of the \csr{} category of CoCo 2023.

\medskip
Some observations follow:
\begin{enumerate}
\itemsep=0.9pt
\item The results on \cstrs{s} obtained by \CONFident{} 2023 \& 2024, and summarized in
Tables~\ref{TableExperimentalResultsCSR}~and ~\ref{TableExperimentalResultsConfluenceCSR}~are partly
due to the implementation of results only recently reported in \cite{Lucas_OrthogonalityOfGeneralizedTermRewritingSystems_IWC24}
which are included as part of processor $\ProcOrthogonality$
(not explicit in the description of the processor in Section \ref{SecOrthogonality},
but included in the statistical account of Table \ref{TableUseOfProcessorsInExperiments}).
\item Except for 8 cases,
(COPS \#1423, \#1424, \#1467, \#1583, \#1587, \#1588, \#1609 and  \#1610)
all positive answers of \CONFident{} 2023
for \csctrs{s} due to the buggy application of $\ProcTrUconf$ have been confirmed
by \CONFident{} 2024 by using other techniques.
This suggests that using $\ProcTrUconf$ with \csctrs{s} is probably correct;
a formal proof of this conjecture should be provided, though.
\item After removing the use of $\ProcTrUconf$ from \CONFident{} 2024  proof strategy for \csctrs{s}, we are able to handle \emph{more} \csctrs{s}: in 2023 we were able
to (dis)prove confluence of $273-132 = 141$ \csctrs{s}. In contrast, \CONFident{} 2024 solves 198
\csctrs{s}.
Since no other change has been introduced, this shows the impact of the
appropriate selection and order of use of processors in the proof strategy of
automated analysis tools.
\end{enumerate}

\begin{table}[!h]
\vspace*{-5mm}
  \caption{Use of processors in the experiments}
  \begin{center}
   \scalebox{0.85}{
\begin{tabular}{lc@{~~}c@{~~}c@{~~}c@{~~}c@{~~}c@{~~}c@{~~}c@{~~}c@{~~}c@{~~}c@{~~}c@{~~}c@{~~}c@{~~}c@{~~}c}
& \#Y
& \#N
& $\ProcExtraVars$
    & $\ProcSimplify$
    & $\ProcInlining$
    & $\ProcModularDecomp$
    & $\ProcTrUconf$
    & $\ProcOrthogonality$
    & $\ProcLocalConfluence$
    & $\ProcStrongConfluence$
    & $\ProcKnuthBendix$
     & $\ProcCanCSR$
    & $\ProcHuetExtended$
  & $\ProcJoinability$
    \\
    \hline
    TRS
  & 109
  & 138
  &  $0$
    & $37$
    & $0$
    & $20$
    & $-$
   & $46$
    & $137$
    & $37$
    & $35$
     & $16$
    & $175$
    & $518$
  \\
    CTRS
  & 81
  & 38
    &  $0$
    & $51$
    & $24$
    & $8$
    & $34$
   & $72$
    & $44$
    & $0$
    & $16$
    & $0$
    & $39$
    & $71$
  \\
  CS-TRS
  & 49
  & 83
    &  $0$
    & $0$
    & $0$
    & $-$
    & $-$
     & $49$
    & $87$
    & $0$
    & $0$
    & $-$
   & $83$
    & $84$
  \\
 CS-CTRS
  & 105
  & 93
    & $9$
     & 99
     & 40
     & $-$
    & $-$
     & 95
    & $32$
    & $0$
    & $46$
    & $-$
   & 59
    & 125
  \\
  \hline
  TOTAL
& $344$
& $349$
    &  $9$
     & $187$
     & $64$
     & $28$
    & $34$
     & $262$
    & $300$
    & $37$
    & $97$
    & $16$
   & $356$
    & $798$
\\
\hline
  \end{tabular} }
  \end{center}
  \label{TableUseOfProcessorsInExperiments}\vspace*{-6mm}
  \end{table}

\paragraph{Use of processors.}
Table \ref{TableUseOfProcessorsInExperiments} summarizes the use of processors
in the latest version of \CONFident{}.
We display the \emph{number of uses} of each processor along the whole benchmark set.
Some remarks are in order:
\begin{enumerate}
\item According to  Tables \ref{TableFinishingProcessorsInTheConfluenceFramework},\ref{TableUseOfCleansingAndModularProcessorsInTheConfluenceFramework},
and
\ref{TableUseOfConfluenceProcessorsInTheConfluenceFramework}:
\begin{itemize}
\item \emph{Proofs of confluence} eventually use the (sound)
processors
$\ProcOrthogonality$,
$\ProcHuetExtended$,
or $\ProcKnuthBendix$ on CR-problems in the proof tree.
This is because they are  able to either (i) finish proof branches with a positive result (e.g.,  $\ProcOrthogonality$ and $\ProcKnuthBendix$), or (ii) translate the original CR-problem into a WCR/SCR-problem (using $\ProcLocalConfluence$ or $\ProcStrongConfluence$) which is then handled by $\ProcHuetExtended$ and ultimately by $\ProcJoinability$  to obtain a positive result.\vspace*{1mm}

\item Proofs of \emph{non-confluence} are due to the use of (complete) processors on CR-problems like
$\ProcExtraVars$ (which finishes a proof branch with a negative result which is propagated to the root of the proof tree),
$\ProcLocalConfluence$ (which translates the CR-problem into a WCR-problem),
and then $\ProcHuetExtended$ (which returns JO-problems which can be qualified as \emph{negative} by $\ProcJoinability$).
\end{itemize}

\item Regarding the second row (CTRSs): although \CONFident{} does not implement any modularity result for CTRSs (yet),
processor $\ProcModularDecomp$ is used in some proofs with CTRSs, but only after transforming them
into TRSs by applying other processors, e.g., $\ProcSimplify$, $\ProcInlining$, or
$\ProcTrUconf$.
\item Processor $\ProcExtraVars$ is used in some proofs with \csctrs{s} (fourth row)
due to the presence of rules with extra variables in some of the COPS examples (see, e.g., \verb$COPS/1367.trs$).
\item Some processors were not used in the competition (e.g.,
$\ProcTrU$ or $\ProcCanJoinability$, possibly due to their ``low'' position in the strategy; or
$\ProcCnvJoinability$, which was not available yet). Thus, they are not
mentioned in the table.
\item Some processors are used several times to solve a single example.
For instance, in the CTRS category (see Table
\ref{TableExperimentalResultsCTRS}), the 38 proofs of non-confluence required 39
uses of $\ProcHuetExtended$ due to one of the decompositions introduced by
$\ProcModularDecomp$.
Similarly, the 81 positive proofs obtained required $72+16=88$ uses of $\ProcOrthogonality$ or
 $\ProcKnuthBendix$ due to the remaining 7 decompositions introduced by $\ProcModularDecomp$.
\end{enumerate}

\begin{table}[!h]
\vspace*{-4mm}
\caption{Use of heuristics and joinability methods within $\ProcJoinability$}
\begin{center}
\scalebox{0.85}{
\begin{tabular}{cccc|ccc|cc}
& H1 & H2 & H3 & J1 & J2 & J3 & \infChecker{} & Other\\
\hline
TRSs &
$0$ & $4$ & $0$ &
10 &
211 &
189 &
58 &
47
\\
CTRSs &
2 & 0 & 2 &
 5 &
 21  &
 0 &
 40 &
 6

\\
CS-TRSs &
0 & 4 & 0 &
 5 & 78 &  0 &
0 &
0
 \\
CS-CTRSs &
4 & 0 & 3 &
 14 & 11 & 0 &
4 &
175
\end{tabular} }
\end{center}
\label{TableUseOfHeuristicsAndJoinabilityTechniques}\vspace*{-8mm}
\end{table}

\paragraph{Use of heuristics and joinability techniques.}
Table \ref{TableUseOfHeuristicsAndJoinabilityTechniques} summarizes the use of heuristics
(H1, H2, and H3 for Proposition \ref{PropJoinabilityOfCCPsInGTRSs}, see Remark \ref{RemPropJoinabilityOfConditionalPairsBySatisfiability_nonJoinability}) and
joinability techniques (J1, J2, and J3, see Remark \ref{RemSimpleMethodsForJoinability}) in the implementation of $\ProcJoinability$
in CoCo 2022 (see Section \ref{SecJoinabilityProcessor}).
We display the \emph{number of uses} of each heuristic or technique for each category.
Besides, we report on uses of \infChecker{} to prove joinability using other results in
Proposition \ref{PropJoinabilityOfCCPs},
or \emph{other} techniques (e.g., use of
\MaceFour{} and \ProverNine).
Some remarks are in order:
\begin{enumerate}
\item First row (TRSs): H2 is used with TRSs due to the use of $\ProcCanCSR$ which transforms a confluence problem for a TRS into a confluence problem for a CS-TRS (see the first row in Table \ref{TableUseOfProcessorsInExperiments}).
This may lead to compute $\LH{\mu}$-critical pairs whose non-joinability is treated using
Proposition Proposition \ref{PropJoinabilityOfCCPsInGTRSs}
and hence the heuristics in Remark \ref{RemPropJoinabilityOfConditionalPairsBySatisfiability_nonJoinability}.
\item J3 (for strong joinability) is \emph{not} used with CTRSs, CS-TRSs, or \csctrs{s} because, for SCR-problems, $\ProcHuetExtended$ (which would eventually require it through a call to $\ProcJoinability$) is defined   for linear TRSs only, see Section
\ref{SecProcessorForLocalStrongConfluenceProblems}.
\end{enumerate}

\section{Related work}\label{SecRelatedWork}

Several tools can be used to automatically prove and disprove confluence of
TRSs and oriented CTRSs.
ACP (Automated Confluence Prover)~\cite{AotYosToy_ProvingConfluenceOfTermRewritingSystemsAutomatically_RTA09}
implements a divide-and-conquer approach in two steps: first, a decomposition
step is applied (based on modularity results);
then, direct techniques are applied to each decomposition.
ConCon~\cite{Sternagel_ConCon_IWC20} first tries to simplify rules and remove
infeasible rules from the input system, then it employs
a number of
confluence criteria for (oriented) 3-CTRSs.
ConCon uses several confluence criteria, and tree automata techniques on
reachability to prove (in)feasibility of conditional parts.
In parallel,
ConCon
tries  to show non-confluence using conditional narrowing (and some
other heuristics).
CSI~\cite{NagFelMid_CSInewEvidenceAProgressReport_CADE17} uses a set of
techniques
(Knuth and Bendix' criterion, non-confluence criterion, order-sorted
decomposition, development closed criterion, decreasing diagrams and extended
rules) and a strategy language to combine them.
CoLL-Saigawa~\cite{ShiHir_CoLLSaigawa_IWC21} is the combination of two tools:
CoLL and Saigawa. If the input system is left-linear, it uses CoLL;
otherwise,
it uses Saigawa. Among the techniques used by these tools are Hindley's
commutation theorem together with the three commutation criteria, almost
development closedness, rule labeling with weight function, Church-Rosser modulo
AC, criteria based on different kinds of
critical pairs, rule
labeling, parallel closedness based on parallel critical pairs, simultaneous
closedness, parallel-upside closedness, and outside closedness.
CO3~\cite{NisKurYanGme_CO3aCOnverterForProvingCOnfluenceOfCOnditionalTRSs_IWC15}
uses confluence (and termination) of $\cU(\cR)$ and $\cU_{opt}(\cR)$ (which is a
variant of $\cU(\cR)$, see \cite[Section
7]{SchGra_CharacterizingAndProvingOperationalTerminationOfDCTRSs_JLAP10} for a
discussion), in addition to narrowing trees for checking infeasibility of
conditional parts in proofs of confluence of CTRSs.
\Hakusan{}~\cite{ShiHir_Hakusan0_8aConfluenceTool_IWC23} is a confluence tool for left-linear TRSs that
analyzes confluence by using two compositional confluence
criteria~\cite{ShiHir_CompositionalConfluenceCriteria_FSCD22}. It returns
certified outputs for rule labeling and develops a novel reduction method.

To demonstrate confluence of CSR, apart from \CONFident, only \ConfCSR{}~\cite{SteMit_ConfCSR_IWC23} is able to prove confluence of \cstrs{s}
by using, essentially, the results in \cite{LucVitGut_ProvingAndDisprovingConfluenceOfContextSensitiveRewriting_JLAMP22}
while relying on
\AProVE{} to prove termination of CSR.

Although the previous tools also use combinations of different techniques to
obtain proofs of confluence, to the best of our knowledge, none of them
formalize such a combination as done in this paper, where the confluence framework
provides a systematic way to organize and decompose proofs of (local, strong) confluence problems by transforming them into other (possibly simpler) problems.
Also, an important difference between \CONFident{} and all previous tools is
the encoding of (non-)joinability of (conditional) pairs as combinations of
\emph{(in)feasibility} problems \cite{GutLuc_AutomaticallyProvingAndDisprovingFeasibilityConditions_IJCAR20}.
As explained in \cite{Lucas_ProvingSemanticPropertiesAsFirstOrderSatisfiability_AI19}
and further developed in \cite{Lucas_LocalConferenceOfConditionalAndGeneralizedTermRewritingSystems_JLAMP24,%
GutLucVit_ConfluenceOfConditionalRewritingInLogicForm_FSTTCS21,%
LucVitGut_ProvingAndDisprovingConfluenceOfContextSensitiveRewriting_JLAMP22},
this is possible due to the logic-based treatment of rewrite systems as first-order theories.
This is also the key for \CONFident{} to be able to smoothly handle join, oriented, and semi-equational
CTRSs, something which is also a novel feature of the tool in comparison with the aforementioned
tools.
Furthermore, as far as we know, \CONFident{} is the only tool implementing
confluence criteria depending on \emph{termination} of CTRSs
rather than more
restrictive termination properties like
quasi-decreasingness \cite[Def.\ 7.2.39]{Ohlebusch_AdvTopicsTermRew_2002},
see \cite[Section 8]{Lucas_LocalConferenceOfConditionalAndGeneralizedTermRewritingSystems_JLAMP24} for a deeper discussion and further motivation.

\section{Conclusions and future work}\label{SecConclusions}

CONFident is a tool which is able to automatically prove and disprove confluence of variants
of rewrite systems: TRSs, CS-TRSs, and Join, Oriented, and Semi-Equational CTRSs
and \csctrs{s}.
The proofs are obtained by combining different techniques in what we call \emph{Confluence Framework}, which we have introduced here,
where different variants of \emph{confluence} and \emph{joinability} problems are handled (simplified, transformed, etc.)\ by
means of \emph{processors}, which can be freely combined to obtain the proofs which are displayed as a
\emph{proof tree} (Definition \ref{DefConfProofTree} and
Theorem \ref{ThConfFramework}).
In this paper, we have introduced 16 processors which can be used in the confluence framework.
Some of them just integrate existing results by other authors as processors to be used in the confluence framework (see the description of the processors and the corresponding references to the used methods and supporting results in Section \ref{SecListOfProcessors}).
We believe that other results not considered here, possibly involving
\begin{itemize}
\item
other approaches like the \emph{compositional approach} developed in \cite{ShiHir_CompositionalConfluenceCriteria_FSCD22} which reformulates and somehow \emph{unifies}
a number of well-known confluence results for TRSs from an abstract compositional presentation, and then obtains improvements from them by applying the technique developed by the authors; or
\item
other kind of confluence problems (e.g., \emph{ground} confluence, i.e., confluence of rewriting restricted to \emph{ground} terms only); or
\item
rewriting forms (e.g., \emph{rewriting modulo} a set of equations \cite{PetSti_CompleteSetsOfReductionsForSomeEquationalTheories_JACM81},
\emph{nominal rewriting} \cite{FerGab_NominalRewriting_IC07}, etc.), see, e.g., \cite{DurMes_OnTheChurchRosserAndCoherencePropertiesOfConditionalOrderSortedRewriteTheories_JLAP12,%
Jouannaud_ConfluentAndCoherentEquationalTermRewritingSystemsApplicationToProofsInAbstractDataTypes_CAAP83,%
JouKir_CompletionOfASetOfRulesModuloASetOfEquations_SIAMJC86,%
Kikuchi_GroundConfluenceAndStrongCommutationModuloAlphaEquivalenceInNominalRewriting_ICTAC22,%
KikAo_ConfluenceAndCommutationForNominalRewritingSystemsWithAtomVariables_LOPSTR20,%
Lucas_ConfluenceOfConditionalRewritingModulo_CSL24}, etc.,
\end{itemize}
could also be
ported into the confluence framework.
This may involve the definition of new problems: for instance $\GCRproblem{\cR}$ for ground confluence of a \gtrs{} $\cR$ and $\ECRproblem{\cR}$ for confluence modulo of an \emph{Equational} \gtrs{} $\cR$ (which, essentially, are \gtrs{s} extended with a set of conditional equations, see \cite[Definition 10]{Lucas_ConfluenceOfConditionalRewritingModulo_CSL24}).
Whether new kind of problems have been considered or not, we think that
available or forthcoming techniques for proving such problems positive or negative
can be integrated in the confluence framework by means of
appropriate processors
so that they can be smoothly combined with other processors implementing different techniques to obtain a more powerful framework to prove confluence.

\medskip
We have shown how to implement the confluence framework using the declarative
programming language \Haskell. \CONFident{} has proved to be a powerful tool for
proving confluence of CTRSs.
This is witnessed by the first position obtained in the CTRS category  of
CoCo 2021, 2022 and 2023
and also by the good results obtained in the CSR category at CoCo 2023,
where both CS-TRSs and CS-CTRSs were considered.

\paragraph{Future work.}
First of all, \CONFident{} should be extended to deal with arbitrary \gtrs{s}.
Providing direct access to proofs of \emph{local} and \emph{strong} confluence
(and even joinability of terms) would also be interesting by using the specific problems introduced here and their treatment within the confluence framework.
Also, some features of the discussed processors are not implemented yet (e.g., modularity of local and strong confluence of TRSs) and we should also include existing results for CTRSs (see, e.g.,
\cite[Table 8.4]{Ohlebusch_OnTheModularityOfConfluenceOfConstructorSharingTRS_CAAP94}).
Besides, as far as we know, modularity of confluence of \csr, either for CS-TRSs or for
\csctrs{s} has not been investigated yet.
Also considering confluence of order-sorted rewrite systems \cite{GogMes_OrderSortedAlgebraI_TCS92}
and equational rewrite systems \cite{Jouannaud_ConfluentAndCoherentEquationalTermRewritingSystemsApplicationToProofsInAbstractDataTypes_CAAP83} could be an
interesting task for future  work.

\paragraph{Acknowledgements.}
We thank the referees for their detailed comments and suggestions that led to many improvements in the paper.
We thank Vincent van Oostrom for providing early access to \cite{Oostrom_TheZpropertyForLeftLinearTermRewritingViaConvectiveContextSensitiveCompleteness_IWC23}
and also for his comments on \cite{GutVitLuc_ConfluenceFrameworkProvingConfluenceWithCONFident_LOPSTR22}.
Finally, we thank the organizers of the Confluence Competition and maintainers of the COPS database for their (successful) effort to promote the research on confluence and automation of program analysis.

\medskip
Ra\'ul Guti\'errez was partially supported by the grant PID2021-122830OB-C44 funded by \linebreak CIN/AEI/10.13039/501100011033 and by the project "ERDF - a way to build Europe".

Salvador Lucas and Miguel V\'itores were
partially supported by grant PID2021-122830OB-C42 funded by MCIN/AEI/10.13039/501100011033
and by ``ERDF - a way of making Europe''. Salvador Lucas was also supported by the grant CIPROM/2022/6 funded by Generalitat Valenciana.

\appendix

\section{Proof of Proposition \ref{PropInliningGTRSs}}
\label{ApProofPropInliningGTRSs}

\noindent
\textbf{Proposition \ref{PropInliningGTRSs}}
\emph{
Let $\cR=(\Symbols,\SPredicates,\mu,H,R)$ be a \gtrs{},
$\alpha\in R$, $i$,
and $x$ as in Definition \ref{DefInlining}. Let $s$ and $t$ be terms.
\begin{enumerate}
\itemsep=0.95pt
\item  If $s\rew{\cR}t$,
then
$s\rews{\cR_{\alpha,x,i}}t$.
\item If $s\rew{\cR_{\alpha,x,i}}t$,
then
$s\rew{\cR}t$.
\end{enumerate}
}

\begin{proof}
As for the \emph{first claim} ($s\rew{\cR}t$ implies
$s\rews{\cR_{\alpha,x,i}}t$), by using the version for \gtrs{s} of \cite[Theorem 10]{Lucas_LocalConferenceOfConditionalAndGeneralizedTermRewritingSystems_JLAMP24} for CTRSs (see
\cite[Section 7.5]{Lucas_LocalConferenceOfConditionalAndGeneralizedTermRewritingSystems_JLAMP24}), $s\rew{\cR}t$
iff $\deductionInThOf{\rewtheoryOf{\cR}}{\grounding{s}\rew{}\grounding{t}}$.
We proceed by induction on the length $n\geq 0$
of a   \emph{Hilbert-like} proof of $\grounding{s}\rew{}\grounding{t}$
from $\rewtheoryOf{\cR}$.
Base: $n=0$.
With $\grounding{s}\rew{}\grounding{t}$
we are using $(\RuleHornClause)_\alpha$ for some unconditional rule
$\alpha:\lhsr\to\rhsr\in R$ such that
$\grounding{s}=\sigma(\lhsr)$ and $\grounding{t}=\sigma(\rhsr)$.
Since all unconditional rules in $\cR$ are also in $\cR_{\alpha,x,i}$,
by also using  (\RuleReflexivity) and (\RuleCompatibility)
we conclude $\grounding{s}\rews{\cR_{\alpha,x,i}}\grounding{t}$,
as required.
Induction step: $n>0$.
We have two possibilities:
\begin{enumerate}
\item A sentence $(\RulePropagation)_{f,i}$ is used for some $f\in\Symbols$ and
$i\in\mu(f)$. Then,
$\grounding{s}=f(s_1,\ldots,s_i,\ldots,s_k)$ and $\grounding{t}=f(s_1,\ldots,t_i,\ldots,s_k)$
for some (ground) terms $s_1,\ldots,s_k$ and $t_i$ and $\rewtheoryOf{\cR}\vdash s_i\rew{}t_i$ has been proved in less than $n$ steps.
By the induction hypothesis, $\rewtheoryOf{\cR_{\alpha,x,i}}\vdash s_i\rew{}t_i$ holds as well.
Thus, by using again $(\RulePropagation)_{f,i}$
(which is part of $\rewtheoryOf{\cR_{\alpha,x,i}}$), we conclude
$\grounding{s}\rew{\cR_{\alpha,x,i}}\grounding{t}$.
\item A sentence $(\RuleHornClause)_{\beta}$ has been used for some $\beta\in R$.
If $\beta\neq\alpha$, then $\beta$ is also a rule of $\cR_{\alpha,x,i}$ and by using the induction hypothesis, we conclude $\grounding{s}\rew{\cR_{\alpha,x,i}}\grounding{t}$.
Otherwise, $\beta=\alpha$ and $\grounding{s}=\varsigma(\lhsr)$ and
$\grounding{t}=\varsigma(\rhsr)$ for some substitution $\varsigma$,
and also
(by the induction hypothesis and because $\alpha$ is an O-rule), for all $1\leq j\leq n$,
\begin{eqnarray}
\varsigma(s_j)\rews{\cR_{\alpha,x,i}}\varsigma(t_j)\label{LblInductionHypothesis}
\end{eqnarray}
In $\cR_{\alpha,x,i}$ for $\sigma=\{x\mapsto s_i\}$,
we have $\alpha_{x,i}:\lhsr\to\sigma(\rhsr)\IF \sigma(s_1)\cto t_1, \cdots, \sigma(s_{i-1})\cto
t_{i-1}, \sigma(s_{i+1})\cto t_{i+1}, \cdots,\sigma(s_n)\cto t_n$
instead of $\alpha$.
Note that $\alpha_{x,i}$ contains $n-1$ conditions.
For all $1\leq j\leq n$,
(a) if $x\notin\Var(s_j)$, then $\sigma(s_j)=s_j$ and we have
$\varsigma(s_j)=\varsigma(\sigma(s_j))\rews{}\varsigma(t_j)$; and
(b) if $x\in\Var(s_j)$, then we can write $s_j=C_j[x]_{P_j}$ for some context $C_j$
and set $P_j$ of positions of $x$ in $s_j$ and therefore,
$\sigma(s_j)=C_j[s_i]_{P_j}$.
Since,
by (\ref{LblInductionHypothesis}), we have $\varsigma(s_i)\rews{\cR_{\alpha,x,i}}\varsigma(t_i)=\varsigma(x)$ and also (by (\ref{LblVariableRestrictionInlining}))
$x$ is not frozen in $s_j$ (and it does not occur in $s_i$),
we have
\[\begin{array}{rcl}
\varsigma(\sigma(s_j)) & = & \varsigma(C_j[s_i]_{P_j})\\
 & = & \varsigma(C_j)[\varsigma(s_i)]_{P_j}\\
& \rews{\cR_{\alpha,x,i}} &
\varsigma(C_j)[\varsigma(x)]_{P_j}\\
& = & \varsigma(C_j[x]_{P_j})\\
& =& \varsigma(s_j)\\
& \rews{\cR_{\alpha,x,i}} & \varsigma(t_j)
\end{array}
\]
where the last rewriting sequence is proved
by using (\ref{LblInductionHypothesis}) again,
now for $\varsigma(s_j)$.
Therefore, all conditions $\sigma(s_j)\rew{}t_j$, $1\leq j\leq n$, $j\neq i$
in $\alpha,x,i$ are satisfied by $\varsigma$.
Thus, since $x\notin\Var(\lhsr)$,
we have
$\grounding{s}=\varsigma(\lhsr)\rew{\cR_{\alpha,x,i}}\varsigma(\rhsr)$.
Again,
(a) if $x\notin\Var(\rhsr)$, then $\sigma(\rhsr)=\rhsr$ and we have
$\varsigma(\rhsr)=\grounding{t}$, as required.
Otherwise, if (b) $x\in\Var(\rhsr)$, then $\rhsr=C[x]_P$ for some context $C$ and
set $P$ of positions of $x$ in $\rhsr$, and $\sigma(\rhsr)=C[s_i]_P$.
Since (by (\ref{LblVariableRestrictionInlining})) $x$ is \emph{not} frozen in $\rhsr$ (i.e., $P\subseteq\Pos^\mu(\rhsr)$),
we have
\[\varsigma(\sigma(r))=\varsigma(C[s_i]_{P})
=\varsigma(C)[\varsigma(s_i)]_{P}
\rews{\cR_{\alpha,x,i}}
\varsigma(C)[\varsigma(x)]_{P}
=\varsigma(C[x]_{P})
=\varsigma(\rhsr)=\grounding{t}
\]
Therefore, $\grounding{s}\rews{\cR_{\alpha,x,i}}\grounding{t}$, as required.
\end{enumerate}
Regarding the \emph{second claim} ($s\rew{\cR_{\alpha,x,i}}t$ implies
$s\rew{\cR}t$), we proceed by induction on length $n\geq 0$
of a   \emph{Hilbert-like} proof of $\grounding{s}\rew{}\grounding{t}$
from $\rewtheoryOf{\cR_{\alpha,x,i}}$.
Base: $n=0$.
With $\grounding{s}\rew{}\grounding{t}$
we are using $(\RuleHornClause)_\alpha$ for some unconditional rule
$\alpha:\elhsr\to\erhsr$ such that
$\grounding{s}=\varsigma(\elhsr)$ and $\grounding{t}=\varsigma(\erhsr)$.
The only possibility for an unconditional rule in $\cR_{\alpha,x,i}$ of \emph{not} 
being in $\cR$
is $\alpha_{x,i}$, provided that the conditional part of $\alpha$
consists of a single condition $s\cto x$.
In this case, $\alpha_{x,i}$ is $\lhsr\to\rhsr[s]_P$ where $P$ are the positions of $x$ in $\rhsr$, i.e., $r=r[x]_P$.
Hence, $\grounding{s}=\varsigma(\lhsr)$ and $\grounding{t}=\varsigma(\rhsr[s]_P)$.
Since $x\notin\Var(\lhsr,s)$, we can assume that $\varsigma$ does not instantiate $x$.
Define a substitution $\varsigma'$ as follows:
$\varsigma'(x)=\varsigma(s)$ and
for all $y\in\Variables-\{x\}$, $\varsigma'(y)=\varsigma(y)$.
Note that $\varsigma'(\lhsr)=\varsigma(\lhsr)$ and $\varsigma'(s)=\varsigma(s)$.
Then, the only condition $s\cto x$ in the conditional part of $\alpha$ is trivially satisfied by $\varsigma'$.
Therefore, since $x$ is not frozen in $\rhsr$ (by (\ref{LblVariableRestrictionInlining})), we have
\[\grounding{s}=\varsigma(\lhsr)=\varsigma'(\lhsr)
\rew{\alpha}\varsigma'(\rhsr)
=\varsigma'(\rhsr[x]_P)
=\varsigma'(\rhsr)[\varsigma(s)]_P=\varsigma(\rhsr)[\varsigma(s)]_P=\varsigma(\rhsr[s]_P)
=\grounding{t}
\]
as desired, as the conditional part $s\cto x$ of $\alpha$ is satisfied by $\varsigma'$:
$\varsigma'(s)=\varsigma(s)=\varsigma'(x)$.
Induction step: $n>0$.
We have two possibilities:
\begin{enumerate}
\item A sentence $(\RulePropagation)_{f,i}$ is used for some $f\in\Symbols$ and
$i\in\mu(f)$ . Analogous to the corresponding case of the first claim.
\item A sentence $(\RuleHornClause)_{\beta}$ has been used for some conditional rule $\beta$.
If $\beta\neq\alpha_{x,i}$, then $\beta\in R$; hence $\grounding{s}\rew{\cR}\grounding{t}$.
Otherwise, if $\beta=\alpha_{x,i}$, then $\grounding{s}=\varsigma(\lhsr)\rew{\cR_{\alpha,x,i}}\varsigma(\rhsr[s]_P)=\grounding{t}$ for some substitution $\varsigma$
such that $\varsigma(s_j[s]_{P_j})\rews{\cR_{\alpha,x,i}}\varsigma(t_j)$ holds for all $1\leq j\leq n$, $j\neq i$.
By the induction hypothesis (and repeatedly using $(\RuleCompatibility)$ and $(\RuleReflexivity)$),
$\varsigma(s_j[s]_{P_j})\rews{\cR}\varsigma(t_j)$ holds for all $1\leq j\leq n$, $j\neq i$.
Define $\varsigma'$ as follows: $\varsigma'(y)=\varsigma(y)$ if $y\neq x$, and
$\varsigma'(x)=\varsigma(s)$.
For all $1\leq j\leq n$, $j\neq i$,
\[\varsigma'(s_j)
=\varsigma'(s_j[x]_{P_j})
=\varsigma(s_j)[\varsigma(x)]_{P_j}
=\varsigma(s_j)[s]_{P_j}
=\varsigma(s_j[s]_{P_j})
\rews{\cR}\varsigma(t_j)=\varsigma'(t_j)\]
because $x\notin\Var(s)$ and $x\notin\Var(t_j)$.
We also have
\[\varsigma'(s_i)=\varsigma'(s)=\varsigma(s)=\varsigma'(x)\]
Therefore,
we finally have
\[\grounding{s}=\varsigma'(\lhsr)\rew{\cR}\varsigma'(\rhsr)=\varsigma'(\rhsr[x]_P)=\varsigma(\rhsr)[\varsigma(s)]_p=\varsigma(\rhsr[s]_P)=\grounding{t}\]
as required.
\end{enumerate}

\vspace*{-7mm}
\end{proof}

\end{document}